\documentclass[a4paper,onecolumn,11pt,
accepted=2026-03-09
]{quantumarticle}
\pdfoutput=1
\usepackage[utf8]{inputenc}
\usepackage[english]{babel}
\usepackage[T1]{fontenc}
\usepackage{amsmath}
\usepackage{hyperref}
\usepackage{amsthm}

\usepackage{tikz}
\usepackage{lipsum}
\usepackage{euscript}
\usepackage{epsfig}

\usepackage{amssymb}
\usepackage{amscd}
\usepackage{braket}

\usepackage{epic}
\usepackage{color}

\usepackage{soul}

\usepackage{enumitem}
\usepackage{nicematrix} 
\usepackage[sorting=none,style=phys]{biblatex}
\usepackage{tikz-cd} 
\usepackage{tabularx} 

\numberwithin{equation}{section}
\numberwithin{equation}{subsection}

\theoremstyle{plain}
\newtheorem{theorem}
[equation]
{Theorem}
\newtheorem{lemma}[equation]{Lemma}

\newtheorem{thm}[equation]{Theorem}
\newtheorem{cor}[equation]{Corollary}
\newtheorem{lem}[equation]{Lemma}
\newtheorem{prop}[equation]{Proposition}

\theoremstyle{definition}

\newtheorem{remark}[equation]{Remark}

\newtheorem{ex}[equation]{Example}

\newtheorem{rem}[equation]{Remark}

\numberwithin{equation}{section}
\numberwithin{equation}{subsection}

\DeclareMathOperator{\tr}{{\rm tr}}

\newcommand{ \rk }{ \mbox{rk} }

\def\C{\mathbb C}

\def\R{\mathbb R}

\usepackage{makecell}

\begin{document}

\title{The geometry of the Hermitian matrix space and the Schrieffer--Wolff transformation}

\author{Gerg\H{o} Pint\'er}
\affiliation{Department of Theoretical Physics, Institute of Physics, Budapest University of Technology and Economics, M\H{u}egyetem rkp. 3., H-1111 Budapest, Hungary}
\affiliation{HUN-REN-BME-BCE Quantum Technology Research Group, Műegyetem rkp. 3., H-1111 Budapest, Hungary}

\author{Gy\"orgy Frank}
\affiliation{Department of Theoretical Physics, Institute of Physics, Budapest University of Technology and Economics, M\H{u}egyetem rkp. 3., H-1111 Budapest, Hungary}

\author{D\'aniel Varjas}
\affiliation{Department of Theoretical Physics, Institute of Physics, Budapest University of Technology and Economics, M\H{u}egyetem rkp. 3., H-1111 Budapest, Hungary}
\affiliation{IFW Dresden and W{\"u}rzburg-Dresden Cluster of Excellence ct.qmat, Helmholtzstrasse 20, 01069 Dresden, Germany}
\affiliation{Max Planck Institute for the Physics of Complex Systems, Nöthnitzer Strasse 38, 01187 Dresden, Germany}

\author{Andr\'as P\'alyi}
\affiliation{Department of Theoretical Physics, Institute of Physics, Budapest University of Technology and Economics, M\H{u}egyetem rkp. 3., H-1111 Budapest, Hungary}
\affiliation{HUN-REN-BME-BCE Quantum Technology Research Group, Műegyetem rkp. 3., H-1111 Budapest, Hungary}

\maketitle


\begin{abstract}
    In quantum mechanics, the Schrieffer--Wolff (SW) transformation (also called quasi-degenerate perturbation theory) is known as an approximative method to reduce the dimension of the Hamiltonian. We present a geometric interpretation of the SW transformation: We prove that it induces a local coordinate chart in the space of Hermitian matrices near a $k$-fold degeneracy submanifold. Inspired by this result, we establish a `distance theorem': we show that the standard deviation of $k$ neighboring  eigenvalues of a Hamiltonian equals the distance of this Hamiltonian from the corresponding $k$-fold degeneracy submanifold, divided by $\sqrt{k}$. 
    Furthermore, we investigate one-parameter perturbations of a degenerate Hamiltonian, and prove that the standard deviation and the pairwise differences of the eigenvalues lead to the same order of splitting of the energy eigenvalues,
    which in turn is the same as the order of distancing from the degeneracy submanifold. 
   As applications, we prove the `protection' of Weyl points using the transversality theorem, and infer geometrical properties of certain degeneracy submanifolds based on results from quantum error correction and topological order.
\end{abstract}

\tableofcontents

\section{Introduction}\label{s:intro}

Many physics problems are translated to matrix eigenvalue problems.
For example, the frequency spectrum and the spatial patterns (modes) of small oscillations of a mechanical system are described by the eigenvalues and the eigenvectors of the dynamical matrix, which is a real symmetric matrix. 
Another example, which is the focus of this work, is quantum mechanics, where the stationary states and the energies are given by the eigenvectors and eigenvalues of the Hamiltonian, which is a Hermitian operator, often a finite-dimensional matrix.

In certain cases, a physics problem translated to a matrix eigenvalue problem is treated using perturbative methods. 
One of these is the Schrieffer-Wolff (SW) transformation \cite{SW,luttinger1955motion,BravyiSW} --- also known as quasidegenerate perturbation theory~\cite{Winkler, day2024pymablock} Van Vleck perturbation theory~\cite{vanVleck1929},
or L\"{o}wdin partitioning~\cite{Lowdin_1962} --- which decouples the subspaces of the relevant and irrelevant energy eigenstates through an appropriate unitary transformation.

In this work, motivated by many examples in quantum mechanics (see below), we focus on the case, when the Hamiltonian ($n \times n$ Hermitian matrix) of interest is in the vicinity of an \emph{unperturbed Hamiltonian} $H_0$ whose energy spectrum hosts a $k$-fold degenerate eigenvalue, and we need to determine only those eigenvalues and eigenstates of the perturbed Hamiltonian $H$ that correspond to the degenerate subspace of $H_0$.
In this case, the SW transformation takes the $n\times n$ Hermitian matrices $H_0$ and $H$ as inputs, and outputs a $k\times k$ Hermitian matrix $H_\text{eff}$, the \emph{effective Hamiltonian}, whose eigenvalues and eigenvectors represent the relevant eigenvalues and eigenvectors of $H$. 
This provides a practical, useful computational method: if $k$ is much smaller than $n$, then computing the eigensystem of $H_\text{eff}$ may be much easier than doing the same for $H$.
Furthermore, the SW transformation is typically applied together with an approximate, truncated power expansion of $H_\text{eff}$ in the perturbation parameter(s), i.e., $H_\text{eff}$ is approximated as a low-order polynomial. 

In this work, we provide a geometrical interpretation of the SW transformation, relating the latter to the geometry of the submanifolds of degenerate matrices.
Our work is motivated by concrete examples in physics: band-structure degeneracy points (including Weyl points) and degenerate ground states of quantum spin models (including topologically ordered systems and stabiliser code Hamiltonians). 
At the same time, we provide a purely mathematical description of concepts and relations throughout the paper. 
The geometrical view we develop here builds a new connection between differential geometry and quantum mechanics, and hence enables the application of tools in one domain to deepen the understanding of, or solve problems in, the other domain.

Our first contribution is that we recognize that for any unperturbed degenerate Hamiltonian $H_0$, the SW transformation provides a \emph{canonical analytical local chart} of the manifold of Hermitian matrices, which is aligned with the corresponding degeneracy submanifold.
Furthermore, we identify the effective Hamiltonian as the collection of those coordinates of this chart that lead out of the degeneracy submanifold.

Our second contribution is a `distance theorem', which is  a proportionality relation between 
(i) the Frobenius (a.k.a. Hilbert-Schmidt) distance between a generic Hamiltonian $H$ and a $k$-fold degeneracy submanifold, and 
(ii) the standard deviation of the quasidegenerate eigenvalues of $H$ corresponding to that degeneracy submanifold (sometimes simply called the `energy splitting of the degeneracy').
More precisely, the standard deviation of the quasidegenerate eigenvalues equals the distance divided by $\sqrt{k}$.

Third, we relate the first and second contributions: we show that the norm of the effective Hamiltonian $H_\text{eff}$ obtained from the SW transformation is the same as the distance between $H$ and the degeneracy submanifold.

Fourth, we consider the energy splitting of the degenerate energy eigenvalues of $H_0$ due to a linear perturbation $H(t) = H_0 + t H_1$.
We define the \emph{order of energy splitting} $r \in \{1,2,\dots \} \cup \{\infty\}$ for higher degeneracies $k>2$ in multiple ways, and we prove that those definitions are equivalent to each other.
Moreover, as a consequence of our distance theorem, the order of energy splitting coincides with the order at which the distance of $H(t)$ from the degeneracy submanifold increases (\emph{order of distancing}).
The order of energy splitting also generalises from linear to analytical perturbations.

One application of our results is an alternative explanation of the `protection'of Weyl points, which are  degeneracy points appearing, e.g., in the electronic band structure of crystalline materials \cite{Armitage2018}.
Making use of the relations developed here, we show that the protection of Weyl points against perturbations is analogous to the protection of the crossing point of two lines drawn on a paper sheet, and we formalize this analogy by a common underlying theorem, namely the transversality theorem. 

We also demonstrate the applicability of our results at the intersection of differential geometry, condensed matter physics, and quantum information science.
In the condensed-matter and quantum-information domains, quantum systems with robust degeneracies are often desirable.
In this particular context, a robust degeneracy is such that physically relevant perturbations break it with a high order $r$ of energy splitting; $r = 1$ (\emph{linear energy splitting}) is not robust, $r = 2$ (\emph{quadratic energy splitting}) is already robust, and the higher the $r$ the more robust the degeneracy.

Examples include the (1) the ground-state degeneracy of the quantum Ising model at zero field, which is robust against a small tranverse field \cite{Igloi,Damski,Juhasz}, (2) the zero-energy degeneracy of the edge modes of the fully dimerized Su-Schrieffer-Heeger model \cite{SSH,Asboth}, which is robust against chiral-symmetric local perturbations, (3) the zero-energy degeneracy of the Majorana edge modes of the Kitaev chain \cite{Kitaev2001}, which is robust against local perturbations respecting the particle-hole symmetry of the Bogoliubov-de Gennes formalism, and (4) the ground-state degeneracy of the Toric Code \cite{KitaevToricCode}, which is robust against 1-local perturbations (Zeeman fields). 
In all four cases, the considered perturbations cause an energy splitting whose order is proportional to the linear size of the system.
Further examples are (5) stabiliser code Hamiltonians, which are Hamiltonians generalising (4) and corresponding to quantum error correction codes \cite{GottesmanPhD}. Stabiliser code Hamiltonians have degenerate ground states that are robust against 1-local perturbations, exhibiting an order of energy splitting set by the code distance.

As we have described above (fourth contribution), an important finding of this work is that the order of energy splitting of a degeneracy due to a linear perturbation is the order of distancing of the perturbed Hamiltonian from the degeneracy submanifold. 
As an application of this result, we show that the key characteristics of a stabiliser code Hamiltonian, such as size, ground-state degeneracy, and code distance, can be translated into geometric information about the corresponding degeneracy submanifold in the space of Hamiltonians. 

The rest of this paper is structured as follows.
Sec.~\ref{s:pre} summarizes relevant preliminaries.
Sec.~\ref{sec:summary} provides a summary of the key concepts and statements of this work, including physics-related examples as well as applications. 
Sec.~\ref{sec:proofs} contains the proofs and further details of the results, with subsections in one-to-one correspondence with those of Sec.~\ref{sec:summary}. 
The Appendix collects further notes on the SW transformation.

\section{Preliminaries, notations, and conventions}\label{s:pre} 

Here we summarize some well known properties of the space of Hermitian matrices. We aim to fix our notation and to provide the setting of our paper.

\subsection{The space of the Hermitian matrices}\label{ss:spaceherm} The Hermitian matrices (of size $n \times n$) are the complex matrices $H$ with $H=H^{\dagger}$. Their set $\mbox{Herm}(n)$ is a real vector space, it is a subspace of  $\C^{n \times n}$ of all $ n \times n$ complex matrices. Although $\C^{n \times n}$ is a complex vector space (of dimension $n^2$), $\mbox{Herm}(n)$ inherits only real vector space structure, because it is not closed under the multiplication by the imaginary unit $i$. 

The dimension of $\mbox{Herm}(n)$ over $\R$ is $n^2$. 
We fix a basis of $\mbox{Herm}(n)$, which we will refer to as `canonical basis', which enables to identify $\mbox{Herm}(n)$ with $\R^{n^2}$. This basis is formed by three families of matrices, the real off-diagonal parts ($(n^2-n)/2$ matrices), imaginary parts ($(n^2-n)/2$ matrices) and the diagonal parts ($n$ matrices):
\begin{equation}\label{eq:bas1}
    \sigma_{ab}^{(\textup{real})}= 
    \frac{1}{\sqrt{2}} (e_b \cdot e_a^{\dagger} + e_a \cdot e_b^{\dagger})=
    \frac{1}{\sqrt{2}}\begin{pNiceMatrix}[last-row=6,last-col=6]
        &\vdots&&\vdots&&\\
        \dots&0&\dots&1&\dots&a\\
        &\vdots&&\vdots&&\\
        \dots&1&\dots&0&\dots&b\\
        &\vdots&&\vdots&&\\
        &a&&b&&
    \end{pNiceMatrix}\
    \mbox{\hspace{1cm}(for $a < b$)},
\end{equation}
\begin{equation}\label{eq:bas2}
    \sigma_{ab}^{(\textup{im})}= 
    \frac{i}{\sqrt{2}} (e_b \cdot e_a^{\dagger} - e_a \cdot e_b^{\dagger})=
    \frac{1}{\sqrt{2}}\begin{pNiceMatrix}[last-row=6,last-col=6]
        &\vdots&&\vdots&&\\
        \dots&0&\dots&-i&\dots&a\\
        &\vdots&&\vdots&&\\
        \dots&i&\dots&0&\dots&b\\
        &\vdots&&\vdots&&\\
        &a&&b&&
    \end{pNiceMatrix}\
    \ \mbox{\hspace{1cm}(for $a < b$)},
\end{equation}
\begin{equation}\label{eq:bas3}
    \sigma_{aa}^{(\textup{diag})}= 
     e_a \cdot e_a^{\dagger} =
     \begin{pNiceMatrix}[last-row=6,last-col=6]
        0&&&&&\\
        &\ddots&&&&\\
        &&1&&&a\\
        &&&\ddots&&\\
        &&&&0&\\
        &&a&&&
    \end{pNiceMatrix},
\end{equation}
where $e_1, \dots, e_n$ is the standard basis of $\C^n$. For any other unitary basis $u_1, \dots, u_n$ of $\C^n$ (where $u_a=Ue_a$ is the $a$-th column of a unitary matrix $U \in U(n)$), there is an associated basis of $\mbox{Herm}(n)$ (over $\R$), obtained from Equations~\eqref{eq:bas1}, \eqref{eq:bas2}, \eqref{eq:bas3} by replacing $e_a$ with $u_a$, resulting $\sigma^U_{ab}=U\sigma_{ab} U^{\dagger}$.

For brevity, we denote the canonical basis of $\mathrm{Herm}(n)$ by $C=(c_1, \dots, c_{n^2})$, where $c_a$-s are the above defined $\sigma$ matrices in an appropriate order where the first $k^2$ elements generate the upper left $k\times k$ block. For example, 
\begin{equation}
    c_1=\sigma_{11}^{(\textup{diag})}, \ 
    c_2=\sigma_{12}^{(\textup{real})}, \ 
    c_3=\sigma_{12}^{(\textup{im})}, \ 
    c_4=\sigma_{22}^{(\textup{diag})}, \
    c_5=\sigma_{13}^{(\textup{real})}, \
    c_6=\sigma_{13}^{(\textup{im})}, \ 
    \dots
\end{equation}
Similarly, we denote the basis associated with a unitary matrix $U \in U(n)$ by $C^U=(c^U_1, \dots, c^U_{n^2})$, with $c^U_a=Uc_aU^{\dagger}$, 
 therefore, $c^U_1, \dots, c^U_{k^2}$ span the subspace of $\mbox{Herm}(n)$ consisting of the  matrices acting on the subspace of $\C^n$ spanned by $u_1, \dots, u_k$.

For $k < n$, $\mbox{Herm}(k)$ is embedded in $\mbox{Herm}(n)$ as the subspace that consists of the matrices having zero entries outside the upper left $k \times k$ block, that is, the subspace spanned by $c_1, \dots, c_{k^2}$. We often identify $\mbox{Herm}(k)$ with this subspace of $\mbox{Herm}(n)$. 

The complex vector space $\C^{n \times n}$ is endowed with a Hermitian inner product (Frobenius inner product) defined as \begin{equation}\label{eq:compinn}
    \langle M, N \rangle =\tr (M^{\dagger} \cdot N)=\sum_{a=1}^{n} \sum_{b=1}^{n} M^\ast_{ab} N_{ab},
\end{equation}
where $z^\ast$ denotes the complex conjugate of $z \in \C$. That is, the inner product of two matrices agrees with the inner product of the vectors formed by the entries.

Although, this inner product takes complex values, its restriction to $\mbox{Herm}(n)$ is real valued because every summand on the right side of Equation~\eqref{eq:compinn} has its complex conjugate too. Therefore, the Frobenius inner product on $\mbox{Herm}(n)$ simplifies as \mbox{$\langle H, K \rangle=\tr(H \cdot K)\in\mathbb{R}$}, making $\mbox{Herm}(n)$ a Euclidean space. The basis $C^U$ of $\mathrm{Herm}(n)$ associated to $U \in U(n)$ (in particular, the canonical basis $C$) is orthonormal with respect to the Frobenius inner product. If the coordinates of $H$ and $K$ in the basis $C^U$ are $h^U_a$ and $k^U_b$, respectively, then
\begin{equation}
\langle H, K \rangle=\sum_{a=1}^{n^2} h^U_a k^U_a .  
\end{equation}
The induced Frobenius norm of a Hermitian matrix is $\| H \|= \sqrt{\tr(H^2)}$. 
The distance of two matrices $H, G \in \mbox{Herm}(n)$ induced by the Frobenius norm is denoted by $d(H,G)=\| H-G\|$. This distance induces a topology on $\mbox{Herm}(n)$. An open neighborhood, or simply a neighborhood of $H$ in $\mbox{Herm}(n)$ is an open subset $A \subset \mbox{Herm}(n)$ which contains an open ball around $H$, that is, there is a radius $0<r$ such that $d(H, G)<r$ implies that $G \in A$.

\begin{lem}\label{le:uniact}
    The scalar product is invariant under the conjugation by unitary matrices. That is, $\langle UHU^{\dagger},  UKU^{\dagger} \rangle=\langle H, K \rangle$ holds for every $H, K \in \mbox{Herm}(n)$ and $U \in U(n)$.
\end{lem}

\begin{proof}
$\langle UHU^{\dagger},  UKU^{\dagger} \rangle=\tr (UH^{\dagger} U^{\dagger} UKU^{\dagger})= \tr (U^{\dagger} U H^{\dagger} K)=\tr (H^{\dagger} K)=\langle H, K \rangle $.
\end{proof}

Let $\mbox{Herm}_0(n) \subset \mbox{Herm}(n)$ be the subspace consisting of the matrices with zero trace. The dimension of $\mbox{Herm}_0(n)$ over $\R$ is $n^2-1$. A basis  $\breve{c}_1, \dots, \breve{c}_{n^2-1}$ is formed by the $\sigma_{ij}^{\textup{(real)}}$ and $\sigma_{ij}^{\textup{(im)}}$ matrices, extended with an orthonormal basis of the diagonal matrices with zero trace.
This latter basis of traceless diagonal matrices replaces the matrices $\sigma_{ii}^{\textup{(diag)}}$, which are not contained in $\mbox{Herm}_0(n)$. For example, for $n=2$ one can choose the normalized Pauli matrices as a basis of $\mbox{Herm}_0(2)$:
\begin{eqnarray}
    \frac{1}{\sqrt{2}} \sigma_x&=& 
    \frac{1}{\sqrt{2}}
    \begin{pmatrix}
        0 & 1 \\
        1 & 0
    \end{pmatrix}=\sigma_{12}^{\textup{(real)}},\\
     \frac{1}{\sqrt{2}} \sigma_y&=&\frac{1}{\sqrt{2}}
    \begin{pmatrix}
        0 & -i \\
        i & 0
    \end{pmatrix}=\sigma_{12}^{\textup{(im)}},\\
        \frac{1}{\sqrt{2}}\sigma_z&=&\frac{1}{\sqrt{2}}\begin{pmatrix}
        1 & 0 \\
        0 & -1
    \end{pmatrix}=\frac{1}{\sqrt{2}}\left(\sigma_{11}^{\textup{(diag)}}-\sigma_{22}^{\textup{(diag)}}\right).
    \end{eqnarray}
    
Note that, in $\mbox{Herm}_0(n)$ with $n>2$, the differences $\sigma_{a,a}^{\textup{(diag)}}-\sigma_{a+1,a+1}^{\textup{(diag)}}$ do form a basis of the traceless diagonal matrices, however, they are not orthogonal to each other. In order to obtain an orthonormal basis, one needs to perform the Gram--Schmidt process. For example, for $n=3$ the normalized diagonal Gell--Mann matrices read
\begin{eqnarray}
    \frac{1}{\sqrt{2}} \lambda_3&=& 
    \frac{1}{\sqrt{2}}
    \begin{pmatrix}
        1 & 0 & 0\\
        0 & -1 & 0\\
        0 & 0 & 0
    \end{pmatrix}=\frac{1}{\sqrt{2}}\left(\sigma_{11}^{\textup{(diag)}}-\sigma_{22}^{\textup{(diag)}}\right),\\
     \frac{1}{\sqrt{2}} \lambda_8&=& 
    \frac{1}{\sqrt{6}}
    \begin{pmatrix}
        1 & 0 & 0\\
        0 & 1 & 0\\
        0 & 0 & -2
    \end{pmatrix}=\frac{1}{\sqrt{6}}\left(\sigma_{11}^{\textup{(diag)}}+\sigma_{22}^{\textup{(diag)}}-2\sigma_{33}^{\textup{(diag)}}\right).
    \end{eqnarray}
    (The other Gell-Mann matrices are the same as the off-diagonal matrices defined in Equation~\eqref{eq:bas1} and \eqref{eq:bas2}, up to $\sqrt{2}$ factor.)

\begin{rem}\label{re:op2norm} We note that in quantum mechanics, the operator 2-norm $\| H \|_2$  is more frequently used than the Frobenius norm, since it refers directly to the 2-norm of the state vectors $\Psi \in \C^n$. Indeed, it is defined as $\| H \|_2=\mbox{sup}\{\|H \cdot \Psi\| \ | \ \| \Psi \|=1  \}$, or equivalently (for Hermitian matrices) $\| H \|_2=\max \{ |\lambda_i|\}$, where $\lambda_i$ are the eigenvalues of $H$. However, the Frobenius norm is more suitable to study the geometry of $\mbox{Herm}(n)$: Importantly, the Frobenius norm is induced by an inner product, which allows to study angles and distances as well. In contrast, the operator 2-norm is not induced by an inner product, as it does not satisfy the parallelogram law. 
By Lemma~\ref{le:uniact} the Frobenius norm is $\| H\| = \sqrt{\sum_{i=1}^n \lambda_i^2}$, so we have the relation $ \| H \|_2 \leq \| H \| \leq \sqrt{n} \cdot \| H \|_2$.
This shows that the operator 2-norm and the Frobenius norm on $\mbox{Herm}(n)$ are equivalent, hence they induce the same topology. (Indeed, every norm on a finite dimensional vector space is equivalent.) 

In Ref.~\cite{BravyiSW}, the operator 2-norm is used, and we also use it to compare our results to those of \cite{BravyiSW} in the \hyperref[ss:Bravyi]{Appendix}.
\end{rem}

\subsection{Degeneracy of the eigenvalues}\label{ss:preldegen} Hermitian matrices have real eigenvalues and they can be diagonalized by a unitary basis transformation.
Furthermore, for each $H \in  \mbox{Herm}(n)$, one can choose a unitary matrix $U \in U(n)$ such that 
\begin{equation}\label{eq:dia} \Lambda = U^{-1} \cdot H \cdot U \end{equation}
is diagonal with the eigenvalues in increasing order $\lambda_1 \leq \lambda_2 \leq \dots \leq \lambda_n$. The ordered eigenvalues are continuous functions $\lambda_i: \mbox{Herm}(n) \to \R$, as it can be proven e.g. using Weyl's inequality~\cite{Weyl1912}.
However, the diagonalizing matrix $U$ is not unique; each column, in fact, can be independently multiplied by an arbitrary phase factor. Moreover, in the case of degenerate eigenvalues $\lambda_{i+1} =\dots = \lambda_{i+k} $, the corresponding columns of $U$ form a unitary basis of the corresponding $k$-dimensional eigenspace. This basis can be transformed by any unitary action of $U(k)$ to obtain a different matrix $U'$ that still diagonalizes $H$.

The \emph{degeneracy set} $\Sigma \subset \mbox{Herm}(n)$ is the set of matrices with at least two coinciding eigenvalues. 
This degeneracy set $\Sigma$ is a subvariety. Indeed, it can be defined
by one polynomial equation, namely, $\Sigma$ is the zero locus of the discriminant of the characteristic polynomial.
By the Neumann--Wigner theorem \cite{Neumann1929} the codimension of $\Sigma$ is 3.
For convenience, we summarize the original proof in Table~\ref{tab:neumann}. Our work provides an alternative proof, as a byproduct of the SW chart, see Corollary~\ref{co:locchart}.

\begin{table}[h]
    \begin{center}
	\begin{tabular}{|c|cc|cc|}
	    \hline
		&\multicolumn{2}{c|}{Non-degenerate matrices}&\multicolumn{2}{c|}{$k$-fold ground state degenerate matrices}\\ &\thead{Diagonalization\\structure}&\thead{Number of\\parameters}&\thead{Diagonalization\\structure}&\thead{Number of\\ parameters}\\
		\hline
		$\Lambda$ & $\lambda_1<\dots<\lambda_n$ & $n$ & $\lambda_1=\dots=\lambda_k<\dots<\lambda_n$ & $n-k+1$ \\
		\hline
  $U$ & \small{$[U]\in\frac{\textstyle\text{U}(n)}{\underbrace{\text{U}(1)\times \dots \times \text{U}(1)}_{\textstyle n\text{ times}}}$} &$n^2-n$&\small{$[U]\in\frac{\textstyle\text{U}(n)}{\textstyle\text{U}(k)\times \underbrace{\text{U}(1)\times \dots \times \text{U}(1)}_{\textstyle n-k\text{ times}}}$} & $\begin{array}{c} n^2-k^2\;\;\;\\\;\;\;-(n-k)\end{array}$\\
  \hline
  $H$&&$n^2$&&$n^2-\left(k^2-1\right)$\\
		\hline
	\end{tabular}
	\caption{Summary of the proof of the Neumann--Wigner theorem for $k$-fold (ground state) degeneracy. Consider a Hermitian matrix in form $H=U \Lambda U^{-1}$, where $\Lambda$ is diagonal, containing the ordered eigenvalues, and the columns of the unitary matrix $U \in U(n)$ are eigenvectors of $H$. The dimension of $U(n)$ is $n^2$, since every unitary matrix sufficiently close to the identity can be written as $e^{iK}$ with $K \in \mbox{Herm}(n)$ (in other words, the Lie algebra of the Lie group $U(n)$ is the space $\mathfrak{u}(n)=i \cdot \mbox{Herm}(n)$ of the anti-Hermitian matrices). If we construct a non-degenerate $H$,  every column of $U$ can be modified by a $U(1)$ action providing the same matrix $H$. Hence the choice of $U$ has $n^2-n$ free parameters, which together with the $n$ parameters of $\Lambda$ verifies that the dimension of the non-degenerate matrices (i.e., the generic part of $\mbox{Herm}(n)$) is $n^2$. If we construct a $k$-fold degenerate $H$ ($ \lambda_{1}  = \dots = \lambda_{k} < \lambda_{k+1}$), then the eigenvectors $u_{1}, \dots, u_{k}$ spanning the degenerate eigenspace can be modified by a $U(k)$ action. Hence the choice of $U$ has $n^2-k^2-(n-k)$ free parameters, which together with the $n-k+1$ parameters of $\Lambda$ verifies that the dimension of the $k$-fold degeneracy is $n^2-k^2+1$, hence, its codimension is $k^2-1$ independently of $n$. In particular, the codimension of the two-fold degeneracy (i.e., the smooth part of $\Sigma$)  is 3. The same proof shows that the codimension of a general stratum $ \Sigma_{\kappa}$ is $\sum_{i=1}^l \left(k_i^2-1 \right)$, where $\kappa=(k_1, \dots, k_l)$. \label{tab:neumann}}
     \end{center}
\end{table}

The set $\Sigma$ can be decomposed into the disjoint union of different \emph{strata} based on the type of the degeneracy, i.e., which eigenvalues coincide \cite{ArnoldSelMath1995}. Each stratum can be labeled by an ordered partition of $n$ in the following  way. 
Let $\kappa=(k_1, \dots, k_l)$ be a sequence of integers $1 \leq k_i \leq n$ of length $1 \leq l < n$ such that $n=\sum_{i=1}^l k_i$, defining $\kappa$ as an ordered partition of $n$. 
The number of the ordered partitions of $n$ is $2^{n-1}$, (including the case with $l=n$). The associated stratum $\Sigma_{\kappa}$ of $\Sigma $ consists of the matrices with coinciding eigenvalues $\lambda_1 = \dots = \lambda_{k_1} <  \lambda_{k_1 + 1} = \dots = \lambda_{k_1+k_2} < \dots$.

The partition with $l=n$ and $k_i=1$ marks the complement of $\Sigma$, the set of non-degenerate matrices. Each stratum $\Sigma_{\kappa}$ is a (not closed) smooth submanifold of $\mbox{Herm}(n)$, whose dimension is $\dim \Sigma_{\kappa} =n^2-\sum_{i=1}^l \left(k_i^2-1 \right)$ by the Neumann--Wigner theorem \cite{Neumann1929}, cf. Table~\ref{tab:neumann}. The closure $\mbox{cl}(\Sigma_{\kappa})$ contains the higher degeneracies. The smooth points of $\Sigma$ are the strictly two-fold degenerate matrices, that is, exactly 2 eigenvalues coincide and all the others are different. The set $\Sigma$ is singular at all other points, see Figure~\ref{fig:strat4}. We refer to \cite{ArnoldSelMath1995, vassiliev2014spaces} for details. 

\begin{figure}
	\begin{center}
		\includegraphics[width=0.8\columnwidth]{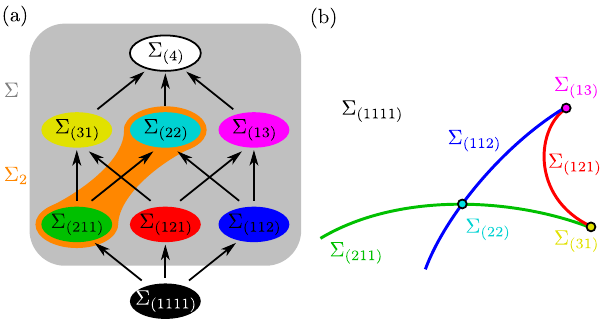}
	\end{center}
	\caption{The stratification of $\mbox{Herm}(4)$, the space of $4\times 4$ Hermitian matrices. (a) The strata are labeled with the ordered partitions of 4, whose number is $2^3=8$. The non-degenerate matrices form an open and dense subset $\Sigma_{(1111)}$ in $\mbox{Herm}(4)$. Its complement is the degeneracy set $\Sigma$ (grey), the set of matrices with at least two coinciding eigenvalues. For example, $\lambda_1=\lambda_2<\lambda_3 < \lambda_4$ in $\Sigma_{(211)}$, and $\lambda_1=\lambda_2 < \lambda_3=\lambda_4$ in $\Sigma_{(22)}$. The two-fold ground state degeneracy set $\Sigma_2$ (orange) consists of the matrices with $\lambda_1=\lambda_2<\lambda_3$, that is, $\Sigma_2= \Sigma_{(2,2)} \cup \Sigma_{(2,1,1)}$. (b)
$\Sigma$ is a subvariety of $\mbox{Herm}(4)$. Although, each stratum is a (non-closed) smooth submanifold of $\mbox{Herm}(4)$, however, $\Sigma$ is not a smooth manifold. The smooth points of $\Sigma$  are the exactly two-fold  degenerate matrices in $\Sigma_{(211)} \cup \Sigma_{(121)} \cup \Sigma_{(112)}$. The two-fold ground state degeneracy set $\Sigma_2$ is a smooth submanifold in $\mbox{Herm}(4)$.
Its generic points belong to $\Sigma_{(2 1 1)}$. At the points   of  $\Sigma_{(22)}$, $\Sigma_2$ is smooth, but $\Sigma$  is not smooth. Indeed, $\Sigma_{(2,2)}$ is the (transverse) intersection of the smooth branches $\Sigma_2$ and $\Sigma_{(1,1,2)} \cup \Sigma_{(2,2)}$. 	\label{fig:strat4}}
\end{figure}

In this article we restrict our study to the subset $\Sigma_k $ of $\Sigma$ consisting of $k$-fold ground-state degenerate matrices, that is, matrices with $\lambda_1=\lambda_2= \dots = \lambda_k < \lambda_{k+1}$. Since higher eigenvalues can coincide with each other, $\Sigma_k$ is the union of all strata $\Sigma_{\kappa}$  corresponding to $\kappa=(k_1, \dots, k_l)$ with $k_1=k$. The generic points of $\Sigma_k$ belong to the stratum corresponding to $(k, 1, \dots, 1)$.

The set $\Sigma_k$ is a smooth, not closed submanifold in  $\mbox{Herm}(n)$ of codimension $k^2-1$, that is, $\dim(\Sigma_k)=n^2-(k^2-1)$, cf. Figure~\ref{fig:strat4}. This follows from the Neumann-Wigner theorem, see Table~\ref{tab:neumann} for a sketch of the proof.

In addition to $\Sigma_k$, its closure $\mbox{cl}(\Sigma_k)$ contains the higher degeneracies as well, it precisely decomposes as the disjoint union $\mbox{cl}(\Sigma_k)=\Sigma_k \cup \Sigma_{k+1} \cup \dots \cup \Sigma_{n}$. The proof is the following. If an infinite  sequence $H_1, H_2, \dots$ of elements $H_i \in \Sigma_k$ is convergent in $\mbox{Herm}(n)$, then its limit $\lim_{i \to \infty} H_i$ is in $\mbox{cl}(\Sigma_k)$ by the definition of the closure. As a consequence of the continuity of the eigenvalues, the lowest $k$ eigenvalues of $\lim_{i \to \infty} H_i$ are degenerate, and it might happen that one or more other eigenvalues also converge to them. 

Our results can be naturally adapted to arbitrary $k$-fold degeneracy, not necessarily ground state, i.e., the set of matrices with $\lambda_i < \lambda_{i+1}= \dots = \lambda_{i+k} < \lambda_{i+k+1}$.

We introduce one more notation. Let 
$\Sigma_{[k, k+1]} \subset \Sigma$ denote the set of matrices $H \in \mbox{Herm}(n)$ with $\lambda_k = \lambda_{k+1}$, that is, the closure of the two-fold degeneracy stratum corresponding to $\lambda_k=\lambda_{k+1}$. For example, in $\mbox{Herm}(4)$ we have $\Sigma_{[2,3]}=\Sigma_{(1,2,1)} \cup \Sigma_{(1,3)} \cup \Sigma_{(3,1)} \cup \Sigma_{(4)}$, its complement is $\Sigma_{(1,1,1,1)} \cup \Sigma_{(2, 1,1)}  \cup \Sigma_{(1,1,2)}\cup \Sigma_{(2,2)}$, cf. Figure~\ref{fig:strat4}.

\section{Summary of results}
\label{sec:summary}

Here we summarize the results of this paper. The proofs and further details can be found in Section~\ref{sec:proofs}, each one in the corresponding subsection with the same title.

\subsection{The Schrieffer--Wolff transformation induces a local chart}

We consider matrices $H \in \mbox{Herm}(n)$ in a sufficiently small neighborhood $\mathcal{V}_0$ of a fixed $k$-fold ground-state degenerate matrix $H_0 \in \Sigma_k$. For simplicity, we assume that $H_0$ is diagonal in the canonical basis of $\C^n$ with increasing order of the eigenvalues. For a more general setting (non-diagonal or non-degenerate $H_0$) see the \hyperref[ss:Bravyi]{Appendix}.

As we mentioned before, the choice of the unitary matrix $U \in U(n)$ in the diagonalization \eqref{eq:dia} of $H$ is not unique. Furthermore, $U$ cannot be chosen to depend continuously on $H$ in any (open) neighborhood of $H_0$ in $\mbox{Herm}(n)$. This is because the eigenvectors corresponding to degenerate eigenvalues cannot be chosen continuously\footnote{The eigenspaces form a complex line bundle over the non-degenerate matrices, which is nontrivial. For example, it is always possible to find a small 2-sphere $S^2$ around $H_0$ in $\mbox{Herm}(n) \setminus \Sigma$, such that the first Chern number of the lowest eigenstate is non-zero, showing the impossibility of a continuous choice of  eigenvectors.}.
Instead, what happens if our objective is to attain only a block diagonal structure
\begin{equation}\label{eq:elsofelb}  \widetilde{B}=U^{-1}  H U,\end{equation}
where $\widetilde{B}$ consists of a $k \times k$ and a $(n-k) \times (n-k)$ block? 
It turns out that such  a family of unitary matrices $U \in U(n)$ can be chosen not only as a continuous, but an analytical function of $H$ as well.
Moreover, if $U$ is assumed to be of the form $U=e^{iS}$ with $S$ a block off-diagonal Hermitian matrix, the decomposition \eqref{eq:elsofelb} is essentially unique. It is formulated by our first statement as follows.

\begin{thm}[Exact SW decomposition, cf. Figure~\ref{fig:pikto}]\label{th:sw}
    Fix a diagonal matrix $H_0 \in \Sigma_k$ with increasing order of its eigenvalues along its diagonal. Then, there are neighborhoods $\mathcal{V}_0, \mathcal{W}_0 \subset \mbox{Herm}(n)$ of $H_0$  and a neighborhood $  \mathcal{X}_0 \subset \mbox{Herm}(n)$ of 0  such that for every $H \in \mathcal{V}_0$ there is a  unique decomposition
\begin{equation}\label{eq:koo}
H=e^{iS} \cdot \widetilde{B} \cdot e^{-iS},
\end{equation}
where $\widetilde{B}$ and $S$ are $n \times n $ Hermitian matrices with the following special properties: 
\begin{enumerate}
	\item $ \widetilde{B} \in \mathcal{W}_0$ is a block diagonal matrix with $k \times k$ and $(n-k) \times (n-k)$ blocks.
	\item $S \in \mathcal{X}_0$ is an off-block matrix, i.e., its $k \times k$ and $(n-k) \times (n-k)$ blocks along the diagonal are zero. Non-zero entries are in the $k\times (n-k)$ and $(n-k)\times k$ off-diagonal blocks. 
 \item The first $k$ columns of $e^{iS}$ span the sum of the eigenspaces of $H$ corresponding to the lowest $k$ eigenvalues.
\end{enumerate}

Furthermore, the dependence of $S$ and $\widetilde{B}$ on $H \in \mathcal{V}_0$ is (real) analytic.  
\end{thm}

Actually, the theorem gives an exact formulation of the well-known method called SW transformation \cite{luttinger1955motion, SW, BravyiSW, Winkler, day2024pymablock}, by breaking up $\widetilde{B}$ into parts according to the blocks as follows, cf. Figure~\ref{fig:pikto}:
\begin{equation}\label{eq:koo}
H=e^{iS} \cdot (H_0+B+T+H_{\textup{eff}}) \cdot e^{-iS},
\end{equation}
where $T$, $B$, $H_{\textup{eff}}$ and $S$ are $n \times n $ Hermitian matrices with the following special properties: 
\begin{enumerate}
	\item $ B$ is block diagonal matrix and it has non-zero entries only in the $(n-k) \times (n-k)$ bottom right block.
 \item $T$ is a scalar matrix in the $k \times k$ upper left block and all the other entries are zero.
 \item The traceless \emph{effective Hamiltonian}  $H_{\textup{eff}}$ has a traceless $k \times k$ block and all the other entries are zero.
	\item $S $ is an off-block matrix, see above (point (2) in Theorem~\ref{th:sw}). 
\end{enumerate}
Assuming that $(S, T, B, H_{\textup{eff}})$ is in a sufficiently small neighborhood of $(0,0,0,0)$, then the decomposition is unique, and 
the dependence of $S$, $T$, $B$ and $H_{\textup{eff}}$  on $H \in \mathcal{V}_0$ is (real) analytic, by Theorem~\ref{th:sw}.  Note that in the literature the effective Hamiltonian  is usually defined together with its trace, i.e., $H_{\textup{eff}}+T $.

The proof of Theorem~\ref{th:sw}, and hence the unique decomposition \eqref{eq:koo} is based on the analytic inverse function theorem, see  Section~\ref{ss.prsw}.
The unitary transformation $e^{iS}$ is described in \cite{BravyiSW}  as a `direct rotation' between the eigenspaces of $H$ and $H_0$ corresponding to the lowest $k$ eigenvalues. In the \hyperref[ss:Bravyi]{Appendix} we clarify the relation of \cite{BravyiSW} with Theorem~\ref{th:sw}, and we adapt the results of \cite{BravyiSW} to specify the validity range of decomposition \eqref{eq:koo}. We also clarify the relation of $e^{iS}$ with the parallel transport in Appendix~\ref{app:para}.

\begin{figure}
	\begin{center}		\includegraphics[width=0.5\columnwidth]{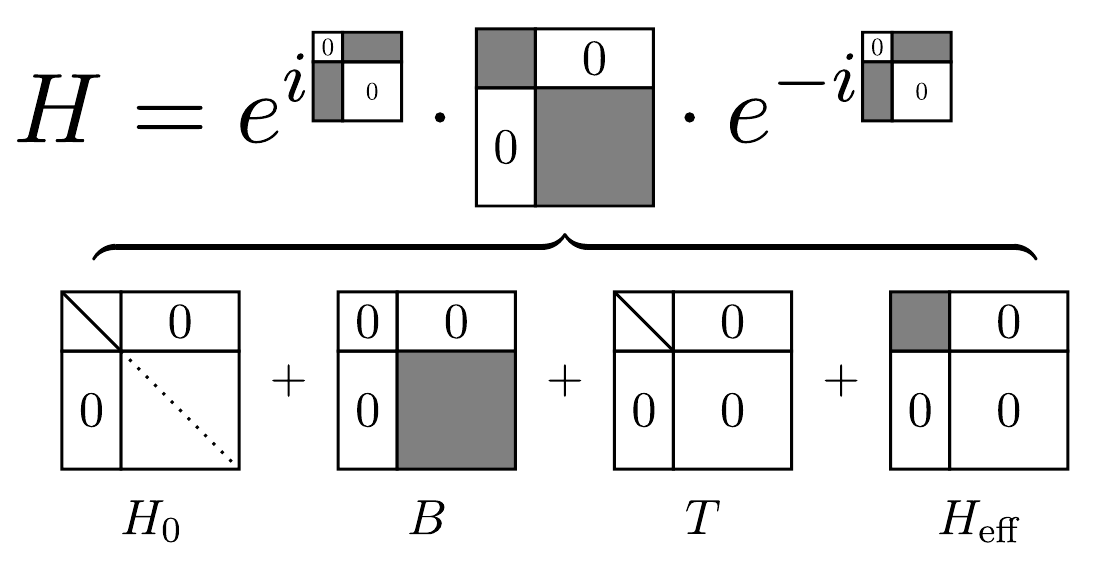}
	\end{center}
	\caption{The exact SW decomposition \eqref{eq:koo}. By Theorem~\ref{th:sw} the block diagonalization of a matrix $H$ close enough to $H_0$ is unique, if the unitary matrix has the special form $e^{iS}$, where $S$ is off-block, and both $S$ and the block diagonal matrix are sufficiently close to 0. Decomposing the block diagonal matrix into parts gives the `non-interesting block' $B$, the scalar matrix $T$ in the `interesting block' carrying the trace of this block, and the traceless effective Hamiltonian $H_{\textup{eff}}$. 
	\label{fig:pikto}}
\end{figure}

By the convention introduced in Section~\ref{ss:spaceherm}, $H_{\textup{eff}}$ can be considered as a traceless $k \times k$ Hermitian matrix, that is, $H_{\textup{eff}} \in \mbox{Herm}_0(k) \subset \mbox{Herm}(n)$. Let $y_1, \dots, y_{k^2-1}$ denote its coordinates in the orthonormal basis $\breve{c}_1, \dots, \breve{c}_{k^2-1}$ of $ \mbox{Herm}_0(k)$ introduced in Section~\ref{s:pre}. 
For example for two-fold degeneracy ($H_0 \in \Sigma_2$), the effective Hamiltonian is expressed as 
\begin{equation}
    H_{\textup{eff}}=
    y_1 \breve{c}_1+y_2 \breve{c}_2+y_3 \breve{c}_3
    =\frac{1}{\sqrt{2}}(y_1 \sigma_x + y_2 \sigma_y + y_3 \sigma_z).
\end{equation}
The number of the (possibly) nonzero coordinates of the matrices $B$, $T$ and $S$ in the canonical basis of $\mbox{Herm}(n)$ is
\begin{enumerate}
    \item $B$: \ $(n-k)^2$,
    \item $T$: \ $1$,
    \item $S$: \ $2 \cdot k \cdot (n-k)$.
\end{enumerate}
Together these are $n^2-k^2+1$ coordinates, let $x_1, \dots, x_{n^2-k^2+1}$ denote them. 

\begin{cor}[SW decomposition induces a local chart, cf. Figure~\ref{fig:chart}]\label{co:locchart} \begin{enumerate}[label=(\alph*)]
    \item The map $\varphi: \mathcal{V}_0 \to \R^{n^2}$ with coordinate functions $x_i$ and $y_j$ is a local chart on $\mbox{Herm}(n)$ around $H_0$, and $\varphi$ is analytic in the canonical coordinates on $\mbox{Herm}(n)$.
    \item In this local chart the set $\Sigma_k$ of $k$-fold ground state degenerate matrices is the common zero locus of the coordinates $y_j$, that is,
\begin{equation}\label{eq:locchart}
    \Sigma_k \cap \mathcal{V}_0 = \{ H \in \mathcal{V}_0 \ | \ y_1(H)=y_2(H)= \dots =y_{k^2-1}(H)=0 \}.
    \end{equation}
\end{enumerate}
    \end{cor}

    \begin{figure}
	\begin{center}
		\includegraphics[width=0.8\columnwidth]{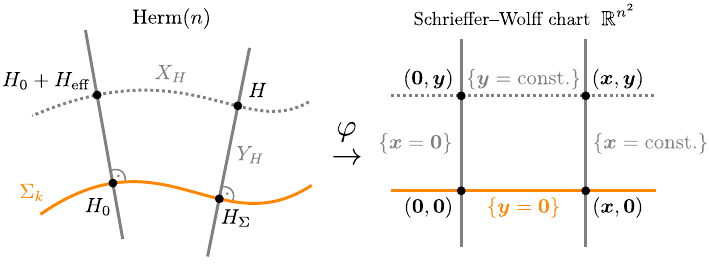}
	\end{center}
	\caption{A schematic picture showing the local chart induced by the SW decomposition around the fixed $H_0$. A matrix $H \in \mathcal{V}_0$ is endowed with two sets of coordinates: (1) the projection $H_{\text{proj}}=H_{\Sigma}$ of $H$ to $\Sigma_k$ is described by the coordinates of $S$, $B$ and $T$ (these are $n^2-k^2+1$ coordinates denoted by $x_i$), (2) the traceless effective Hamiltonian $H_{\textup{eff}}$ is described by $k^2-1$  coordinates denoted by $y_i$. The constant level set of the $y$ and $x$ coordinates through  $H$ is denoted by $X_H$ and $Y_H$, respectively. The $Y_H$ sets are affine subspaces in $\mbox{Herm}(n)$ of dimension $k^2-1$. As an essential step of the proof of Theorem~\ref{th:distance} (see Proposition~\ref{pr:YHorthog}) we will show that the line joining $H$ and $H_{\Sigma}$ is orthogonal to $\Sigma_k$ at $H_{\Sigma}$, or equivalently, the $Y_H$ subspaces are orthogonal to $\Sigma_k$, moreover, $H_{\Sigma}$ is the closest point of $\Sigma_k$ to $H$.
	\label{fig:chart}}
\end{figure}
    
This chart shows that $\Sigma_k$ is a codimension $k^2-1$ submanifold in $\mbox{Herm}(n)$, providing an alternative proof for the Neumann--Wigner theorem (whose original proof is summarised in  Table~\ref{tab:neumann}).

Consider the projection $H_{\textup{proj}}$ of $H$ to $\Sigma_k$  by omitting $H_{\text{eff}}$ from the SW decomposition \eqref{eq:koo}, that is, 
\begin{equation}\label{eq:koo0}
H_{\textup{proj}}=e^{iS} \cdot (H_0+B+T) \cdot e^{-iS}.
\end{equation}
This can be rephrased in the SW chart picture as making the $y_j$ coordinates 0, see Figure~\ref{fig:chart}. We show that  $H_{\textup{proj}}$ can be constructed without using the SW decomposition. The following construction works for the large set $H \in \mbox{Herm}(n) \setminus \Sigma_{[k,k+1]}$ of Hermitian matrices, not only on the domain $\mathcal{V}_0$ of the SW decomposition around an $H_0$. 
Let $\Lambda=U^{-1}HU$ be a diagonalization of $H$, containing the eigenvalues $\lambda_1 \leq \lambda_2 \leq \dots \leq \lambda_k <  \lambda_{k+1} \leq \dots \leq \lambda_n $ of $H$. 
Let
\begin{equation}
    \overline{\lambda}=\frac{\sum_{j=1}^k \lambda_j}{k}.
\end{equation} 
be the mean of the lowest $k$ eigenvalues. 
Let $\Lambda_{\Sigma}$ be the diagonal matrix  obtained from $\Lambda$ by replacing the first $k$ entries $\lambda_j$ ($j=1,\dots, k$) with $\overline{\lambda}$. Define 
\begin{equation}\label{eq:hsigma}
    H_{\Sigma}=U \cdot \Lambda_{\Sigma} \cdot U^{-1},
\end{equation}
the matrix obtained from $H$ by `collapsing the lowest $k$ eigenvalues'. Observe that $H_{\Sigma} \in \Sigma_k$. In Section~\ref{ss.prsw} we show that $H_{\Sigma}$ does not depend on the choice of the unitary matrix $U$, and $H_{\Sigma}$ depends on $H$ in analytic way, although $U$ cannot be chosen continuously. The following statement will be also proved in Section~\ref{ss.prsw}.

\begin{thm}[Projection to $\Sigma_k$, cf. Figure~\ref{fig:chart}]\label{th:proj} If the SW decomposition \eqref{eq:koo} of $H$ with respect to $H_0 \in \Sigma_k$ is defined, then it holds that $H_{\Sigma}=H_{\textup{proj}}$.
\end{thm}

\subsection{Energy splitting and the distance from $\Sigma_k$} 

A classical problem is to find the closest matrix for a given matrix with prescribed properties. 
A well-known example is the Eckart--Young--Mirsky theorem \cite{EckartYoung,Mirsky}: for a given real matrix of size $n \times m$, in Frobenius metric the closest matrix of rank at most $r$ can be obtained from the singular value decomposition by replacing the lowest $\min(m,n)-r$ singular values with 0. Similar results can be found e.g. in \cite{Breiding}.
Here, we show that for a given hermitian matrix $H$, it holds that $H_{\textup{proj}}=H_{\Sigma}$ is the closest matrix with $k$-fold ground state degeneracy, and we express the distance in terms of the standard deviation of the eigenvalues. More generally, our method finds the closest matrix with given degeneracy type.

Recall the distance of two matrices $H, G \in \mbox{Herm}(n)$ induced by the Frobenius metric is denoted by $d(H,G)=\| H-G\|$. The distance $d(H,A)$  of an element $H \in \mbox{Herm}(n)$ and a subset $A \subset \mbox{Herm}(n)$  is the infimum of the distances $d(H,G)$, $G \in A$. 

Consider the standard deviation of the lowest $k$ eigenvalues of $H$, that is, 
\begin{equation}
    D_k(H)=\sigma(\lambda_1, \dots, \lambda_k)=
    \sqrt{\frac{\sum_{j=1}^k  ( \lambda_j- \overline{\lambda} )^2}{k}},
\end{equation}
where $\overline{\lambda}=\sum_{j=1}^k \lambda_j/k$ is the mean of the eigenvalues. 

\begin{thm}[Distance from $\Sigma_k$]\label{th:distance}  
For every $H \in \mbox{Herm}(n) \setminus \Sigma_{[k, k+1]}$
      \begin{equation}\label{eq:distance} 
      d(H, \Sigma_k )= d(H, H_{\Sigma})  = \sqrt{k} \cdot D_k(H)  = \| H_{\textup{eff}} \|
      \end{equation}
holds, where $H_{\textup{eff}}$ is the effective Hamiltonian of $H$ with respect to any $H_0 \in \Sigma_k$ for which the exact SW decomposition \eqref{eq:koo} of $H$ is defined.
\end{thm}

We prove the theorem in Section~\ref{ss:dist}. 
The  first equation expresses that $H_{\Sigma}$ is the closest point of $\Sigma_k$  to $H$.  
We note that the first and second equations hold also for $H \in \Sigma_{[k,k+1]}$, although the projection $H_{\Sigma}$ is not unique, as it depends on the choice of $U$.

In particular, for the special case $k=2$, Eq.~\eqref{eq:distance}  takes the following form:
\begin{equation}\label{eq:distance2}  d(H, \Sigma_2)= d(H, H_{\Sigma} ) = \frac{1}{\sqrt{2}} \cdot | \lambda_2-\lambda_1 |=\| H_{\textup{eff}} \|,
      \end{equation}
       where $\lambda_1$ and $\lambda_2$ are the lowest two eigenvalues of $H$.

\subsection{Order of energy splitting of a degeneracy due to a perturbation}
For $k >2$ one can also consider  the pairwise differences $\lambda_i-\lambda_j$. 
We describe these functions only for one-parameter families.
 Consider a one-parameter perturbation of $H_0 \in \Sigma_k$, also called a one-parameter family, or a curve of Hermitian matrices. This is an analytic function\footnote{For simplicity, we formulate the results for analytic families, but smooth ($\mathcal{C}^{\infty}$) is enough, with the obvious modification of the statements.} $H: \R \to \mbox{Herm}(n)$ defined in a neighborhood of $0 \in \R$ with $H(0)=H_0 \in \Sigma_k$. An example is a linear perturbation in the form $H(t)=H_0 + tH_1$ with a choice of $0\neq H_1 \in \mbox{Herm}(n)$.  
 (Note the slight abuse of notation: so far $H$ denoted a single matrix, here it denotes a map into the matrix space.)

The frequently used quantity `order of energy splitting' can be defined in several different ways. On the one hand, the energy splitting along $H(t)$ is measured by the standard deviation function 
\begin{equation}
\label{eq:orderofstddev}
\mathfrak{s}_k(t):= D_k(H(t)). 
\end{equation}
    By Theorem~\ref{th:distance} this agrees with the distance function up to a scalar factor, that is,     \begin{equation}\label{eq:splitdist} 
    \sqrt{k} \cdot \mathfrak{s}_k(t)= d(H(t), \Sigma_k). \end{equation} 
    
    On the other hand, one can consider the pairwise differences of the eigenvalues, and possibly define the order of the energy splitting either as (1) the minimum of the orders of the pairwise differences or (2) the minimum of the orders of the differences of the neighbours or (3) the order of the difference of the two extrema $\lambda_1$ and $\lambda_k$ or (4) the minimum of the orders of the differences from the mean value.
    Theorem~\ref{th:ordspl} below clarifies that these orders  are equal.
    Moreover, any of them agrees with the order of the standard deviation of the lowest $k$ eigenvalues, and hence, the order of the distancing of $H(t)$ from $\Sigma_k$.

   To formalize the statements of the previous paragraph, we define the pairwise energy splitting functions $\mathfrak{s}_{i,j}$ and
   the splitting from the mean of the eigenvalues $\overline{\mathfrak{s}}_i$, respectively, as 
    \begin{eqnarray}
    \mathfrak{s}_{i,j}(t)&=& \lambda_i(t)-\lambda_j(t),\\
    \overline{\mathfrak{s}}_i(t)&=& \lambda_i(t)- 
    \overline{\lambda}(t),
    \end{eqnarray}
    where $1 \leq i < j \leq k$.     
    Although the functions $\mathfrak{s}_{i,j}$ and $\overline{\mathfrak{s}}_i$ defined in this way are not differentiable at $t=0$ in general (cf. Figure~\ref{fig:order}), for positive $t$ they are given by a power series of $t$ centered at 0, hence their order is well defined in this sense. Indeed, they can be made analytic around 0 by a slightly modified definition, presented in Section~\ref{ss:split}. For this, we use the following well-known fact, and we provide a new proof for it in Section~\ref{ss:split}, based on the SW decomposition:
    
    \begin{thm}[The eigenvalues are analytic]\label{th:anal} 
        The eigenvalues of one-parameter analytic Hermitian matrix families form analytic functions, after a suitable re-ordering of the indexing for negative $t$.
    \end{thm}
      Based on these preliminaries, the equality of the order of the standard deviation function defined in Eq.~\eqref{eq:orderofstddev}, and the further orders listed as (1)-(4) below that, is formalized in the following theorem. 
 
    \begin{thm}[Order of energy splitting]\label{th:ordspl}
\begin{eqnarray}
\mathrm{ord}_0 (\mathfrak{s}_k)&=&\min_{ 1\leq i < j \leq k} \{ \mathrm{ord}_0 (\mathfrak{s}_{i,j})\}\\
&=&\min_{ 1 \leq i \leq k-1} \{ \mathrm{ord}_0 (\mathfrak{s}_{i,i+1})\}\nonumber\\
&=&\mathrm{ord}_0 (\mathfrak{s}_{1,k})\nonumber\\
&=&\min_{ 1 \leq i \leq k} \{ \mathrm{ord}_0 (\overline{\mathfrak{s}}_{i})\}\nonumber ,
\end{eqnarray}
where $\mathrm{ord}_0 (f)$ denotes the order of the function $f$ at $t=0$.
    \end{thm}

Theorem~\ref{th:ordspl} together with Equation~\eqref{eq:splitdist} provides a geometric description for the order of the splitting. Namely, it is independent of the way we measure it, and it agrees with the order of distancing from $\Sigma_k$. Moreover, it also agrees with the leading order of the effective Hamiltonian, which is the minimum of the orders of the coordinates $y_1, \dots, y_{k^2-1}$. 

\begin{theorem}[Energy splitting and distancing have the same order]\label{th:minden}

Let $r$ denote the order of energy splitting, that is, any of the five equal quantities in Theorem~\ref{th:ordspl}. Then 
\begin{equation}\label{eq:ord22}
    r=\mathrm{ord}_0(t \mapsto d(H(t), \Sigma_k))=\mathrm{ord}_0 (t \mapsto H_{\textup{eff}}(t))=\min_{1 \leq j \leq k^2-1} \{ \mathrm{ord}_0 (t \mapsto y_j(t)) \},
\end{equation}
where $H_{\textup{eff}}(t)$ is the effective Hamiltonian of $H(t)$ with respect to any diagonal $H'_0 \in \Sigma_k$ (not necessarily equal to $H_0$) for which the exact SW decomposition \eqref{eq:koo} of $H(t)$ is defined. 
\end{theorem}

\begin{cor}\label{co:tangstick}
    For a linear family $H(t)=H_0+tH_1$ the order of energy splitting is 
    \begin{itemize}
        \item $r=1$ if  $H_1$ is not tangent to $\Sigma_k$ at $H_0$, i.e. $H_1 \notin T_{H_0} \Sigma_k$,
        \item $r \geq 2$ if $H_1$ is  tangent to $\Sigma_k$ at $H_0$, i.e. $H_1 \in T_{H_0} \Sigma_k$.
    \end{itemize}
\end{cor}
According to the above proposition, $r >2$ means stronger `stickiness' of the tangent vector $H_1$ to the degeneracy submanifold $\Sigma_k$ at $H_0$.

\subsection{Parameter-dependent quantum systems and Weyl points}
\label{sec:weylrobustness}

Generic degeneracy points in quantum systems depending on three parameters are often called Weyl points. As we show below, their protection against perturbations relies on the transversality theorem. A similar idea was already suggested by Arnold in \cite[Sec. 7]{ArnoldSelMath1995} in a slightly different context, 
 without explicit formulation. Here we provide a characterization of Weyl points in a general context in terms of SW and transversality. 

Parameter-dependent quantum systems are described by a smooth ($\mathcal{C}^{\infty}$) map from a manifold $M$ of dimension $m$ to the space of Hamiltonians, i.e., Hermitian matrices $\mbox{Herm}(n)$. Slightly abusing the notation again, we denote this map by $H$. So from now on $H: M \to \mbox{Herm}(n)$ is  a smooth map with $H(p_0)= H_0 \in \Sigma_k$ for a  point $p_0 \in M$, which is called degeneracy point. For simplicity, assume that $H_0$ is diagonal --- it can be reached by a unitary change of basis in $\C^n$, see App.~\ref{ss:coordfree}. 

As above, let $\mathcal{V}_0$ denote a neighborhood of $H_0$ where the SW decomposition is unique.
Then, consider the corresponding neighborhood $\mathcal{W}_0 \subset H^{-1}(\mathcal{V}_0)$ of $p_0$ in $M$.
On this $\mathcal{W}_0$, the (traceless) effective Hamiltonian map is defined by the SW decomposition, that is, 
\begin{equation}
          H_{\textup{eff}}: \mathcal{W}_0 \to \mbox{Herm}_0(k).
\end{equation}
By introducing a local chart in $\mathcal{W}_0$ centered at $p_0$ (i.e., $p_0=0$) and expressing $H_{\textup{eff}}$ in the basis $\breve{c}_i$ of $\mbox{Herm}_0(k)$ (cf.~Sec.~\ref{s:pre}), we obtain a map  
\begin{equation}
          h: \R^m \to \R^{k^2-1}
\end{equation}
defined in a neighborhood of the origin, satisfying $h(0)=0$. The $i$-th component of $h$ is $h_i=y_i \circ H$, where $y_i$ are the effective coordinates of the SW chart in Corollary~\ref{co:locchart}, $i=1, \dots, k^2-1$.

If we want to describe the `type' of the degeneracy point $p_0$, this problem leads to the description of the intersection of $H$ and $\Sigma_k$ at $H_0$, which can be reduced to the description of the `type' of the root of $h$ at 0. For this it is necessary to know $h$ in an arbitrarily small neighborhood of the origin. This description of the degeneracy point leads to the study of map germs $h: (\R^m, 0) \to (\R^{k^2-1}, 0)$, and their classifications in singularity theory. 
Although this relation between degeneracy points and  singularities implicitly appears in several works, see e.g. \cite{Teramoto2017, teramoto2020application, Pinter2022, naselli2024stability}, this relation is essentially based on the geometric picture described here. 
We illustrate this relation on the characterization of Weyl points. Weyl points represent the simplest type of degeneracy, and we show that they correspond to the simplest type of intersection called transverse intersection of $H$ and $\Sigma$, and also to the simplest singularity type of $h$. These relations explain the protected nature of Weyl points, as a straightforward consequence of transversality. 
            
Recall that in the physics terminology, Weyl points are isolated twofold degeneracy points in a 3-dimensional parameter space, with linear energy splitting in every direction. We formalize this informal description below.
According to this description, we restrict our study to
twofold (ground state) degeneracy --- that is,  $H(p_0)= H_0 \in \Sigma_2$ -- and $m=3$. 
Hence, the resulting map (germ) is $h: (\R^3, 0) \to (\R^{3}, 0)$.

\begin{rem}
    Observe that the equality of the dimensions is coming from the fact that the codimension of $\Sigma_2$ in $\mbox{Herm}(n)$ is 3, and we also choose 3 for the dimension of the parameter space $M$. This is the case in many physical examples (3D crystals, magnetic fields, see Example~\ref{ex:fizrendszerek}), hence, this coincidence of the (co)dimensions explains the appearance of Weyl points.
\end{rem}

 Recall the notion of transversality, see e.g. \cite{guillemin2010differential, golubitsky2012stable}. Consider smooth manifolds $M$ and $N$, a smooth submanifold $Z \subset N$ of $N$, and a smooth map $f: M \to N$ with  $f(p) \in Z$ for a point $p \in M$. The tangent map  $\mathrm{d} f_{p}$ is a linear map from the tangent space $T_{p} M$ of $M$ at $p$ to the tangent space $T_{f(p)}N$ of $N$ at $f(p)$. In local coordinates $(\mathrm{d}f)_{p}$ is the Jacobian matrix. Then, $f$ is \emph{transverse} to $Z$ at $p$ if 
\begin{equation}\label{eq:transverse}
    T_{f(p)} Z + (\mathrm{d} f)_{p} (T_{p} M) = T_{f(p)} N
\end{equation}
holds. That is, the tangent space of the submanifold $Z$ and the image of the tangent map $(\mathrm{d} f)_p$ span the tangent space of $N$ at $f(p)$, see panel (a) and (d) in Figure~\ref{fig:transvers}.

 \begin{figure}
	\begin{center}
		\includegraphics[width=0.7\columnwidth]{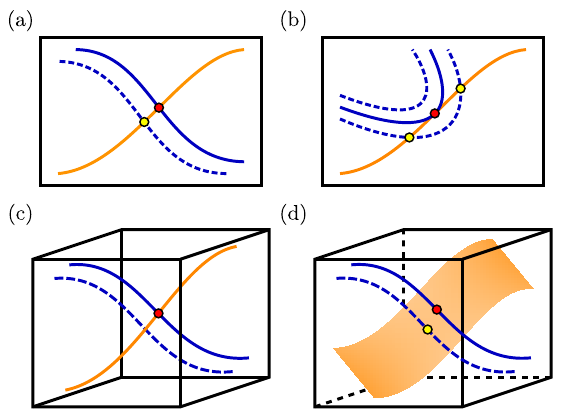}
	\end{center}
	\caption{Transversality, as a characterization of Weyl points, cf. Equation~\eqref{eq:transverse} and Theorem~\ref{th:charweyl}. On each panel the submanifold is orange, and the image of the map is blue. 
 (a) Transverse intersection of two curves at a single point on the plane. For a small perturbation the intersection point slightly moves away, but it remains transverse.
 (b) Non-transverse intersection of two curves on the plane at a single point. Non-transversality follows from the fact that their tangent lines coincide, hence they do not span the tangent plane of the plane, thus Equation~\eqref{eq:transverse} does not hold for them. For a small perturbation the intersection splits into transverse intersection points, whose number is either 0 or 2 in this case.
 (c) Non-transverse intersection of two curves in 3-space at a single point. Non-transversality follows from
the inequality $1+1 < 3$, that is, the sum of the dimensions is less than the dimension of the ambient space. In this case Equation~\eqref{eq:transverse} cannot hold at an intersection point. According to this, the intersection point disappears for a small perturbation.
(d) Transverse intersection of a surface and a curve at a single point in 3-space. Any small perturbation preserves transversality, the intersection point only slightly moves away.  In the context of parameter-dependent quantum systems $H: M \to \mbox{Herm}(n) $, the submanifold $\Sigma_2 $ has codimension 3, the same as the dimension of the parameter space $M$. The red point is the image of a twofold ground state degeneracy point $H(p_0)=H_0 \in \Sigma_2$. It is a Weyl point if $H$ is transverse to $\Sigma_2$ in $\mbox{Herm}(n)$ at $p_0$ (like in panel (a) and (d)), otherwise it is a non-generic degeneracy point, which splits into Weyl points (yellow points) or disappears for a small perturbation of $H$ (dashed blue line). \label{fig:transvers}}
\end{figure}

Transversality implies the expected dimension of the pre-image, namely, if $f$ is transverse to $Z$ at every point of $f^{-1}(Z)$, then $f^{-1}(Z)$ is a submanifold of $M$ of dimension $\dim (f^{-1}(Z))=\dim (M) + \dim (Z) - \dim (N)$, see \cite[pg. 28]{guillemin2010differential}. 
By the statements closely related to the transversality theorem \cite[pg. 35, 68-69]{guillemin2010differential}, transversality with respect to a fixed submanifold is a stable and generic property of smooth maps, cf. Figure~\ref{fig:transvers}.

     \begin{thm}[Characterization of Weyl points]\label{th:charweyl} Given a parameter-dependent quantum system by a $\mathcal{C}^{\infty}$ map $ H: M^3 \to \mbox{Herm}(n)$ and a two-fold ground state degeneracy point $p_0 \in M^3$, $H(p_0)=H_0 \in \Sigma_2$ and the induced effective map germ $h: (\R^3, 0) \to (\R^3, 0)$, the following properties are equivalent:
     \begin{enumerate}
         \item The map $H$ is transverse to $\Sigma_2$ at $p_0$,
         \item The rank of the Jacobian of $h$ at 0 has maximal rank, i.e. $\rk ((\mathrm{d} h)_0)=3$.
         \item Taking any curve $\gamma: (\R, 0) \to (M^3, p_0)$ with  $\gamma(0)=p_0$ and $\gamma'(0) \neq 0$, for the composition $(h \circ \gamma)'(0) \neq 0$ holds. 
         \item Taking any curve $\gamma: (\R, 0) \to (M^3, p_0)$ with  $\gamma(0)=p_0$ and $\gamma'(0) \neq 0$, for the composition $H \circ \gamma$ the order of energy splitting is 1.

     \end{enumerate}
     \end{thm}

       A two-fold degeneracy point $p_0$ satisfying any, hence all of the properties (1)--(4) is called Weyl point. Note that (4) formulates the physicist definition. 
       
       \begin{rem}\label{re:weyl} By (2), the characterization of Weyl points is already determined by the first-order part of the exact effective map $h$. It implies that in the approximate computation of the SW transformation (see e.g. \cite{Winkler, day2024pymablock})      
       the first-order term is sufficient to decide whether a degeneracy point $p_0 \in M^3$  is a Weyl point or not. Namely, it can be decided by the following steps, cf. Example~\ref{ex:fizrendszerek}:
       \begin{enumerate}
           \item Take the first-order SW transformation, i.e. the upper-left traceless $2 \times 2$ block of $H(p)$.
           \item Via the Pauli decomposition it can be considered as a map ${h}^{(1)}: \R^3 \to \R^3$, defined in a neighborhood of  the origin $p_0=0 $.
           \item $p_0$ is a Weyl point if and only if $\rk((\mathrm{d}h^{(1)})_{0}=3$.
       \end{enumerate}
\end{rem}
       
       The properties of the transversality mentioned above imply the following corollaries.

     \begin{cor}[Weyl points are isolated]\label{co:weyliso}
         Every Weyl point $p_0$ is an isolated degeneracy point, in the sense that there is a neighborhood $\mathcal{W}_0 \subset M^3$ of $p_0$ such that $H(p) \notin \Sigma_2$ for $p_0 \neq p \in \mathcal{W}_0$.
     \end{cor}

     Intuitively, the  protected nature of Weyl points includes the following phenomena: 
     \begin{enumerate}[label=(\alph*)]
     \item \emph{Weyl points are stable:} For any small perturbation, a Weyl point does not disappear, it only gets displaced by a small amount in $M$ (if at all).
         \item \emph{Weyl points are generic:} For a generic perturbation, a non-generic degeneracy point splits into Weyl points (or possibly disappears). 
     \end{enumerate}
     
    To translate (a) and (b)  into rigorous claims, we consider one-parameter perturbations $H_t$ of $H$. Formally, these are $\mathcal{C}^{\infty}$ maps (germs) from $M^3 \times \R$ to $\mbox{Herm}(n)$ defined on a neighborhood of $(p_0, 0)$, such that $H_{t=0}=H$. 
    For simplicity we formulate the statement in a local version for isolated twofold degeneracy points, although, it can be generalized for non-isolated or multifold degeneracy points, see Remark~\ref{re:altalanosabb}.

     \begin{cor}[Weyl points are stable and generic]\label{co:weylgene} 
     Let 
     $ H: M^3 \to \mbox{Herm}(n)$ be  a parameter-dependent quantum system with an isolated two-fold ground state degeneracy point $p_0 \in M^3$, $H(p_0)=H_0 \in \Sigma_2$. Let $\mathcal{W}_0 \subset M^3$ be a neighborhood of $p_0$ whose closure $\mbox{cl}(\mathcal{W}_0)$ is compact and it does not contain other  ground state degeneracy points, that is, $H^{-1}(\mbox{cl}(\Sigma_2) )\cap \mbox{cl}(\mathcal{W}_0)=\{ p_0 \}$. 
         \begin{enumerate}[label=(\alph*)]
             \item If $p_0$ is a Weyl point, then for every one-parameter perturbation $H_t$ of $H_{t=0}=H$, there is an $0 < \epsilon $, such that for $|t| < \epsilon$ the perturbed Hamiltonian $H_t$  has exactly one degeneracy point in $\mathcal{W}_0$, and it is a Weyl point. Moreover, there is a $\mathcal{C}^{\infty}$ curve $\gamma:  (-\epsilon, \epsilon) \to M^3$ such that $\gamma(t)$ is the unique Weyl point of $H_t$ in $\mathcal{W}_0$.
             \item If $p_0$ is not a Weyl point, then for every $0<\epsilon$ there is a $K \in \mbox{Herm}(n)$ with $\|K\| < \epsilon$, such that every degeneracy point of the perturbed map $H_K: p \mapsto H_K(p)=H(p)+K$ in $\mathcal{W}_0$ is a Weyl point.  
         \end{enumerate}
     \end{cor}

    Finally, Property (2) of Theorem~\ref{th:charweyl} implies that the topological charge of a Weyl point is $\pm 1$. Indeed, the topological charge is equal to the local degree $\deg_0 h$ of $h$ at $0$, which is $\pm 1$, if $\rk (\mathrm{d} h_0)=3$, see \cite{Pinter2022}.
    (Note that in the physics literature
    the topological charge  is defined as the first Chern number of the eigenvector bundle corresponding to the lowest eigenvalue, evaluated on a small sphere in $M^3$ around $p_0$, but it is equal to the local degree $\deg_0 h$ of $h$ at $0$.)

\subsection{Examples}\label{ss:ex}

\begin{ex}[$\mbox{Herm}(3)$]
In the case of $3\times 3$ matrices, the exact SW decomposition \ref{th:sw} can be given in a closed form using Cardano's formula to determine the eigenvalues, then, to get the block-diagonalized form, one needs to perform the exact direct rotation between the near-degenerate subspaces, cf. Appendix~\ref{app:directrot}. However, the resulting expressions are extremely complicated. As an alternative to exact decomposition, one might use a series expansion~\cite{Winkler} to approximate the terms of the decomposition.

As an example, we take a general $H \in \mbox{Herm}(3)$ around $H_0=\mathrm{diag}(0,0,1)\in\Sigma_2$ with the elements as coordinates such that
\begin{equation}\label{eq:3x3chart}
H=\begin{pmatrix}
    0&0&0\\
    0&0&0\\
    0&0&1
\end{pmatrix}
+\begin{pmatrix}
    v+z&x-iy&p-iq\\
    x+iy&v-z&r-is\\
    p+iq&r+is&w
\end{pmatrix}.
\end{equation}
The SW decomposition up to second order reads
\begin{eqnarray}
iS & = & 
\begin{pmatrix}
    0&0&(1+v-w+z)(p-iq)+(x-iy)(r-is)\\
    0&0&(1+v-w-z)(r-is)+(x+iy)(p-iq)\\
    -\text{h.c.}&-\text{h.c.}&0
\end{pmatrix}+\dots,\\
B &=& \left(1+w+p^2+q^2+r^2+s^2+\dots\right)
\begin{pmatrix}
    0&0&0\\
    0&0&0\\
    0&0&1
\end{pmatrix},\\ 
T &=&  \left(v-\frac{p^2+q^2+r^2+s^2}{2}+\dots\right)\begin{pmatrix}
    1&0&0\\
    0&1&0\\
    0&0&0
\end{pmatrix}, \mbox{ and} \\
  H_{\textup{eff}} & =& 
   \begin{pmatrix}
    z&x-iy&0\\
    x+iy&-z&0\\
    0&0&0
\end{pmatrix}-\frac{1}{2}
\begin{pmatrix}
    p^2+q^2-r^2-s^2&2(p-iq)(r+is)&0\\
    2(p+iq)(r-is)&-p^2-q^2+r^2+s^2&0\\
    0&0&0
\end{pmatrix}+\dots\nonumber\\
&=&\left(x-pr-qs+\dots\right)\sigma_x\nonumber\\
&+&\left(y+ps-qr+\dots\right)\sigma_y\\
&+&\left(z-\frac{p^2+q^2-r^2-s^2}{2}+\dots\right)\sigma_z.\nonumber
\end{eqnarray}
Recall that in the last equation $H_{\textup{eff}}$ is considered as a $ 2 \times 2$ matrix of trace zero. Note that the effective Hamiltonian in the first-order $H_\text{eff}^{(1)}$ is the truncation of $H$ to its near-degenerate upper-left $2\times 2$ block. 

\end{ex}

\begin{ex}[Exact SW decomposition] 
In the special case $p=q=r=s=0$ in Eq.~\eqref{eq:3x3chart} it is trivial to perform the SW decomposition, as $H$ is already block diagonal, and $S=0$. 

For a simple non-trivial example we take the 2 dimensional section of $\mbox{Herm}(n)$ 
\begin{eqnarray}\label{eq:3x3chart2dim}
    H(p,r)  = \begin{pmatrix}
        0&0&p\\
        0&0&r\\
        p&r&1 \\
    \end{pmatrix}.
\end{eqnarray}
The components $S$, $B$, $T$, $H_{\textup{eff}}$  of the SW decomposition \eqref{eq:koo} of $H$ with respect to $H_0$ can be expressed explicitly as functions of $p$ and $r$ as:
\begin{eqnarray}
iS(p,r) & = & \frac{1}{\sqrt{p^2+r^2}}\tan^{-1}\left(\frac{\sqrt{1+4p^2+4r^2}-1}{2\sqrt{p^2+r^2}}\right)\begin{pmatrix}
    0&0&p\\
    0&0&r\\
    -p&-r&0
\end{pmatrix},\\
B(p,r) &=& \frac{1+\sqrt{1+4p^2+4r^2}}{2}
\begin{pmatrix}
    0&0&0\\
    0&0&0\\
    0&0&1
\end{pmatrix},\\ 
T(p,r) &=&  
\frac{1-\sqrt{1+4p^2+4r^2}}{4}\begin{pmatrix}
    1&0&0\\
    0&1&0\\
    0&0&0
\end{pmatrix}, \mbox{ and} \\
  H_\text{eff}(p,r) &=&  
\frac{1-\sqrt{1+4p^2+4r^2}}{4(p^2+r^2)}\begin{pmatrix}
    p^2-r^2&2pr&0\\
    2pr&-p^2+r^2&0\\
    0&0&0
\end{pmatrix}\nonumber\\
&=&\frac{1-\sqrt{1+4p^2+4r^2}}{4(p^2+r^2)}\left(2pr\sigma_x+(p^2-r^2)\sigma_z\right).
\end{eqnarray}
Observe that every term $S$, $B$, $T$ and $H_{\textup{eff}}$ is an analytic function of $(p,r)$ in a neighborhood of $(0,0)$, although this is not obvious at first sight. The non-analytic behaviour of these maps far from $(0,0)$ shows that the SW decomposition can only be defined locally.

\end{ex}

\begin{ex}[Parameter-dependent quantum systems exhibiting Weyl points.]\label{ex:fizrendszerek}

In Sec.~\ref{sec:weylrobustness} we have discussed the generic nature of Weyl points in a mathematical context.
This discussion is relevant to many physical setups. 
Weyl points arise as spectral features in the electronic, phononic, photonic, magnonic band structures of crystalline materials, or metamaterials. 
More generally, Weyl points also arise in quantum systems described by a Hamiltonian depending on three parameters.
One example is an interacting spin system in a homogeneous magnetic field, where the manifold of parameters is $M = \R^3$, corresponding to the external magnetic field vector \cite{Scherubl2019,Frank2020}.
Another example is a multiterminal Josephson junction, where the manifold of parameters is the three-dimensional torus, hosting the values of three magnetic flux biases piercing the loops of the superconducting circuit \cite{Riwar2016,Fatemi2021,Frank2022}. 

An explicit $3\times 3$ example for a Hamiltonian with a Weyl point is the following:
\begin{equation}
H(x,y,z)=\begin{pmatrix}
    z&x-iy&y-ixz\\
    x+iy&-z&x-iyz\\
    y+ixz&x+iyz&1+xyz
\end{pmatrix}.
\end{equation}
This matrix has a two-fold ground state degeneracy at $x=y=z=0$. This is a Weyl point according to Remark~\ref{re:weyl} as the first-order effective Hamiltonian is
\begin{equation}
    H_{\mathrm{eff}}^{(1)}(x,y,z)=x\sigma_x+y\sigma_y+z\sigma_z,
\end{equation}
which, expanded in the orthonormal Pauli basis, corresponds to the effective map
\begin{equation}
    h^{(1)}(x,y,z)=\sqrt{2} \left(x, y, z\right).
\end{equation}
The Jacobian of $h^{(1)}$ is $\sqrt{2}$ times the identity, hence it has maximal rank 3. 
Note that the off-block matrix elements can be arbitrary functions of $x$,$y$ and $z$ with constant term 0, so that the resulting Hamiltonian still describes a Weyl point at the origin. Moreover, one can also perturb the $2\times 2$ block with higher-order terms without destroying the Weyl point, as they do not change the Jacobian (and hence its rank) of the first-order effective map $h^{(1)}$ at the origin. Off-block and higher-order terms have significant effect, with the possibility of creating new degeneracy points, only far from the origin.

A $3\times 3$ counterexample which is a point-like two-fold ground state degeneracy that is not a Weyl point is described in Eq.~\eqref{eq:3x3chart2dim}. It can be considered as a $\mathcal{C}^{\infty}$ map $H: \R^2 \to \mbox{Herm}(3)$.
 The effective map  $h: (\R^2, 0) \to (\R^3, 0)$ defined in a neighborhood of the origin reads 
\begin{equation}
    h(p,r)=\frac{1-\sqrt{1+4p^2+4r^2}}{2\sqrt{2}(p^2+r^2)}\left(2pr, 0, (p^2-r^2)\right).
\end{equation}
This shows that the order of the energy splitting is 2 along every curve $\gamma: \R \to \R^2$ with $\gamma(0)=(0,0)$ and $\gamma'(0)$. Therefore, the order of distancing $\mathrm{ord}_0(t \mapsto d(H (\gamma(t)), \Sigma_2))$ is 2 along every such curve. Moreover, this map is not even equidimensional, meaning that a generic small perturbation lifts the degeneracy, cf. panel (c) in Figure~\ref{fig:transvers}.
\end{ex}

\begin{ex}[Parameter-dependent non-interacting quantum systems exhibiting higher-order energy splitting of a degeneracy]

\label{ex:ssh}

In Sec.~\ref{s:intro}, we have mentioned physical systems having energy degeneracies that exhibit higher-order energy splitting for physically relevant perturbations.
Here, we illustrate how those examples relate to our formalism and results.
Our first example is the Su-Schrieffer-Heeger (SSH) model \cite{SSH,Asboth}.
As an application of our results, we translate a well-known property of the SSH model into information about the geometry of the corresponding degeneracy submanifold of the Hermitian matrix space. 

We describe the SSH model (or \emph{SSH chain}) as a tight-binding model of a single electron on a one-dimensional bipartite crystal lattice with a unit cell of two atoms (orbitals), translational invariance, and open boundary conditions, i.e., the lattice terminates at both ends.
The SSH Hamiltonian is in $\text{Herm}(2N)$, where $N$ denotes the number of unit cells. 
For example, an SSH chain of $N=4$ unit cells is described by the following $8\times 8$ Hamiltonian matrix:
\begin{equation}
H_\text{SSH} = \left( \begin{array}{cccccccc}
0 & v & 0 & 0 & 0 & 0 & 0 & 0\\
v & 0 & w & 0 & 0 & 0 & 0 & 0\\
0 & w & 0 & v & 0 & 0 & 0 & 0\\
0 & 0 & v & 0 & w & 0 & 0 & 0\\
0 & 0 & 0 & w & 0 & v & 0 & 0\\
0 & 0 & 0 & 0 & v & 0 & w & 0\\
0 & 0 & 0 & 0 & 0 & w & 0 & v\\
0 & 0 & 0 & 0 & 0 & 0 & v & 0
\end{array}
\right).
\end{equation}

We define the unperturbed Hamiltonian as $H_0 = H_\text{SSH}(v=0,w=1)$, which is referred to as the \emph{topological fully dimerized limit} of the SSH chain.
It is well known, and straightforward to show from the block diagonal structure of $H_0$, that $H_0$ has a twofold degenerate eigenenergy at zero, separated from the other two eigenenergies $1$ and $-1$ which both have $(N-1)$-fold degeneracy. 
In this section, we will use $\Sigma_\text{SSH}$ to denote the degeneracy submanifold of $\text{Herm}(2N)$ where the $N$th and $(N+1)$th ordered eigenvalues are degenerate, but different from the others, i.e., $\{\dots \leq \lambda_{N-1} < \lambda_{N} = \lambda_{N+1} < \lambda_{N+2} \leq \dots \}$. Clearly, $H_0$ is on this degeneracy submanifold $\Sigma_\text{SSH}$.

The twofold zero-energy degeneracy is robust against special perturbations.
Consider a perturbation $H_1$ that is referred to as \emph{disordered nearest-neighbor hopping} in physics terminology.
This means that $H_1$ is a tridiagonal Hermitian matrix with zeros on the diagonal, i.e., 
it has nonzero elements only on the first diagonal above the main diagonal, and the complex conjugates of those on the first diagonal below the main diagonal.
Such perturbation defines a $(4N-2)$-dimensional subspace in the Hermitian matrix space. 
This perturbation causes an energy splitting of the zero-energy degeneracy, with an order of energy splitting of $N$ (at least), a fact which can be shown, e.g., by SW perturbation theory after diagonalizing $H_0$. 
This result, well known in the subfield of physics studying \emph{topological insulators}, combined with Corollary \ref{co:tangstick} and the first equality of Theorem \ref{th:minden}, implies the following geometrical property of the submanifold $\Sigma_\text{SSH} \subset \text{Herm}(2N)$: 
in the point $H_0 \in \Sigma_\text{SSH}$, the $(4N-2)$-dimensional subspace corresponding to the perturbation $H_1$ is in the $(4N^2-3)$-dimensional tangent space of $\Sigma_\text{SSH}$ at $H_0$ (according to Corollary \ref{co:tangstick}, since $r = N \geq 2$), and this subspace of the tangent space is sticking to the to the degeneracy submanifold especially strongly, as the distancing function in this subspace has order $N$ (according to Theorem \ref{th:minden}).

\end{ex}

\begin{ex}[Parameter-dependent interacting quantum systems exhibiting high-order energy splitting of a degeneracy.]
\label{ex:ising}

A further example of a robust degeneracy is the twofold ground-state degeneracy of the quantum mechanical Ising model, a model of interacting qubits (or spins), in the presence of a transverse-field perturbation.

Here, we focus on the 1D Ising model, i.e., the Ising chain, with open boundary conditions. 
The Ising chain consists of $N$ qubits that are nearest-neighbor coupled with Ising interaction: 
\begin{equation}
    H_0 = - \left(
        \sigma_z \otimes \sigma_z \otimes I \otimes I \otimes \dots 
        +
        I \otimes \sigma_z \otimes \sigma_z \otimes I \otimes \dots 
        + 
        \dots
    \right)
    \equiv
    - \sum_{i=1}^{N-1} Z_i Z_{i+1}.
\end{equation}
Here $Z = \ket{1} \bra{1} - \ket{0} \bra{0}$ is the single-qubit Pauli $z$ operator defined using the orthonormal basis states $\ket{0}$ and $\ket{1}$ of the qubit.  
Furthermore, we use the usual shorthand notation for the tensor products of operators, e.g., 
$Z_1 Z_2 \equiv \sigma_z \otimes \sigma_z \otimes I \otimes I \otimes \dots$, etc.
The unperturbed Hamiltonian $H_0$ is in $\text{Herm}(2^N)$, and $H_0$ has a twofold degenerate ground-state energy at $-N+1$, with two orthogonal ground states given as $\ket{00\dots}$ and $\ket{11\dots}$.

Then, a special type of perturbation is the disordered transverse field
\begin{equation}
H_1 = \sum_{i=1}^N \left( x_i X_i + y_i Y_i \right),
\end{equation}
which is described by $2 N$ real parameters, namely, the \emph{local transverse fields} $x_i$ and $y_i$ multiplying the Pauli operators $X$ and $Y$ of each qubit.
Hence, the parameter space of this perturbation $H_1$ is $\R^{2N}$.
In this case, the order of energy splitting due to this perturbation is $N$ (at least), which can be proven, e.g., by mapping this model to the SSH model via the Jordan-Wigner transformation (see, e.g., Eq.~(6) of \cite{Juhasz}).
The fact that the order of energy splitting is $N$ is translated, using Corollary \ref{co:tangstick} and the first equality of Theorem \ref{th:minden}, to a geometry result, as follows: 
The parameter space $\R^{2N}$ of the perturbation correspond to a $2N$-dimensional subspace of the $(2^{2N}-3)$-dimensional tangent space of the twofold ground-state degenerate submanifold at its point $H_0$
(according to Corollary \ref{co:tangstick}, since $r=N \geq 2$), and this $2N$-dimensional subspace is sticking to the degeneracy submanifold especially strongly, as the distance function in this subspace  has order $N$ (according to Theorem \ref{th:minden}).
    
\end{ex}

\begin{ex}[Hamiltonians defined from stabiliser quantum error correction codes] 

\label{ex:stabiliser}

Stabiliser-code Hamiltonians describe interacting spins, and are directly related to $[[n,k,d]]$ stabiliser quantum error correction codes \cite{GottesmanPhD}.
Here, $n$ is the number of physical qubits, $k$ is the number of logical qubits, and $d$ is the code distance.
Examples are the Toric Code \cite{KitaevToricCode} or the five-qubit code \cite{GottesmanPhD}.
For code families incorporating the $n \to \infty$ limit and corresponding to a finite-dimensional qubit lattice, such as the Toric Code, the stability of the ground-state degeneracy under local perturbations in the thermodynamic limit has been proven \cite{Klich,BravyiTO,BravyiTO2}; this research has been recently extended \cite{Lavasani,YaodongLi} to the more general class of quantum low-density parity check codes \cite{BreuckmannPRXQuantum,Panteleev}. 

For a stabiliser-code Hamiltonian derived from an $[[n,k,d]]$ stabiliser code, (i) the ground-state degeneracy is $2^k$-fold, and (ii)
the order of ground-state energy splitting caused by 1-local perturbations are of order $d$.
(The latter claim is indicated in \cite{KitaevToricCode} for the Toric Code, and is proven for $[[n,1,3]]$ stabiliser-code Hamiltonians in \cite{Bacon}; we are not aware of a formal proof for general stabiliser-code Hamiltonians though.)
Combining this information with our results, we deduce geometric information for the $2^k$-fold ground-state degeneracy submanifolds.
We illustrate this on the five-qubit code, but a similar analysis can also be done for other stabiliser code Hamiltonians, such as the Toric Code.

In the $[[5,1,3]]$ five-qubit code, 5 physical qubits are used to encode 1 logical qubit. The sum of the so-called generators defines the stabiliser-code Hamiltonian $H_0$, i.e., a Hermitian matrix in $\mathrm{Herm}(32)$,

\begin{equation}
H_0 = - X_1 Z_2 Z_3 X_4
- X_2 Z_3 Z_4 X_5
- X_1 X_3 Z_4 Z_5
- Z_1 X_2 X_4 Z_5.
\end{equation}
$H_0$ has twofold ground-state degeneracy, that is, $\lambda_1 = \lambda_2 < \lambda_3 \leq \dots$, and its two-dimensional ground-state subspace encodes the logical qubit. The logical qubit is robust against errors of the physical qubits, in the following sense. Consider one-parameter perturbations $H(t)=H_0+tH_1$ where $H_1$ is 1-local, that is, 
\begin{equation}\label{eq:local}
H_1=
A_1 + A_2 + A_3 + A_4 + A_5,
\end{equation}
where $A_i$-s are linear combinations of the three Paulis acting on site $i$.
Then the robustness means that the order of energy splitting $\mathrm{ord}_0(\mathfrak{s}_2(t))$ of the lowest two eigenvalues along $H(t)$ is at least 3, i.e., the so-called `code distance' of the five-qubit code.

Using our results connecting the order of energy splitting and the order of distancing, we can infer geometrical information as follows.
In the space $\mbox{Herm}(32)$ of dimension $32^2=1024$, the twofold ground state degeneracy submanifold $\Sigma_2$ has dimension $1021$. 
The 15-dimensional space of perturbations described by $H_1$ correspond to a 15-dimensional subspace of the 1021-dimensional tangent space of the twofold ground-state degenerate submanifold at its point $H_0$ (according to Corollary \ref{co:tangstick}, since $r = 3 \geq 2$), and this 15-dimensional subspace is sticking to the degeneracy submanifold strongly, as the distance function in this subspace has order $3$ (according to Theorem \ref{th:minden}).

Examples \ref{ex:ssh}, \ref{ex:ising}, and \ref{ex:stabiliser} demonstrate that known properties of energy splittings in tight-binding models and interacting spin models can be used to illustrate the geometrical description of the degeneracy submanifolds in the space of Hermitian matrices. 
An exciting application of our results would be to make use of this connection in the reverse direction: to describe the degeneracy submanifold using geometrical tools, translate that information to the language of quantum Hamiltonians, and thereby enable the construction of interacting quantum systems with robust energy degeneracies, or even novel quantum error correction codes.
    
\end{ex}

\section{Proofs and details}
\label{sec:proofs}

This section contains the proofs and further details of the statements in Section~\ref{sec:summary}. Every subsection here has the same title as the corresponding subsection in Section~\ref{sec:summary}.

\subsection{The Schrieffer–Wolff transformation induces a local chart}\label{ss.prsw} 

\begin{proof}[Proof of Theorem~\ref{th:sw} on the exact SW decomposition]
Recall that the diagonal matrix $H_0 \in \Sigma_k$ is fixed, and define $\breve{B}:=B+T+H_{\textup{eff}}=\widetilde{B}-H_0$, see decomposition~\eqref{eq:koo} and Figure~\ref{fig:pikto} for notations. $\breve{B}$ is a block diagonal matrix of $k \times k$ and $(n-k) \times (n-k)$ blocks. The correspondence  
    \begin{equation}
        (S, \breve{B}) \mapsto H=f(S, \breve{B}):=e^{iS} \cdot (H_0+\breve{B}) \cdot e^{-iS}
    \end{equation}
    defines  a real analytic map $f: \R^{n^2} \to \R^{n^2}$, in fact, the dimension of the space of $(S, \breve{B})$ is $(n-k)^2+k^2+2 \cdot k \cdot (n-k)=n^2$. We apply the analytic inverse function theorem \cite[pg. 47, Thm. 2.5.1.]{krantz2002primer}.

We show that the Jacobian  of $H$ at $0$ has maximal rank, that is, $\rk ((\mathrm{d} f)_0)=n^2$.
	    This is equivalent to the fact that for every fix $(S, \breve{B}) \neq (0,0)$, the first-order part of $f(tS, t \breve{B})=e^{itS} \cdot (H_0+t \breve{B}) \cdot e^{-itS}$ at $t =0$ is nonzero. That is, 
\begin{equation}\label{eq:erinto1}
  \left.  \frac{\mathrm{d}}{\mathrm{d}t} \left(e^{itS} \cdot (H_0+t \breve{B}) \cdot e^{-itS} \right) \right|_{t=0}=
   i [S, H_0]+ \breve{B} \neq 0.
\end{equation}
To show that $i[S, H_0]+ \breve{B} \neq 0$, consider the entry $a,b$ of the commutator:
\begin{equation}\label{eq:erinto2}
    [S, H_0]_{a,b}=\sum_{l=1}^n 
    (S_{a,l}H_{0;l,b}-H_{0;a,l}S_{l,b})=
    S_{a,b}(\lambda_b-\lambda_a).
\end{equation}
It shows that $[S, H_0]$ is an off-block matrix, and it is nonzero if $S \neq 0$. Indeed,  $S_{a,b} \neq 0$ can happen only if the index $a,b$ satisfies $a \leq k<b$ or $b\leq k<a$, hence in this case $\lambda_b-\lambda_a \neq 0$. Since $\breve{B}$ is block diagonal, for this $a,b$ index $\breve{B}_{a,b}=0$,  therefore it cannot cancel $[S, H_0]_{a,b} \neq 0$. That is, $(i[S, H_0]+\breve{B})_{a,b} \neq 0$, if $S_{a,b} \neq 0$, implying that $i[S, H_0]+\breve{B}=0$ only if $S=0$ and $\breve{B}=0$.  

By the analytic inverse function theorem, there is a neighborhood $\widetilde{\mathcal{W}}_0 $ of $(0,0) \in \R^{n^2}$ (in the $(S, \breve{B})$ space)  and $\mathcal{V}_0$ of $H_0 \in \mbox{Herm}(n)$ such that $f|_{\widetilde{\mathcal{W}}_0}: \widetilde{\mathcal{W}}_0 \to \mathcal{V}_0$ is a bijection, whose inverse is also analytic. This gives the unique decomposition if $H \in \mathcal{V}_0$ and $(\breve{B}, S) \in \widetilde{W}_0$, and the analytic dependence of $\breve{B}$ and $S$ on $H$.

Moreover, we have to show that the `lowest $k$ state property' is satisfied, that is, the first $k$ columns of $e^{iS}$ span the sum of the eigenspaces of $H$ corresponding to the lowest $k$ eigenvalues.  This is an additional property, which possibly requires the choice of smaller neighborhoods $\mathcal{V}_0$ and $\widetilde{W}_0$. Observe the following:
\begin{enumerate}[label=(\alph*)]
    \item The eigenvalues of $\widetilde{B}$ and $H$ are equal.
    \item The subspace spanned by the first $k$ columns of $e^{iS}$ is exactly the sum of the eigenspaces of the $k \times k$ block of $\widetilde{B}$.
\end{enumerate}
Therefore, the lowest $k$ state property is equivalent to the statement that $\lambda < \mu$ holds for each eigenvalue $\lambda$ of the $k \times k$ block of $\widetilde{B}$ and  $\mu$ of the $(n-k) \times (n-k)$ block of $\widetilde{B}$.
 The inequality $\lambda < \mu$ holds for $H=H_0$, and the eigenvalues are continuous functions of $H$.
 Therefore, the inequality, and hence, the lowest $k$ state property holds for every $H$ in a sufficiently small neighborhood of $H_0$.
 
\end{proof}

\begin{remark}
    Surprisingly, in our proof the \emph{continuous} behaviour of the eigenvalues implies  the \emph{analytic} behaviour of the sum of the eigenspaces corresponding to the lowest $k$ eigenvalues around $H_0 \in \Sigma_k$ (more precisely, the projector onto this eigenspace). Indeed, the dependence of $S$ on $H$ is analytic, and the sum of the eigenspaces is equal to the subspace spanned by the first $k$ columns of $e^{iS}$, which is deduced from the continuous behaviour of the eigenvalues, and it implies the analytic dependence of the projector.
\end{remark}

In the \hyperref[ss:Bravyi]{Appendix} we provide exact conditions for
the neighborhoods $\mathcal{V}_0$, $\mathcal{W}_0$ and $\mathcal{S}_0$ based on the results in \cite{BravyiSW}.

\begin{proof}[Proof of Corollary~\ref{co:locchart} on the local chart induced by the SW decomposition]
    
(a) The proof of Theorem~\ref{th:sw} shows that the correspondence $H \mapsto (x_i, y_j)$ is an analytic bijection with analytic inverse between the neighborhood $\mathcal{V}_0 \subset \mbox{Herm}(n)$ of $H_0$ and a neighborhood of $0 \in \R^{n^2}$. (b) Equation~\eqref{eq:locchart} is a straightforward consequence of the construction. Indeed, $y_j=0$ means that $H_{\textup{eff}}=0$, and in this case $H \in \Sigma_k$ by equation~\eqref{eq:koo}.

\end{proof}

\begin{remark}
    The existence of the SW chart fitting to $\Sigma_k$ in the sense of Corollary~\ref{co:locchart} gives an alternative proof for the Neumann--Wigner theorem, stating that $\Sigma_k$ is a submanifold of codimension $k^2-1$ (see Table~\ref{tab:neumann} for the sketch of the original proof).
\end{remark}

\begin{proof}[Proof of Theorem~\ref{th:proj} on the projection to $\Sigma_k$] 
Let $\mathcal{P}$ denote the subspace spanned by the first $k$ columns of $e^{iS}$, which is equal to the sum of the eigenspaces of the lowest $k$ eigenvalues of $H$. Let $\mathcal{P}^{\perp}$ denote its Hermitian complement, this is the subspace spanned by the last $n-k$ columns of $e^{iS}$, and it also agrees with  the sum of the eigenspaces of the highest $n-k$ eigenvalues of $H$. In particular, $\mathcal{P}$ and $\mathcal{P}^{\perp}$ are invariant subspaces of $H_{\Sigma}$ and $H_{\textup{proj}}$. We show that the restrictions of $H_{\textup{proj}}$ and $H_{\Sigma}$ to both $\mathcal{P}$ and $\mathcal{P}^{\perp}$ are equal.

Both $H_{\Sigma}$ and $H_{\textup{proj}}$ has $k$-fold degeneracy, and 
$\mathcal{P}$ is the eigenspace corresponding to the lowest  $k$ degenerate eigenvalues of $H_{\textup{proj}}$, and also $H_{\Sigma}$. Moreover the degenerate eigenvalues of these matrices are also equal. Indeed, the trace of the $k \times k$ block of $H_{\textup{proj}}$ is equal to $\sum_{i=1}^k \lambda_i$ (where $\lambda_i$ denotes the eigenvalues of $H$), since $H_{\textup{eff}}$ has trace 0. Hence the restrictions of $H_{\textup{proj}}$ and $H_{\Sigma}$ to their common degenerate eigensubspace $\mathcal{P}$ are equal.

On the other hand, their restrictions to $\mathcal{P}^{\perp}$ are equal to the restriction of $H$. Then $H_{\textup{proj}}$ and $H_{\Sigma}$ agree on both $\mathcal{P}$ and $\mathcal{P}^{\perp}$, hence $H_{\textup{proj}}=H_{\Sigma}$.
    
\end{proof}

\subsection{Energy splitting and the distance from $\Sigma_k$}\label{ss:dist}  In this subsection we prove Theorem~\ref{th:distance} on the distance from $\Sigma_k$. We start with its easy parts.

\begin{prop}\label{co:distHsigma}
    $d(H, H_{\Sigma})=\| H_{\textup{eff}} \|=\sqrt{k} \cdot D_k(H)$.
\end{prop}

\begin{proof}
    $d(H, H_{\Sigma})= \| H-H_{\Sigma} \| =\| e^{iS} \cdot H_{\textup{eff}} \cdot e^{-iS} \|= \| H_{\textup{eff}} \|$, the last equation follows from Lemma~\ref{le:uniact}. Applying it to Equation~\eqref{eq:hsigma} implies that $d(H, H_{\Sigma})=d(\Lambda, \Lambda_{\Sigma})$, which is obviously equal to 
    \begin{equation}
     \| \Lambda- \Lambda_{\Sigma} \| =   \sqrt{\sum_{i=1}^k (\lambda_i-\overline{\lambda})^2}=\sqrt{k} \cdot \sigma(\lambda_1, \dots, \lambda_k)=\sqrt{k} \cdot D_k(H),
    \end{equation}
    where the last equality defines $D_k(H)$ as the standard deviation of the lowest $k$ eigenvalues of $H$.
    \end{proof}

Next we show the harder part, namely, for every $H \in \mbox{Herm}(n) \setminus \Sigma_{[k,k+1]}$ the unique closest point of $\Sigma_k$ to $H$ is $H_{\Sigma}$. 
First of all, notice that $\Sigma_k$ is not a compact set, hence, theoretically it might happen that  $\Sigma_k$ does not have closest point to $H$, more precisely, the distance function $d_H: \mbox{Herm}(n) \to \R $ defined as $d_H(G)=d(H, G)$ does not have a minimum on $\Sigma_k$. Avoiding this possibility causes several complications, which are managed at the end of this subsection. But before this,  we first observe that a minimum point of the restricted function $d_H|_{\Sigma_k}$ is also critical point of it, therefore we have the following.

\begin{prop}\label{pr:criti}
    If $\mu$ is a  minimum value of $d_H$ on $\Sigma_k$, that is, $\mu=\min \{d(H, G) \ | \ G \in \Sigma_k \}$, and $K \in \Sigma_k$ satisfies $d(H, K)=\mu$, then the line $\{tK+(1-t)H \ | \ t \in \R \}$ is orthogonal to $\Sigma_k$ at $K$.
\end{prop}

\begin{proof}
   $d_H$ has a global minimum at $K$, in particular, it is a local minimum. Since $\Sigma_k$ is a smooth manifold, every local minimum of a smooth function is a critical point, meaning that the differential $(\mathrm{d} (d_H|_{\Sigma_k}))_K$ of $d_H|_{\Sigma_k}$ at $K$ is 0. It means that the gradient of $d_H$ at $K$ is orthogonal to $\Sigma_k$. Indeed, for any tangent vector $v \in T_K \Sigma_k$ of $\Sigma_k$ at $K$ the evaluation of the differential is
    \begin{equation}
(\mathrm{d} (d_H|_{\Sigma_k}))_K (v)=(\mathrm{d} (d_H))_K (v)= \langle \mbox{grad}_K (d_H), v \rangle,
    \end{equation}
where the first equality comes from the definition of the restriction, and second equality is definition of the gradient.
    Hence $(\mathrm{d} (d_H|_{\Sigma_k}))_K$ is zero for every $v \in T_K \Sigma_k$ if and only if $\mbox{grad}_K(d_H)$ is orthogonal to $T_K \Sigma_k$.

    On the other hand, since $d_H$ is the distance from $H$, the gradient $\mbox{grad}_K(d_H)$ is parallel to $K-H$, hence it is parallel to the line joining $H$ and $K$. Therefore, this line is orthogonal to the tangent space $T_K \Sigma_k$, hence to $\Sigma_k$ at $K$.

    \end{proof}

    From now on we look for all the lines through $H$ orthogonal to $\Sigma_k$. According to Proposition~\ref{pr:criti}, the intersection points of these lines with $\Sigma_k$ are the candidates for the closest point of $\Sigma_k$ to $H$.
Let 
    \begin{equation}
        L_H=\{t H_{\Sigma} + (1-t) H \ | \ t \in \R \} 
    \end{equation}
  be the line  joining $H$ and $H_{\Sigma}$.

\begin{prop}\label{pr:YHorthog}
    $L_H$ is orthogonal to $\Sigma_k$ at the intersection point $H_{\Sigma}$.
\end{prop}

Before the proof we highlight its essential step, the descripiton of the tangent space.

\begin{lem}\label{le:tangspace}
    Let $H_0 \in \Sigma_k$ be a diagonal matrix. Then the tangent space $T_{H_0} \Sigma_k$ of $\Sigma_k$ at $H_0$ consists of the Hermitian matrices whose upper-left $k \times k$ block is a scalar matrix. See Figure~\ref{fig:SWtangent}. 
\end{lem}

\begin{figure}
	\begin{center}		\includegraphics[width=0.5\columnwidth]{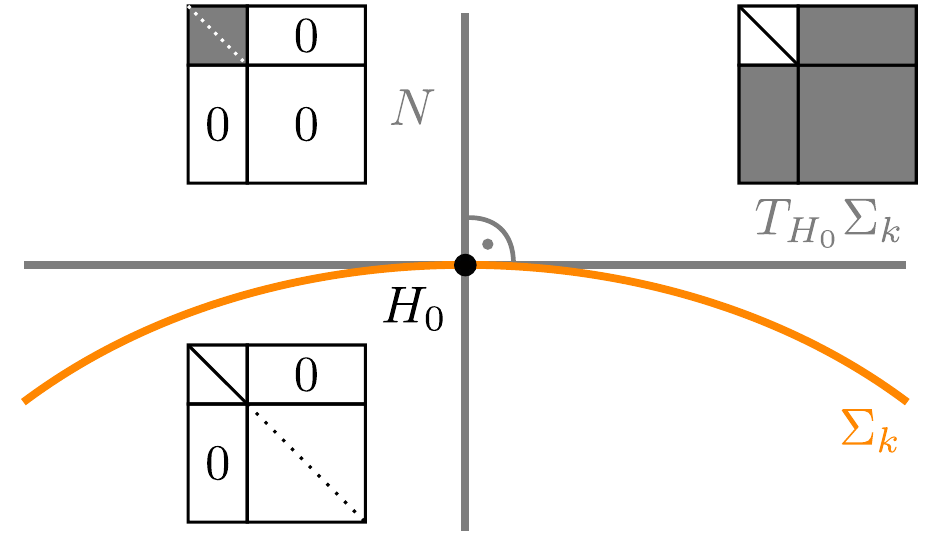}
	\end{center}
	\caption{The matrix form of the elements of the tangent space $T_{H_0}\Sigma_k$ and normal space $N$ of $\Sigma_k$ at a diagonal matrix $H_0 \in \Sigma_k$. By Lemma~\ref{le:tangspace}, the tangent space consists of matrices whose upper-left $k \times k$ block is a scalar matrix. The normal space, i.e. the orthogonal complement of the tangent space agrees with the space of effective Hamiltonians $\mbox{Herm}_0(k) \subset \mbox{Herm}(n)$: its elements have only a traceless upper left $k \times k$ block and zero entries everywhere else. \label{fig:SWtangent}}
\end{figure}

\begin{proof}[Proof of Lemma~\ref{le:tangspace}]
By the SW decomposition \eqref{eq:koo} and Corollary~\ref{co:locchart}, $\Sigma_k$ is locally given by the equation $H_{\textup{eff}}=0$, that is, $y_1=\dots =y_{k^2-1}=0$. Hence the tangent space $T_{H_0} \Sigma_k$ consists of the directions orthogonal to the $H_{\textup{eff}}$ directions, which includes everything in the $(n-k) \times (n-k)$ block, in the off-block, and the trace of the $k \times k$ block, proving the lemma.

\end{proof}

\begin{rem}\label{re:Sgenerates} A more precise analysis also highlights the role of the off-block form of the exponent $S$ in SW decomposition \eqref{eq:koo}: its variation changes the off-block elements of $H_0$ up to first order. Indeed, by Equation~\eqref{eq:erinto1}, every tangent vector can be written in form $i[S,H_0]+\breve{B}$ with an off-block $S$ and a block diagonal $\breve{B}$, where now $\breve{B}=B+T$, that is, $H_{\textup{eff}}=0$, since we are in $\Sigma_k$. By Equation~\eqref{eq:erinto2}, the entries of this tangent vector are $iS_{a,b}(\lambda_b-\lambda_a)+\breve{B}_{a,b}$, showing that $S$ generates the off-block elements, and $\breve{B}$ generates the $(n-k) \times (n-k)$ block and the trace of the $k \times k$ block. 

Moreover, one can prove Lemma~\ref{le:tangspace} without referring to the SW decomposition, but starting from an over-parametrization of $\Sigma_k$ around $H_0$ of the form $e^{iG} (H_0+ \breve{B})e^{-iG}$, where $\breve{B}=B+T$ as above, but now $G$ can be any element of $\mbox{Herm}(n)$. This leads to tangent vectors of the form $i[G,H_0]+\breve{B}$ with entries $iG_{a,b}(\lambda_b-\lambda_a)+\breve{B}_{a,b}$, which also shows that the off-block elements of a tangent vector depend on the off-block entries of $G$ and it can be arbitrary, since $\lambda_b-\lambda_a \neq 0$  between different blocks; the $k \times k$ block of $G$ is irrelevent, since $\lambda_b-\lambda_a=0$ if $a, b \leq k$; and the $(n-k) \times (n-k)$ block of $G$ contributes only to the $(n-k) \times (n-k)$ block of the tangent vector, which can be arbitrary by a choice of $\breve{B}$, as well as the trace of the $k \times k$ block.

\end{rem}

\begin{proof}[Proof of Proposition~\ref{pr:YHorthog}]  It is enough to show that $H -H_{\Sigma}$ is orthogonal to the tangent space $T_{H_{\Sigma}} \Sigma_k$ of $\Sigma_k$ at $H_{\Sigma}$. Because of Lemma~\ref{le:uniact} it is enough to show the orthogonality in the diagonal case, namely, $\Lambda-\Lambda_{\Sigma}$ is orthogonal to $T_{\Lambda_{\Sigma}} \Sigma_k$. But $\Lambda-\Lambda_{\Sigma}$ is a diagonal matrix with nonzero elements only in the $k \times k$ block, and its trace is 0. More precisely, the diagonal  elements are $\lambda_i-\overline{\lambda}$ ($i=1, \dots, k$) where $\lambda_i$ are lowest $k$ eigenvalues of $H$ and $\overline{\lambda}$ is their mean value. By Lemma~\ref{le:tangspace}, the $k \times k$ block of each tangent vector $K \in T_{\Lambda_{\Sigma}} \Sigma_k$ is a scalar matrix, we denote its diagonal entries by $c$. Then
\begin{equation}
\langle \Lambda-\Lambda_{\Sigma}, K \rangle =\sum_{i=1}^kc(\lambda_i-\overline{\lambda})=0,
\end{equation}
proving the proposition.
\end{proof}

\begin{rem}  If $H $ is sufficiently close to $\Sigma_k$, then the orthogonality of $L_H$ to $\Sigma_k$ implies that $H_{\Sigma}$ is the unique closest point of $\Sigma_k$ to $H$. This follows from the tubular neighborhood theorem \cite{foote1984regularity} (see also \cite[pg. 74, exercise 3.]{guillemin2010differential}): if a point $p \in \R^n$  is sufficiently close to a submanifold $N \subset \R^n$, then it has a unique closest point $p' \in N$, characterised by the orthogonality of the line joining $p$ and $p'$ to $N$. However, it is not enough for the global version of Theorem~\ref{th:distance}, that is, for every $H \in \mbox{Herm}(n) \setminus \Sigma_{[k,k+1]}$.
\end{rem}

In the following we characterise the lines orthogonal to $\Sigma_k$ at a point.

\begin{prop}\label{pr:YHorthog2}
    Let $L$ be a line in $\mbox{Herm}(n)$ through a point $H_0 \in \Sigma_k$. Then the following are equivalent:
    \begin{enumerate}
    
    \item $L=L_H$ for an element $H \in L \setminus \{H_0\}$ (in particular, $H_{\Sigma}=H_0$).

     \item $L$ is orthogonal to $\Sigma_k$ at $H_0$.

    \item $L$ can be parametrized as follows: Starting with any diagonalization $H_{0}=U \Lambda_{0} U^{-1}$ with increasing order of the eigenvalues, we choose a  $k \times k$ diagonal matrix $D$ of trace 0 with increasing order of the diagonal elements in the upper-left block, and take the parametrization
    \begin{equation}\label{eq:egyenes}
        t \mapsto U (\Lambda_{0} + t D) U^{-1}.
    \end{equation}
    
    \item $L=L_H$ for every $H \in L \setminus \{H_0 \}$ sufficiently close to $H_0$ (in particular, $H_{\Sigma}=H_0$).

    \end{enumerate}
\end{prop}

\begin{proof}
    (1) $\Rightarrow$ (2): Follows from Proposition~\ref{pr:YHorthog}. 
    
    (1) $\Rightarrow$ (3): Follows directly from the construction of $H_{\Sigma}$, by choosing $D=\Lambda-\Lambda_{0}$, with any choice of $U$ diagonalizing $H$ (hence, also $H_0$), where $\Lambda=U^{-1} H U$ is diagonal.

    (3) $\Rightarrow$ (4): Consider an element $G=U (\Lambda_{0} + t D) U^{-1} \in L$, then $G_{\Sigma}=H_{0}$ if the diagonal elements of $\Lambda_{\Sigma} + t D$ are still in increasing order, that is, until the line reaches the $\Sigma_{[k, k+1]}$ degeneracy stratum. This holds for sufficiently small $t$.

    (4) $\Rightarrow$ (1) is obvious. Until now we proved the equivalence of (1), (3) and (4), and any of them implies (2).

    (2) $\Rightarrow$ (3): It follows by counting the dimensions, cf. Table~\ref{tab:neumann} for a similar method. First note that (3) implies (2), that is, the parametrization \eqref{eq:egyenes} provides orthogonal lines to $\Sigma_k$ at $H_0$. Then we count how many dimensions can be covered by such a parametrization. The choice of $D$ up to a real scalar factor gives $k-1$ dimensions.  The choice of the lowest $k$ eigenspaces, that is, the first $k$ columns of $U$ up to $U(1)$ rotations gives $k^2-k$  dimensions. Hence, the dimension of the subspace of matrices which can be reached by a parametrization in form \eqref{eq:egyenes} has $k-1+k^2-k=k^2-1$ dimensions, therefore it covers the whole normal space of $\Sigma_k$ at $H_0$.

    The above implications imply the equivalence, proving the theorem.
    
\end{proof}

\begin{remark} By point (3) of the above proposition, the orthogonal lines to $\Sigma_k$ at $H_0$ can be described as `spreading the eigenvalues linearly'. 
    Although this construction seems to be insightful, its difficulty is that the direction of the line depends on both $D$  and the unitary matrix $U$ -- more precisely, on the subspaces spanned by its first $k$ columns. Another way to characterize the orthogonal lines via SW decomposition is `turning an effective Hamiltonian on': Starting from the unique SW decomposition \eqref{eq:koo0} of $H_0$ (with respect to a possibly different base point $H_0'$), we choose a $k \times k$ Hermitian matrix $H_{\textup{eff}}$ of trace 0 in the upper left block. The corresponding line is parametrized as
    \begin{equation}
        t \mapsto H_{0} + t \cdot e^{iS} \cdot H_{\textup{eff}} \cdot e^{-iS},
    \end{equation}
    which is orthogonal to $\Sigma_k$ at $H_0$, cf. Lemma~\ref{le:tangspace}.
\end{remark}

\begin{figure}
	\begin{center}		\includegraphics[width=0.9\columnwidth]{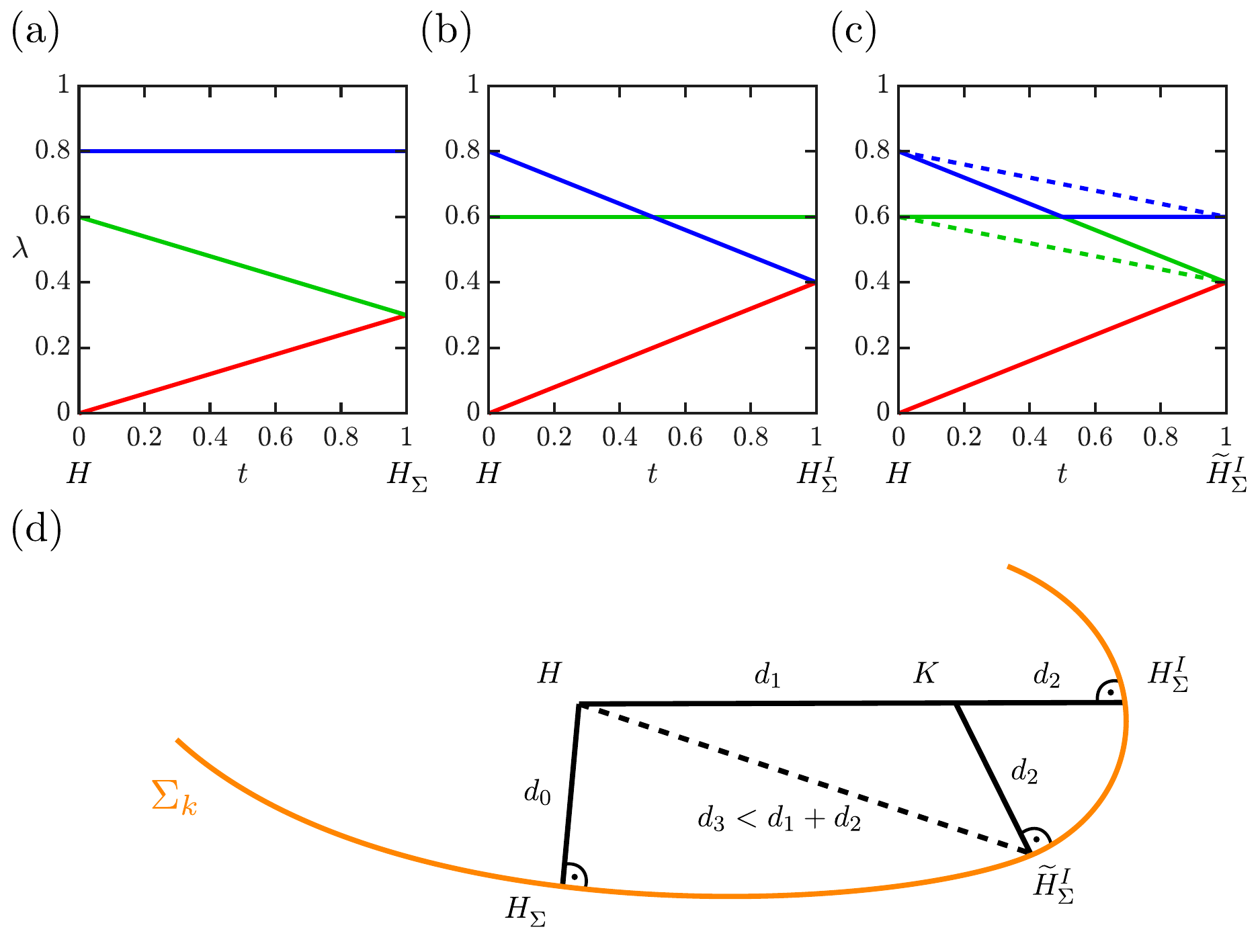}
	\end{center}
	\caption{The projections of $H$ onto $\Sigma_k$ for $n=3$, $k=2$, illustrating the proof of Proposition~\ref{pr:notclosest}.\\
 Panel $(a)$ ($(b)$ respectively):
 $H_{\Sigma}=H_{\Sigma}^{\{1,2\}}$ ($H_{\Sigma}^{I}=H_{\Sigma}^{\{1,3\}}$) is constructed by contracting the eigenvalues $\lambda_1$ and $\lambda_2$ ($\lambda_1$ and $\lambda_3$) of $H$ to their mean value. 
 The linear motion of the eigenvalues (red, green and blues line segments) realizes the line segments of $L_H$ ($L_H^{I}$), joining $H$ with $H_{\Sigma}$ ($H_{\Sigma}^{I}$).\\
Panel $(c)$: The crossing of the eigenvalues is resolved (see colors), providing another matrix $\widetilde{H}_{\Sigma}^{I} \in \Sigma_k$. The non-linear motion of the eigenvalues traces a broken-line (consisting of two line segments) joining $H$ and  $\widetilde{H}_{\Sigma}^{I}$. The vertex of the broken line is denoted by $K$. \\
Panel $(d)$: Illustration of the whole configuration in $\mbox{Herm}(n)$. By construction, the distances $d(K, H_{\Sigma}^{I})$ and  $d(K, \widetilde{H}_{\Sigma}^{I})$ are equal, it is denoted by $d_2$. With the notations $d_1=d(H, K)$ and $d_3=d(H, \widetilde{H}_{\Sigma}^{I})$ (dashed line), $d_3 < d_1+d_2$ holds by the triangle inequality, showing that  $\widetilde{H}_{\Sigma}^{I}$ is a closer point of $\Sigma_k$ to $H$ than $H_{\Sigma}^{I}$	\label{fig:eigproj}}
\end{figure}

 Next we characterise the lines through a point $H \in \mbox{Herm}(n) \setminus \Sigma_{[k,k+1]}$ which intersect $\Sigma_k$ orthogonally. $L_H$ is one of these lines. To find the others, first consider an example, a matrix $H \in \mbox{Herm}(3)$ with eigenvalues $\lambda_1 < \lambda_2 < \lambda_3$. To obtain its projection $H_{\Sigma}^{\{1,2\}} =H_{\Sigma} \in \Sigma_2$, we contract $\lambda_1 $ and $\lambda_2$ to their mean value $\overline{\lambda}^{\{1,2\}}=(\lambda_1+\lambda_2)/2$, then the line through $H_{\Sigma}$ and $H$ is $L_H^{\{1,2\}}=L_H$. 

Another possibility is the contraction of $\lambda_1$ and $\lambda_3$ to $\overline{\lambda}^{\{1,3\}}=(\lambda_1+\lambda_3)/2$. If $ \overline{\lambda}^{\{1,3\}} < \lambda_2$, then we obtain a point $H_{\Sigma}^{\{1,3\}}  \in \Sigma_2$. Consider the line 
\begin{equation}
    L_H^{\{1,3\}}=\{ t H_{\Sigma}^{\{1,3\}} + (1-t) H  \ | \ t \in \R \}
\end{equation} 
joining $H_{\Sigma}^{\{1,3\}}$ and $H$. It consists of the matrices
\begin{eqnarray}
  t H_{\Sigma}^{\{1,3\}} + (1-t) H &=&
  U \begin{pmatrix}
        \lambda_1+ \frac{t}{2} (\lambda_3-\lambda_1) & 0 & 0 \\
       0 &  \lambda_2 & 0 \\
       0 & 0 &  \lambda_3+ \frac{t}{2} (\lambda_1-\lambda_3)
    \end{pmatrix} U^{-1}\nonumber\\
    &=&
  U' \begin{pmatrix}
        \lambda_1+ \frac{t}{2} (\lambda_3-\lambda_1) & 0 & 0 \\
       0 &  \lambda_3+ \frac{t}{2} (\lambda_1-\lambda_3) & 0 \\
       0 & 0 &  \lambda_2
    \end{pmatrix} (U')^{-1},
\end{eqnarray}
where $U'$ is the product of $U$ with the transposition of the 2nd and 3th basis elements. The role of $U'$ is to satisfy the increasing order of the eigenvalues, if we want to have a conventional form.
Then, by Proposition~\ref{pr:YHorthog2}, $L_H^{\{1,3\}}$ is orthogonal to $\Sigma_2$ at $H_{\Sigma}^{\{1,3\}}$. However, it crosses the stratum $\Sigma_{(12)}$ of $\Sigma$ (corresponding to the degeneracy $\lambda_2=\lambda_3$) between $H_{\Sigma}^{\{1,3\}}$ and $H$, namely, for $t=2(\lambda_3 -\lambda_2)/(\lambda_3-\lambda_1)$. See Figure~\ref{fig:eigproj}.

The third possibility, the contraction of $\lambda_2$ and $\lambda_3$ to $\overline{\lambda}^{\{2,3\}}$ results a matrix $H_\Sigma^{\{2,3\}}$, which is not in $\Sigma_2$, since $\overline{\lambda}^{\{2,3\}} > \lambda_1$.

The same can be done in general, for an $H \in \mbox{Herm}(n) \setminus \Sigma_{[k,k+1]}$. Choose $k$ indices  $1 \leq i_1 < i_2 <  \dots < i_k \leq n$, and let $I=\{i_1, \dots, i_k \}$ denote their set. Let $H_{\Sigma}^I$ denote the matrix obtained from $H$ by replacing the eigenvalues $\lambda_{i_1}, \dots, \lambda_{i_k}$ with their mean $\overline{\lambda}^I =\sum_{j=1}^k \lambda_{i_j}/k$.  More precisely,  let $\Lambda_{\Sigma}^I$ be the diagonal matrix obtained from $\Lambda=U^{-1} H U= \mbox{diag}(\lambda_1, \dots, \lambda_n)$ by replacing $\lambda_i$ with $\overline{\lambda}^I$ for $i \in I$, and define 
\begin{equation}\label{eq:HSigma}
    H_{\Sigma}^I =U \cdot \Lambda_{\Sigma}^I \cdot U^{-1}.
\end{equation}
In general, $H_{\Sigma}^{I}$ can depend on the choice of $U$. In fact, if there is a degeneracy $\lambda_{i}=\lambda_{i'}$ of eigenvalues of $H$ with $i \in I$ and $i' \notin I$, then their separation depends on the choice of the diagonalization inside the degenerate eigenspace. For simplicity we omit the $U$ dependence from the notation, and $H_{\Sigma}^{I}$ denotes any of the possible choices.

If $\overline{\lambda}^I $ is smaller then the lowest omitted eigenvalue $\lambda_m$, where $m=\min (\{1, \dots, n\} \setminus I)$, then $H_{\Sigma}^I \in \Sigma_k$. Let 
\begin{equation}
    L_H^{I}=\{ t H_{\Sigma}^{I} + (1-t) H  \ | \ t \in \R \}
\end{equation} 
be the line joining $H_{\Sigma}^{I}$ and $H$ Obviously, $L_{H}^{I}$ crosses other strata of $\Sigma$ between $H$ and $H_{\Sigma}^{I}$. By Proposition~\ref{pr:YHorthog2}, the lines $L_{H}^{I}$ are orthogonal to $\Sigma_k$ at $H_{\Sigma}^{I}$, moreover:

\begin{prop}\label{pr:YorthogOnly}
 The only lines through $H \in \mbox{Herm}(n) \setminus \Sigma_{[k,k+1]} $ which are orthogonal to $\Sigma_k$ are the $L_H^I $ lines with $I \subset \{1, \dots, n \}$ satisfying  $H_{\Sigma}^I \in \Sigma_k$. 
\end{prop}

\begin{proof}
    Assume that a line $L$ through $H$ intersects $\Sigma_k$ orthogonally at $H_0 \in \Sigma_k$. By point (3) of Proposition~\ref{pr:YHorthog2}, $L$ has a special parametrization in form \eqref{eq:egyenes}. The linear motion of the degenerate eigenvalues of $H_0$ arrives at $k$ eigenvalues of $H$ corresponding to the indices $\{i_1, \dots, i_k\}=I$, and then, $H_0=H_{\Sigma}^{I}$ by the construction.
\end{proof}

Recall that our goal is to prove that $H_{\Sigma}$ is the closest point of $\Sigma_k$ to $H$. Towards this goal the next step is to show  that $H^{I}_{\Sigma}$ cannot be a closest point of $\Sigma_k$ to $H$ if $I \neq \{1, 2, \dots, k\}$. Recall that it does not imply directly that $H_{\Sigma}$ is the closest point, until we prove that there is a closest point, which will be the last step.  
    
\begin{prop}\label{pr:notclosest}
Consider an index set $I \neq \{1, 2, \dots, k\}$ and the corresponding point $H_{\Sigma}^{I} \in \Sigma_k$ (with a fixed unitary matrix $U$ diagonalizing $H$, cf. Equation~\eqref{eq:HSigma}). Then, there is a point $\widetilde{H}_{\Sigma}^{I} \in \Sigma_k$ such that $d (H, \widetilde{H}_{\Sigma}^{I}) < d(H, H_{\Sigma}^{I})$.
\end{prop}

\begin{proof}
We fix $U$ during the construction, hence we actually work in the space $\R^n$ of the diagonal entries, endowed with the usual inner product. The lines correspond to the linear motion of the diagonal entries. 

Consider two indices $i \in I $ and $i' \notin I$ with $i > i'$. Then  $\lambda_i > \lambda_{i'}$ holds for the eigenvalues of $H$, and $ \overline{\lambda}^{I} < \lambda_{i'}$, since $H^{I}_{\Sigma} \in \Sigma_k$. Define $\widetilde{H}_{\Sigma}^{I}$ as 
\begin{equation}
    \widetilde{H}_{\Sigma}^{I}=U \cdot \widetilde{\Lambda}_{\Sigma}^{I} \cdot U^{-1},
\end{equation}
where $\widetilde{\Lambda}_{\Sigma}^{I}$ is the diagonal matrix containing the same entries as ${\Lambda}_{\Sigma}^{I}$, but the $i$-th and $i'$-th eigenvalues swapped. Namely, the $i$-th diagonal element of $\widetilde{\Lambda}_{\Sigma}^{I}$ is $\lambda_{i'}$, and the $i'$-th diagonal element of $\widetilde{\Lambda}_{\Sigma}^{I}$ is $\overline{\lambda}^{I}$. See Figure~\ref{fig:eigproj} and Table~\ref{tab:HK}. 

\begin{table}[h]
    \setlength{\extrarowheight}{3pt}
	\begin{tabularx}{\textwidth}{|>{\raggedright\arraybackslash}X|>{\raggedright\arraybackslash}X|>{\raggedright\arraybackslash}X|>{\raggedright\arraybackslash}X|>{\raggedright\arraybackslash}X|}
	    \hline
		& $H$ & $H_{\Sigma}^{I}$ & $\widetilde{H}_{\Sigma}^{I}$ & $K$ \\
  \hline\hline
   $ j \in I$, $j \neq i$ & 
   $\lambda_j$ & 
  $\overline{\lambda}^{I}$ & 
  $\overline{\lambda}^{I}$ & 
  $t \overline{\lambda}^{I} + (1-t) \lambda_j$\\
  \hline
  $ j' \notin I$, $j' \neq i'$
  & $\lambda_{j'}$ &  $\lambda_{j'}$ &  $\lambda_{j'}$ &  $\lambda_{j'}$ \\
  \hline
  $i \ (\in I)$ & $\lambda_i$ & 
  $\overline{\lambda}^{I}$ & 
  $ \lambda_{i'} $ & 
  $ \lambda_{i'} $ \\
  \hline
  $i' \ (\notin I)$
  &  $\lambda_{i'} $ &  $\lambda_{i'}$ &  $\overline{\lambda}^{I}$ & $\lambda_{i'}$ \\
  \hline
  \end{tabularx}
	\caption{The diagonal entries of the diagonalizations of $H$, $H_{\Sigma}^{I}$,  $\widetilde{H}_{\Sigma}^{I}$ and $K$ in the proof of Proposition~\ref{pr:notclosest}. The four matrices are diagonalized by the same unitary matrix $U$, and the diagonal entries of $H$ are in increasing order. $H_{\Sigma}^{I}$ is constructed by replacing the eigenvalues of $H$ corresponding to the indices in $I$ with their mean value $\overline{\lambda}^{I}$. For two fixed indices $i \in I$ and $i' \notin I$ with $i > i'$, we swap the $i$-th and $i'$-th diagonal entries of ${H}_{\Sigma}^{I}$, obtaining $\widetilde{H}_{\Sigma}^{I}$. $K$ is defined as $K=t H_{\Sigma}^{I} + (1-t) H $ with $t \overline{\lambda}^{I} +(1-t) \lambda_i =\lambda_{i'}$. 
 \label{tab:HK}}
\end{table}

We  show that $d (H, \widetilde{H}_{\Sigma}^{I}) < d(H, H_{\Sigma}^{I})$. Along the line segment of $L_H^{I}$ joining $H$ and $H_{\Sigma}^{I}$, the eigenvalue $\lambda_i$ moves to $\overline{\lambda}^{I}$ linearly, and it crosses $\lambda_{i'}$ at a point $K \in L_H^{I}$. More precisely, $K=t H_{\Sigma}^{I} + (1-t) H $ with $t \overline{\lambda}^{I} +(1-t) \lambda_i =\lambda_{i'}$. Both the $i$-th and $i'$-th diagonal entry of $K$ is $\lambda_{i'}$, see Figure~\ref{fig:eigproj} and Table~\ref{tab:HK}. Then we have 
\begin{equation}
    d(K, \widetilde{H}_{\Sigma}^{I})=d(K, {H}_{\Sigma}^{I}).
\end{equation}
Indeed, the only difference in the distances might come from the $i$-th and $i'$-th diagonal entries (see Table~\ref{tab:HK}), which gives
\begin{equation}
    (d(K, \widetilde{H}_{\Sigma}^{I}))^2-(d(K, {H}_{\Sigma}^{I}))^2=[(\lambda_{i'}- \lambda_{i'})^2+(\lambda_{i'}-\overline{\lambda}^{I})^2]-
    [(\lambda_{i'}-\overline{\lambda}^{I})^2+(\lambda_{i'}- \lambda_{i'})^2]=0.
\end{equation}
Then, 
\begin{equation}
    d(H, H_{\Sigma}^{I})=d(H, K)+d(K, H_{\Sigma}^{I})=d(H, K)+d(K, \widetilde{H}_{\Sigma}^{I}) > d(H, \widetilde{H}_{\Sigma}^{I}),
\end{equation}
where the last inequality is the triangle inequality applied to the non-collinear points $H$, $K$ and $\widetilde{H}_{\Sigma}^{I}$, see Figure~\ref{fig:eigproj}. This completes the proof.

\end{proof}

Comparing Proposition 
\ref{pr:criti},  \ref{pr:YHorthog}, \ref{pr:YorthogOnly} and \ref{pr:notclosest}, we conclude the following.

\begin{cor}\label{co:majdnem}
    If the distance function $d_H$ has a global minimum $\mu$ on $\Sigma_k$, then $d(H, H_{\Sigma})=\mu$ and $H_{\Sigma}$ is the unique closest point of $\Sigma_k$ to $H$. In particular, $d(H, \Sigma_k)=d(H, H_{\Sigma})$ holds.
\end{cor}

In the following we show that $d_H$ has a global minimum on $\Sigma_k$. First we take a weaker observation.

\begin{prop}\label{pr:closure}
    The distance function $d_H$ has a global minimum on the closure $\mbox{cl}(\Sigma_k)$ of $\Sigma_k$ in $\mbox{Herm}(n)$.
\end{prop}

\begin{proof} Actually it is  a classical fact for any closed subset $A$ of $\R^n$ and a point $P \in \R^n$ that $A$ has a point with minimal distance from $P$. For completeness we write a proof for our particular situation, which works in general.

    Take a closed ball $\mathcal{B} \subset \mbox{Herm}(n)$ of radius $R$ centered at $H$, that is, $\mathcal{B}=\{ G \in \mbox{Herm}(n) \ | \ d(H, G) \leq R \}$. The radius $R$ has to be chosen such that $\mbox{cl}(\Sigma_k) \cap \mathcal{B}$ is non-empty, e.g., $R> d(H, H_{\Sigma})$ is good.   If $d_H$ has a global minimum on $\mbox{cl}(\Sigma_k) \cap \mathcal{B}$, then it is the global minimum on $\mbox{cl}(\Sigma_k)$ as well. 

    $\mbox{cl}(\Sigma_k) \cap \mathcal{B}$ is a closed and bounded subset of $\R^n$, hence it is compact, therefore, any continuous function on it has a global minimum, proving the proposition.
\end{proof}

\begin{prop}\label{pr:vanlegkozelebbi}
  $H_{\Sigma}$ is the unique closest point of $\Sigma_k$ to $H$. In particular, $d(H, \Sigma_k)=d(H, H_{\Sigma})$ holds.
\end{prop}

\begin{proof} 
    According to Proposition~\ref{pr:closure}, take a point $K \in \mbox{cl}(\Sigma_k)$ with minimal distance from $H$. Recall from Section~\ref{ss:preldegen} the disjoint decomposition $\mbox{cl}(\Sigma_k)=\Sigma_k \cup \Sigma_{k+1} \cup \dots \cup \Sigma_{n}$, hence, $K \in \Sigma_{k'}$ holds with a $k' \geq k$. Then, by applying Corollary~\ref{co:majdnem} to $\Sigma_{k'}$, it follows that $K$ is the projection of $H$ to $\Sigma_{k'}$ by contracting the lowest $k'$ eigenvalues of $H$ to their mean value. 

    Assume indirectly that $k' > k$. The triangle of vertices $H$, $H_{\Sigma}$, $K$  has a right angle at $H_{\Sigma}$. Indeed, by Proposition~\ref{pr:YHorthog}, the line $L_H$ joining $H$ and $H_{\Sigma}$ is orthogonal to $\Sigma_k$. Since the line joining $H_{\Sigma}$ and $K$ lies in $\Sigma_k$, it is orthogonal to $L_H$. See Figure~\ref{fig:pithagorean}. It follows that $d(H, H_{\Sigma}) < d(H,K)$, which contradicts with the premise that $K$ is the closest point of $\mbox{cl}(\Sigma_k)$ to $H$. 

    \begin{figure}
	\begin{center}		\includegraphics[width=0.4\columnwidth]{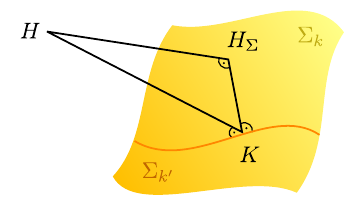}
	\end{center}
	\caption{The proof of Proposition~\ref{pr:vanlegkozelebbi}. Since $k'>k$, the closure of $\Sigma_k$ contains $\Sigma_{k'}$. The projections of $H$ to $\Sigma_k$ and $\Sigma_{k'}$, (obtained by contracting the lowest $k$ and the lowest $k'$ eigenvalues of $H$ to their mean value) is denoted by $H_{\Sigma}$ and $K$, respectively. $K$ is also the projection of $H_{\Sigma}$ to $\Sigma_{k'}$ in the same sense. Moreover, the line segment joining $H_{\Sigma}$ and $K$ lies in $\Sigma_k$, implying that it is orthogonal to the line segment joining $H$ and $H_{\Sigma}$, by Proposition~\ref{pr:YHorthog}. Therefore, the right triangle with vertices $H$, $H_{\Sigma}$ and $K$ shows that $d(H, H_{\Sigma}) < d(H, K)$. (The orthogonality of the line segment joining $H$ and $K$ to $\Sigma_{k'}$ and $H_{\Sigma}$, respectively, and $K$ to $\Sigma_{k'}$ is not used in the proof.)
 \label{fig:pithagorean}}
\end{figure}

    Therefore, $k'=k$ and $K=H_{\Sigma}$, which proves the proposition.
\end{proof}

\begin{proof}[Proof of Theorem~\ref{th:distance} on the distance from $\Sigma_k$] Proposition~\ref{pr:vanlegkozelebbi} and Corollary~\ref{co:distHsigma} prove Theorem~\ref{th:distance}.
    
\end{proof}

\begin{remark}
    The same argument holds for arbitrary (not necessarily ground state) $k$-fold degeneracy with the obvious modifications. The degeneracy set $\Sigma^{(a)}_k$ of matrices with $\lambda_{a} < \lambda_{a+1} = \lambda_{a+2} = \dots = \lambda_{a+k} < \lambda_{a+k+1}$ is a submanifold of $\mbox{Herm}(n)$, and if $\lambda_{a} < \lambda_{a+1} $ and $\lambda_{a+k} < \lambda_{a+k+1}$ holds for a matrix $H \in \mbox{Herm}(n)$, then $\Sigma^{(a)}_k$ has a unique closest element $H_{\Sigma}^{(a)}$ to $H$. The matrix $H_{\Sigma}^{(a)}$ is constructed from $H$ by replacing its eigenvalues $\lambda_{a+1}, \dots, \lambda_{a+k}$ with their mean value. The proof presented here can be generalized with the straightforward modifications.
\end{remark}

\begin{remark}
    If $H \in \Sigma_{[k, k+1]}$, then $H_{\Sigma}$ is not unique, it depends on the choice of the unitary matrix diagonalizing $H$, more precisely, on the subspace generated by the $k$-th column of $U$. However, $d(H, H_{\Sigma})$ is the same for every $H_{\Sigma}$, and it is the global minimum of the distance function $d_H$ on $\Sigma_k$. Therefore, $d(H, \Sigma_k)=d(H, H_{\Sigma})$ holds in this case with any choice of $H_{\Sigma}$.

\end{remark}

\subsection{Order of energy splitting of a degeneracy due to a perturbation}\label{ss:split}
In this subsection we prove Theorem~\ref{th:anal} on the analyticity of the eigenvalues, Theorem~\ref{th:ordspl}, Theorem~\ref{th:minden} and Corollary~\ref{co:tangstick} on the order of energy splitting and distancing. 

      Consider a one-parameter analytic perturbation $H: \R \to \mbox{Herm}(n)$, $H(0)=H_0 \in \Sigma_k$. The standard deviation $D_k(H(t))$ is not differentiable at 0 in general. We want to modify it slightly to obtain an analytic function $\mathfrak{s}_k(t)$ which agrees with $D_k(H(t))$ up to a sign for every $t$, and for $t \geq 0$ they are equal. For this we use the following observation:
    
\begin{lemma}\label{le:order}
    Let $f_1, \dots, f_l: \R \to \R $ be analytic functions defined in a neighborhood of 0, that is, locally convergent power series centered at 0. Let $r=\min_{1 \leq i \leq l} \{\mathrm{ord}_0 (f_i(t))  \}$ be the minimum of the orders. Then 
    \begin{equation}\label{eq:order}
         F(t):=  (\mbox{sgn}(t))^r \cdot \sqrt{\sum_{i=1}^l (f_i(t))^2} 
        \end{equation}
        is also analytic at 0, and its order $\mathrm{ord}_0 (F(t))$ is equal to $r$.
\end{lemma}

    \begin{proof}
        
        The proof can be done by factoring out the lowest degree term under the root sign. We can write \begin{equation}
            f_i(t)=t^{r_i} \cdot \widetilde{f}_i(t),
        \end{equation}
        where $r_i=\mathrm{ord}_0 (f_i)$ hence $\widetilde{f}_i$ is a power series with nonzero constant term, that is, $\widetilde{f}_i(0) \neq 0$. Taking $r:=\min \{ r_i\}$, we can write
         \begin{equation}\label{eq:order}
            F(t) =
           ( \mbox{sgn}(t))^r \cdot |t^r| \cdot  \sqrt{\sum_{i=1}^l t^{2r_i-2r} \cdot  (\widetilde{f}_i(t))^2}=:t^r \sqrt{\widetilde{F}(t)}.
        \end{equation}
        
        The expression under the square root (denoted by $\widetilde{F}(t)$) is a power series with nonzero constant term, that is, $\mathrm{ord}_0 (\widetilde{F}(t))=0$. Indeed, (1) all the exponents $2r_i-2r$ are bigger or equal to $0$, and (2) at least one of them is zero, moreover (3) the coefficients corresponding to the zero exponents --- that is, the constant terms of the corresponding $(\widetilde{f}_i(t))^2$ --- are positive, hence they cannot cancel each other. 

        Hence $\widetilde{F}(0) \neq 0$, which implies that $\sqrt{\widetilde{F}(t)}$ is also an analytic function (a locally convergent power series), and its rank is 0. Therefore $F(t)$ is analytic too. The order of  $F(t)$ is equal to $r$, the minimum of the orders of $f_i$, proving Lemma~\ref{le:order}.
    \end{proof}

By SW decomposition theorem~\ref{th:sw} the effective Hamiltonian $H_{\textup{eff}}$ depends on $H$ analytically, the real matrix elements $y_j(H(t))$ of the effective Hamiltonian $H_{\textup{eff}}(t)$ of $H(t)$ are analytic functions of $t$. Let $r$ be the minimum of the orders of $y_j(H(t))$. Define
      \begin{equation}\label{eq:signedszoras}
      \mathfrak{s}_k(t):= (\mbox{sgn}  (t))^r \cdot D_k(H(t)).
    \end{equation}
    Note that for $t>0$ this agrees with the function $\mathfrak{s}_k(t)$ defined in Equation~\ref{eq:orderofstddev}, indeed, the present version is its analytic  extension for every $t$.

    \begin{cor}\label{co:anal0}
$\mathfrak{s}_k(t)$ is an analytic function at $0$ of order $r$.
    \end{cor}

    \begin{proof}
       By Theorem~\ref{th:distance},
        \begin{equation}
            \sqrt{k} \cdot \mathfrak{s}_k(t)=(\mbox{sgn}(t))^r \cdot \| H_{\textup{eff}} (t) \|=
            (\mbox{sgn}(t))^r \cdot \sqrt{\sum_{j=1}^{k^2-1} (y_j(H(t)))^2 },
        \end{equation}
      and by Lemma~\ref{le:order}, the right hand side  is analytic and its order is $r$.
    \end{proof}

    In the following we want to compare the order of $\mathfrak{s}_k(t)$ with the order of the pairwise differences of the eigenvalues, but in general, these differences in the form $\lambda_i(t)-\lambda_j(t)$ are not differentiable at the origin, cf. Figure~\ref{fig:order}. Similarly to $\mathfrak{s}_k(t)$, we can modify the pairwise differences by a suitable reordering of the eigenvalues using Theorem~\ref{th:anal}, which is proved here in the following reformulated form.

 \begin{thm}[The eigenvalues are analytic]\label{th:anal2} Let $H(t)$ be an analytic one-parameter family of Hermitian matrices.
         There are analytic functions $\widetilde{\lambda}_1(t), \dots, \widetilde{\lambda}_n(t) $ such that for all $t$, the  eigenvalues of $H(t)$ are $\widetilde{\lambda}_i(t)$ ($i=1, \dots, n$). 
    
Equivalently, there is a permutation $i \mapsto i'$ of the indices $i \in \{1, \dots, n\}$ such that the functions
\begin{equation}\label{eq:analeig}
    \widetilde{\lambda}_i(t)= 
    \left\{ 
    \begin{array}{lcc}
    \lambda_{i}(t)  & \mbox{if} & t \geq 0, \\
\lambda_{i'}(t)  & \mbox{if} & t<0

\end{array}
    \right.
\end{equation}
are analytic (where $\lambda_i(t)$ denotes the eigenvalue functions in increasing order).
    \end{thm}

\begin{figure}
	\begin{center}		\includegraphics[width=0.7\columnwidth]{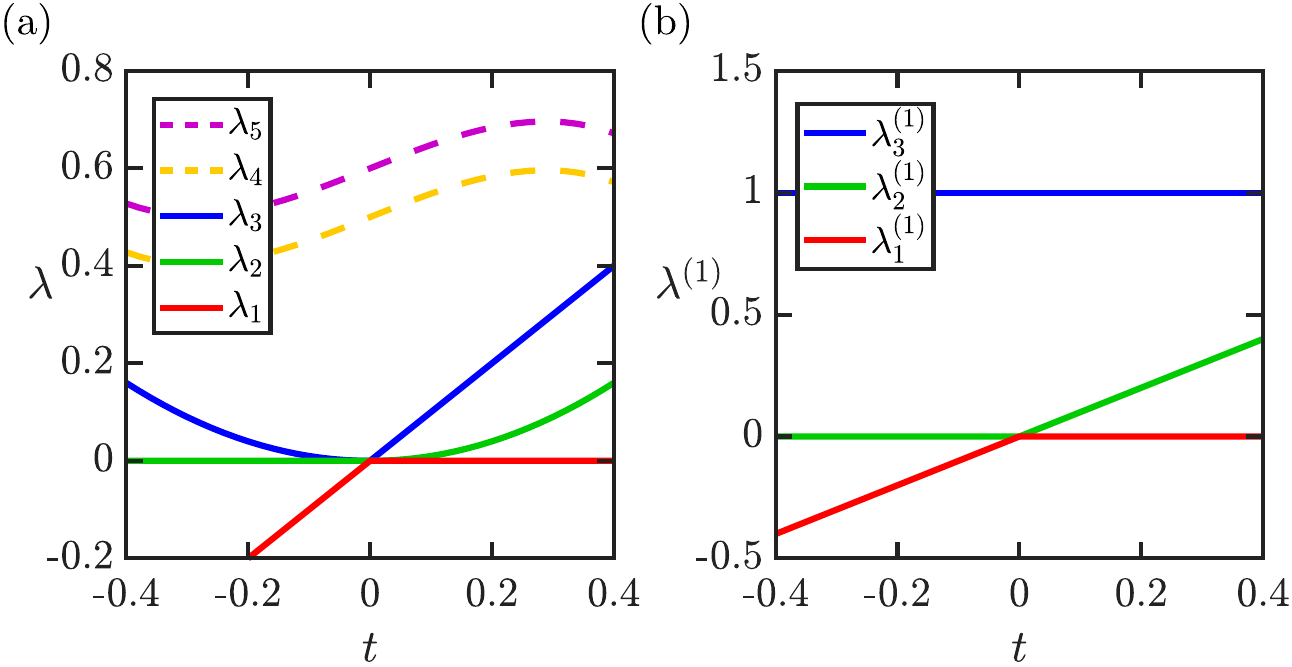}
	\end{center}
	\caption{Separating the degenerate eigenvalues in the proof of Theorem~\ref{th:anal2}. (a) The eigenvalues $\lambda_i$ of $H(t)$. The degenerate eigenvalues in increasing order are not analytic functions, however, after a suitable permutation on the negative part the eigenvalues form analytic functions. (b) To find the right permutation, we consider the eigenvalues $\lambda_i^{(1)}$ of $H_{\textup{eff}}(t)/t$. These  are `less degenerate' than the eigenvalues of $H$ in the sense that the order of the pairwise energy splittings decreases by one, in particular, the first-order intersections become non-degenerate.
	\label{fig:order}}
\end{figure}

 Although this theorem follows from a more general argument \cite[Theorem 1.10.]{kato2013perturbation}, we present here a new direct proof based on the exact SW decomposition Theorem~\ref{th:sw}. 
 The interesting case is when $H(0)$ is degenerate, e.g., $H(0) \in \Sigma_k$. 
 By Remark~\ref{re:root} below, the permutation is the identity for non-degenerate eigenvalues, i.e., $i'=i$ for a non-degenerate eigenvalue $\lambda_i$.

    \begin{rem}\label{re:root}
        The non-degenerate eigenvalues are always analytic functions of the parameters for arbitrary matrix families with arbitrary (finite) number of parameters. More precisely, consider a family $H(t)$ of complex matrices depending analytically on the parameter $t \in \R^m$ (or $t \in \C^m$) defined in a neighborhood of the origin. If $\lambda(0)$ is a non-degenerate eigenvalue of $H(0)$, then there is an analytic family $\lambda(t)$ such that $\lambda(t)$ is an eigenvalue of $H(t)$. Classically, it is proved via the following steps:
        \begin{itemize}
            \item The coefficients of the characteristic polynomial are analytic functions of the entries of the matrix $H(t)$, hence they are analytic functions of the parameter $t$.
            \item A root of a polynomial with multiplicity 1 depends analytically on the coefficients of the polynomial. It follows from the analytic implicit function theorem \cite[pg. 47, Thm. 2.5.1.]{krantz2002primer}, since a single root is not a root of the derivative of the polynomial.
        \end{itemize}
        However, for a Hermitian matrix family $H(t)$ (with $t \in \R^m$), the analytic dependence of a non-degenerate eigenvalue $\lambda$ can be deduced alternatively from the SW decomposition theorem \ref{th:sw}, more precisely, applying the same argument for $k=1$ as follows. Consider an $H_0$ with a non-degenerate eigenvalue $\lambda$. 
        For simplicity, assume that $H(0)$ is diagonal and $\lambda$ is in its $1\times 1$ upper-left corner (this can be reached by a unitary change of the basis of $\C^n$). The correspondence $f: (S,\breve{B}) \mapsto e^{iS} (H_0+\breve{B})e^{-iS}$ defines an analytic map from $\R^{n^2}$ to $\R^{n^2}$, and  its Jacobian has maximal rank at $(0,0)$, as one can show in 
        the same way as we proved Theorem \ref{th:sw} in Section~\ref{ss.prsw}. Hence, locally $f$ has an analytic inverse, and it provides a SW-type decomposition for the nearby matrices $H$. If $T$ is the analytically dependent  $1 \times 1$ upper-left  block of $\breve{B}$, then $T+\lambda$ is an eigenvalue of $H$, and it is equal to $\lambda$ for $H=H_0$. 
    \end{rem}

    \begin{rem}   

For the  eigenvalues degenerate at $t=0$, the situation is more complicated.
Their analytic behavior in an analytic \emph{one-parameter} family is a unique property of normal (in particular, Hermitian) matrices. This is based on the fact that they are unitarily diagonalizable. However, in general, the degenerate eigenvalues of $H(0)$ are not analytic functions of $t$ if $H(t)$ is a family of matrices with more than one parameter. Cf. \cite[Chapter Two]{kato2013perturbation}.
\end{rem}

\begin{proof}[Proof of Theorem~\ref{th:anal2}]
We prove the analiticity of eigenvalues with an iterated process. First, take the eigenvalue $\lambda_i(0)$ of $H(0)$. There are two possibilities.

\emph{Case 1: }  If the eigenvalue $\lambda_i(0)$ of $H(0)$ is non-degenerate, then Remark~\ref{re:root}  implies that 
\begin{equation}
            \widetilde{\lambda}_i(t):=\lambda_i(t)
        \end{equation} is an analytic function of $t$ around $t=0$.

\emph{Case 2: } If the eigenvalues $\lambda_{j+1}(0)=\lambda_{j+2}(0)=\dots = \lambda_{j+k}(0)$ of $H(0)$  are $k$-fold degenerate, then we can apply the SW decomposition to the $k$-fold degeneracy. For simplicity, assume that $j=0$, i.e., $H(0) \in \Sigma_k$. By Theorem~\ref{th:sw}, the traceless effective Hamiltonian $H_{\textup{eff}} (t)$ of $H(t)$ is an analytic function of $t$. Each $H_{\textup{eff}} (t)$ is  a Hermitian matrix with zero trace, and  $H_{\textup{eff}}(0)=0$, since $H(0)  \in \Sigma_k$. Hence $H_{\textup{eff}} (t)$ can be expressed as
        \begin{equation}
            H_{\textup{eff}} (t)=t H_1 + t^2 H_2 + t^3 H_3 + \dots,
        \end{equation}
        with traceless Hermitian matrices $H_j \in \mbox{Herm}_0(k)$.

        Define 
        \begin{equation}
            H^{(1)}(t)=\frac{H_{\textup{eff}}(t)}{t}= H_1 + t H_2 + t^2 H_3 + \dots,
        \end{equation}
        which is an analytic function of $t$. If we show that the eigenvalues of $H^{(1)}(t)$ can be expressed as analytic functions, then that implies this property for $H(t)$ as well. Indeed, let $\widetilde{\lambda}^{(1)}_i(t)$ be the  analytic eigenvalues of $H^{(1)}(t)$, indexed in increasing order for $t>0$. 
        Then the eigenvalues of $H_{\textup{eff}} (t)$ are $t \cdot \widetilde{\lambda}_{i}^{(1)}(t)$, and we can define
     \begin{equation}\label{eq:analeig}
    \widetilde{\lambda}_{i}(t)= \lambda_1(0)+
    \overline{\lambda}(t)+t \cdot  \widetilde{\lambda}_{i}^{(1)}(t),  
\end{equation}
where $\lambda_1(0)$ is the degenerate eigenvalue of $H(0)$, and $\overline{\lambda}(t)$ is the diagonal entry of the scalar matrix $T(t)$ in Equation~\eqref{eq:koo}, hence it is analytic. The functions $ \widetilde{\lambda}_{i}(t)$ are analytic, and one can verify that $\widetilde{\lambda}_{i}(t)$ is an eigenvalue of $H(t)$ in form of equation~\eqref{eq:analeig} (that is,  indexed in increasing order for $t>0$). 

Therefore, the proof is ready in the special case if $H_1$ is non-degenerate. In fact, in this case, the non-degenerate eigenvalues $\widetilde{\lambda}^{(1)}_i(t)=\lambda^{(1)}_i(t)$ are analytic functions (by case 1), and they are in increasing order also for $t\leq 0$. Since the multiplication by a negative $t$ reverses their order, in this case we have 
\begin{equation}\label{eq:permutation}
    \widetilde{\lambda}_{i}(t)= 
    \left\{ 
    \begin{array}{lcc}
    \lambda_{i}(t),  & \mbox{if} & t \geq 0 ,\\
\lambda_{k+1-i}(t),  & \mbox{if} & t<0 ,
\end{array}
    \right.
\end{equation}
for $i=1, \dots, k$. In this particular case the permutation is $i'=k+1-i$ for $i=1, \dots, k$ and $i'=i$ for $i=k+1, \dots, n$.

If $H(0)$ has an arbitrary (not ground state) $k$-fold degeneracy, then the straight-forward generalization of the SW decomposition can be used.

If $H^{(1)}(0)=H_1$ is still degenerate, that is, some of the eigenvalues $\lambda_{i}^{(1)}(0)$ coincide, we iterate the previous steps with $H^{(1)}(t)$ in place of $H(t)$. We create the traceless effective Hamiltonian $H_{\textup{eff}}^{(1)} (t)$ for each degenerate subspace, and since it is divisible by $t$, we can define the analytic family $H^{(2)}(t):=H^{(1)}_{\textup{eff}}(t)/t$.  If an eigenvalue $\lambda_{i}^{(2)}(0)$ of $H^{(2)}(0)$ is non-degenerate (case 1), then the corresponding eigenvalue of $H^{(1)}(t)$ and $H(t)$ is an analytic function of $t$. 
         
For a degenerate eigenvalue $\lambda_{i}^{(2)}(0)$ we iterate the above steps obtaining the analytic matrix families $H^{(j)}(t)$ and the eigenvalues $\lambda_{i}^{(j)}(t)$ (where $j=1, 2, 3, \dots$). If an eigenvalue becomes non-degenerate in finite steps, that is, there is a $j$ such that $\lambda_{i}^{(j)}(0)$ is a non-degenerate eigenvalue of $H^{(j)}(0)$, then by induction the corresponding eigenvalue of $H(t)$ is an analytic function of $t$. 
         
\emph{Case 3: } Some eigenvalues may not split in finite steps. More precisely, for some index $i$ we have $\lambda_{i}^{(j)} (0)=\lambda_{i+1}^{(j)} (0) = \dots =\lambda_{i+l-1}^{(j)} (0)$ holds for every step $j$. In this case consider the first $j=j_0$ such that the other eigenvalues of $H^{(j_0)}(0)$ are different from them. Then for every $j > j_0$ we have $H_{\textup{eff}}^{(j)}(t)=H^{(j)}(t)$, hence $H^{(j+1)}(t)=H^{(j)}(t)/t$. Thus (each entry of) $H^{(j)}(t)$ is divisible by an arbitrary power of $t$, in other words its order is infinite, implying that $H^{(j)}(t)=0$ identically. Hence the eigenvalues of $H^{(j)}(t)=0$ are analytic (they are constant zero functions), therefore the corresponding eigenvalues of $H(t)$ are analytic.
    \end{proof}

\begin{rem}
    The permutation of the indices can be followed step by step. As in Equation~\eqref{eq:permutation}, the  multiplication by a negative $t$ reverses
the order of the eigenvalues in each step, and the final permutation $i \mapsto i'$ is the product of these permutations.
\end{rem}

For $0 < i< j \leq k$ define the pairwise splitting of the eigenvalues in analytic way as
    \begin{equation}
        \mathfrak{s}_{i,j} (t)=\widetilde{\lambda}_i(t)-\widetilde{\lambda}_j(t).
    \end{equation}

    The proof of Theorem~\ref{th:ordspl}  on the order of energy splitting is based on the following lemma.

    \begin{lem}\label{le:mean}
        Let $f_1, \dots, f_l: \R \to \R $ be analytic functions defined in a neighborhood of 0. Let $\overline{f}=\sum_{i=1}^l f_i/l$ be their mean. Then 
        \begin{equation}\label{eq:mean}
                       \min_{1 \leq i, j \leq l} \{\mathrm{ord}_0 (f_i-f_j) \} =  \min_{1 \leq i \leq l-1} \{\mathrm{ord}_0 (f_i-f_{i+1}) \} = \min_{1 \leq i \leq l} \{\mathrm{ord}_0 (f_i- \overline{f}) \} .
        \end{equation}
    \end{lem}

    \begin{proof}[Proof of Lemma~\ref{le:mean}]

    Let
    \[ f_i(t)= \sum_{r=0}^{\infty} a_{i,r} t^r \mbox{, and } \ \overline{f}(t)= \sum_{r=0}^{\infty} \overline{a}_r t^r\]
    be the Taylor series of $f_i$ and $\overline{f}$, respectively.

    For a fixed integer $r$ the following are equivalent:
    \begin{enumerate}
    \item Not all the coefficients $a_{i,r}$  are equal ($1 \leq i \leq l$).
        \item There are indices $1 \leq i, j \leq l$ such that $a_{i,r} \neq a_{j,r}$.
        \item There is an index $1 \leq i \leq l-1$ such that $a_{i,r} \neq a_{i+1,r}$.
        \item There is an index $1 \leq i \leq l$ such that  $a_{i,r} \neq \overline{a}_r $.
    \end{enumerate} 

    By definition, the left hand side of Equation~\eqref{eq:mean} is the smallest integer $r$ for which (2) holds, the middle part is the smallest integer $r$ for which (3) holds, and the right hand side is the smallest integer such that (4) holds. This proves the lemma.

    \end{proof}

 \begin{proof}[Proof of Theorem~\ref{th:ordspl} and Theorem~\ref{th:minden} on the order of energy splitting and distancing] 

 By Lemma~\ref{le:mean} we get
\begin{equation}
    \min_{ 1 \leq i, j \leq k} \{ \mathrm{ord}_0(\widetilde{\lambda}_i(t)-\widetilde{\lambda}_j(t))  \}=
\min_{ 1 \leq i \leq k-1} \{ \mathrm{ord}_0(\widetilde{\lambda}_i(t)-\widetilde{\lambda}_{i+1}(t))
    =\min_{ 1 \leq i \leq k} \{ \mathrm{ord}_0(\widetilde{\lambda}_i(t)-\overline{\lambda}(t)) \},
\end{equation}
that is, 
\begin{equation}
\min_{ 1\leq i < j \leq k} \{ \mathrm{ord}_0 (\mathfrak{s}_{i,j})\}=
\min_{ 1 \leq i \leq k-1} \{ \mathrm{ord}_0 (\mathfrak{s}_{i,i+1})\}
=\min_{ 1 \leq i \leq k} \{ \mathrm{ord}_0 (\overline{\mathfrak{s}}_{i})\}.
\end{equation}
Let $r$ denote this minimum. In particular, considering only the right side, $r$ is the minimum of the orders of the functions $\overline{\mathfrak{s}}_i=\widetilde{\lambda}_i(t)-\overline{\lambda}(t)$. Applying Lemma~\ref{le:order} to these functions (i.e., $f_i=\overline{\mathfrak{s}}_i$) gives that
      \begin{equation}\label{eq:ujra}
         t \mapsto  \frac{ (\mbox{sgn}(t))^r }{\sqrt{k}} \cdot \sqrt{\sum_{i=1}^k (\widetilde{\lambda}_i(t)-\overline{\lambda}(t))^2 }
      \end{equation}
      is an analytic function of order $r$. But this function is equal to $\mathfrak{s}_k(t)$ (defined by Equation~\eqref{eq:signedszoras}).
      This shows that $\mathrm{ord}_0 (\mathfrak{s}_k)=\min_{ 1 \leq i \leq k} \{ \mathrm{ord}_0 (\overline{\mathfrak{s}}_{i})\}=r$. Together with Corollary~\ref{co:anal0}, it also shows that $r$ agrees with the order of the effective Hamiltonian, i.e., the last two expressions of Equation~\eqref{eq:ord22}.

    Next we show that $\mathrm{ord}_0 (\mathfrak{s}_{1,k}) =r$. Clearly, \begin{equation}\mathrm{ord}_0 (\mathfrak{s}_{1,k}) \geq \min_{ 1\leq i < j \leq k} \{ \mathrm{ord}_0 (\mathfrak{s}_{i,j})\} \end{equation} 
    holds by definition.
    Assume indirectly that 
    the inequality is strict. Take $i, j$ such that $\mathrm{ord}_0 (\mathfrak{s}_{i,j}) =r$. Then there is an $\epsilon > 0$ such that $|\mathfrak{s}_{i,j}(t)| > |\mathfrak{s}_{1,k}(t)|$ holds for $0 < |t| < \epsilon$. This is a contradiction, since $|\mathfrak{s}_{i,j}(t)| < |\mathfrak{s}_{1,k}(t)|$ holds for sufficiently small positive $t$, since the eigenvalues are in increasing order. This proves the theorem.
      \end{proof}

   \begin{proof}[Proof of Corollary~\ref{co:tangstick} on the order of energy splitting in linear families] Consider the SW chart (Corollary~\ref{co:locchart}) on the neighborhood $\mathcal{V}_0$ of $H_0$ in $\mbox{Herm}(n)$, with coordinates $\phi(H)=(x,y) \in \R^{n^2-k^2+1} \times \R^{k^2-1}$. Consider the map $\phi_y=\mathrm{pr}_2 \circ \phi:  \mathcal{V}_0 \to \R^{k^2-1}$, that is, $\phi_y(H)=y$. Since $\Sigma_k \cap \mathcal{V}_0=\phi_y^{-1}(0)$, its tangent space at $H_0$  is $T_{H_0} \Sigma_k=\ker ((\mathrm{d} \phi_y)_{H_0})$. Then, $H_1 \in T_{H_0} \Sigma_k$ if and only if 
   \begin{equation} 0=(\mathrm{d} \phi_y)_{H_0} (H_1)=\left. \frac{\mathrm{d}}{\mathrm{d}t} (\phi_y(H_0+tH_1)) \right|_{t=0}, \end{equation}
   which is equivalent to the fact that the order of every component $(\phi_y)_j(H_0+tH_1)=y_j(H_0+tH_1)$ is bigger than 1, that is, 
   \begin{equation}\min_{1 \leq j \leq k^2-1} \{\mbox{ord}_0 (t \mapsto y_j(H_0+tH_1)) \}
   >1 \end{equation}
   and the left side is equal to the order of energy splitting $r$ by Theorem~\ref{th:minden}. This proves the corollary.
       
   \end{proof}

      \subsection{Parameter-dependent quantum systems and Weyl points}

      \begin{proof}[Proof of Theorem~\ref{th:charweyl} on the characterization of Weyl points]
          The equivalence of (1) and (2) is a well-known property of transversality, see e.g. \cite[pg. 28]{guillemin2010differential} or \cite[Lemma 4.3.]{golubitsky2012stable}. To see it in our particular situation, consider a chart on $M^3$ around $p_0=0$, and consider the SW chart (Corollary~\ref{co:locchart}) on the neighborhood $\mathcal{V}_0$ of $H_0$ in $\mbox{Herm}(n)$, with coordinates $\phi(H)=(x, y) \in \R^{n^2-3} \times \R^{3}$. Let $\phi_y: \mathcal{V}_0 \to \R^3$ denote the map $\phi_y (H)=(y_1, y_2, y_3)$. Obviously, $h=\phi_y \circ H$, and the tangent space of $\Sigma_2$ at $H_0$ is $T_{H_0} \Sigma_2=\ker ((\mathrm{d} \phi_y)_{H_0})$. The transversality of $H$ to $\Sigma_2$ at $H(p_0)=H_0$ means that
          \begin{equation}
              T_{H_0} \Sigma_2 + (\mathrm{d} H)_{p_0}(T_{p_0} M^3)=T_{H_0} \mbox{Herm}(n), 
          \end{equation}
           or equivalently,
          \begin{equation}
              \ker ((\mathrm{d} \phi_y)_{H_0}) + (\mathrm{d} H)_{p_0}(T_{p_0} M^3)=T_{H_0} \mbox{Herm}(n).
          \end{equation}
          (Recall the definitions from Section~\ref{sec:weylrobustness}.)
          By applying $(\mathrm{d}\phi_y)_{H_0}$  on both sides, we get
          \begin{equation}
              (\mathrm{d}\phi_y)_{H_0}  \circ (\mathrm{d} H)_{p_0}(T_{p_0} M^3)=(\mathrm{d}\phi_y)_{H_0} (T_{H_0} \mbox{Herm}(n)).
          \end{equation}
          On the left side, $(\mathrm{d}\phi_y)_{H_0}  \circ (\mathrm{d} H)_{p_0}=(\mathrm{d} h)_{p_0}$ by the chain rule, and the right side is equal to $T_0 \R^3$, since $(\mathrm{d}\phi_y)_{H_0}$ is surjective. Therefore, the transversality is equivalent to the fact that $(\mathrm{d} \phi_y)_{H_0}$ is surjective on the the image of $(\mathrm{d}H)_{p_0}$, which is equivalent to the fact that $(\mathrm{d} h)_{p_0}$ has maximal rank 3. This proves (1) $\Leftrightarrow$ (2).

         Point (3) is clearly equivalent to (2), since it is equivalent to the following: for every $\gamma$ with $\gamma(0)=p_0$, $\gamma'(0) \neq 0$, $\gamma'(0) \notin \ker ((\mathrm{d} h)_{p_0})$. Point (3) is also equivalent to (4), since, by Theorem~\ref{th:minden}, the order of energy splitting is equal to the order of $h(\gamma(t))$, which is 1 by (3).
      \end{proof}

      \begin{proof}[Proof of Corollary~\ref{co:weyliso} on Weyl points being isolated]  

      First we show that $p_0$ has a neighborhood $\widetilde{\mathcal{W}}_0$ in $M$ such that the restriction $H|_{\widetilde{\mathcal{W}}_0}$ is transverse to $\Sigma_2$. The differential of $h$ has maximal rank at $p_0$, that is, its determinant is non-zero. But the determinant $p \mapsto \det((\mathrm{d} h)_p)$ is a continuous map, hence $p_0$ has a neighborhood $\widetilde{\mathcal{W}}_0$ such that $\det((\mathrm{d} h)_p) \neq 0$ if $p \in \widetilde{\mathcal{W}}_0$, that is, the rank of the differential of $h$ has maximal rank at every point $p \in \widetilde{\mathcal{W}}_0$.  Then, by point (2) of Theorem~\ref{th:charweyl}, the restriction $H|_{\widetilde{\mathcal{W}}_0}: \widetilde{\mathcal{W}}_0 \to \mbox{Herm}(n)$ is transverse to $\Sigma_2$, i.e., at every $p \in \widetilde{\mathcal{W}}_0$, either $H$ is transverse to $\Sigma_2$ at $H(p) \in \Sigma_2$, or $H(p) \notin \Sigma_2$. Then, by the theorem in \cite[pg. 30]{guillemin2010differential} (see also \cite[Thm. 4.4.]{golubitsky2012stable}), $H^{-1}(\Sigma_2) \cap \widetilde{\mathcal{W}}_0$ is a submanifold of dimension $\dim(M^3)-\mbox{codim}(\Sigma_2)=0$ that is, it consists only of isolated points, proving the corollary.

      \end{proof}
      
      \begin{proof}[Proof of Corollary \ref{co:weylgene} on Weyl points being stable and generic]
Part (a) is essentially the stability theorem \cite[35]{guillemin2010differential}, see also \cite[pg. 59, exercise (1) (a)]{golubitsky2012stable}. We formulate the proof in our particular situation.
      
The perturbation $H_t$ induces a perturbation $h_t$ of $h_{t=0}=h$, that is, a smooth map $(p,t) \mapsto h_t(p)$ defined in a neighborhood of $(p_0,0)$. The determinant map $(p,t) \mapsto \det((\mathrm{d} h_t)_p)$ is continuous, hence, it is non-zero in a neighborhood of $(p_0, 0)$. It implies that there is an $\epsilon_1>0$ such that for a fixed $|t|<\epsilon_1$ the map $p \mapsto H_t(p)$ is transverse to $\Sigma_2$ in a sufficiently small neighborhood of $p_0$ in $M$. Therefore, the map $(p, t) \mapsto H_t(p)$ is also transverse to $\Sigma_2$, hence the preimage of $\Sigma_2$ is a smooth manifold $\Gamma$ in a neighborhood of $(p,0)$ in $M \times \R$ of dimension $\dim(\Gamma)=\dim(M \times \R) - \mbox{codim}(\Sigma_2)=1$. Let $\Gamma_0$ be the component of $\Gamma$ containing $(p_0, 0)$. Applying the projection $(p,t) \mapsto p$ to $\Gamma_0$ gives the $\mathcal{C}^{\infty}$ curve $\gamma: (-\epsilon_1, \epsilon_1) \to M$ through $\gamma(0)=p_0$, satisfying that $\gamma(t) \in H_t^{-1}(\Sigma_2)$ is a Weyl point of $H_t$.

Other degeneracy points can be avoided by a further restriction of $t$, i.e., $|t| < \epsilon$ with a suitable $0 <\epsilon < \epsilon_1$, defined as follows. Choose a neighborhood $\mathcal{W_{\gamma}}$ of the image of $\gamma$ in $\mathcal{W}_0$ which does not contain other degeneracy points, that is, $H_t^{-1}(\Sigma_2) \cap \mathcal{W_{\gamma}}= \{ \gamma(t) \} $ for all $-\epsilon_1 <t <\epsilon_1$. Consider the compact set $\mbox{cl}(\mathcal{W}_0) \setminus \mathcal{W}_{\gamma}$. 
Observe that there is an $0 <\epsilon <\epsilon_1$ such that $H_t(p) \notin \Sigma_2$ for $p \in \mbox{cl}(\mathcal{W}_0) \setminus \mathcal{W}_{\gamma}$ and $0 <|t| <\epsilon$. Indeed, otherwise there would be a series $(p_i, t_i)$ such that $t_i$ converges to $0$ and $H_{t_i}(p_i) \in \Sigma_2$, and the limit of a convergent subseries of $p_i$ would be a point $q \in \mbox{cl}(\mathcal{W}_0) \setminus \mathcal{W}_{\gamma}$ with $H(q) \in \Sigma_2$. This proves point (a).

Part (b) follows from the transversality theorem \cite[pg 68-69]{guillemin2010differential}, see also \cite[Lemma 4.6.]{golubitsky2012stable}.
Consider a perturbation parametrized by the Hermitian matrices $M \times \mbox{Herm}(n) \to \mbox{Herm}(n)$, defined as $(p, K) \mapsto H_K(p):=H(p)+K$. This map is transverse to $\Sigma_2$, indeed, for a fixed $p$ it is a translation of $\mbox{Herm}(n)$. By the transversality theorem, those parameter values $K$ for which the map $p \mapsto H_K(p)$ is transverse to $\Sigma_2$ form a dense subset in $\mbox{Herm}(n)$. This proves the theorem.

      \end{proof}

      \begin{rem}\label{re:altalanosabb}
In Corollary~\ref{co:weylgene} we formulated the protected nature of the Weyl points in a way which can be deduced from the properties of transversality. The goal is to demonstrate the power of this approach in context of the degeneracy points. However, there are several possible generalizations, whose rigorous proofs are obstructed by technical difficulties.
          \begin{enumerate}
          \item In part (b), one may expect a stronger result, namely, the existence of a one-parameter perturbation $H_t$ of $H$ which has only Weyl points for $0<t< \epsilon$. This does not follow from the transversality theorems appearing in the literature. Indeed, the set of the `wrong' perturbations (in the particular form $H_K(p)=H(p)+K$, the set of those $K$ matrices for which $H_K$ is non-transverse) might be very complicated in general, although its complement is a dense set. If the map germ of $h$ at $p_0$ is `sufficiently nice', then the stronger statement holds. See the Appendix of \cite{Pinter2022} for related results and examples.

          \item For a sufficiently small $\epsilon$, the Weyl points of $H_K$ in $\mathcal{W}_0$ can be regarded as the Weyl points born from the non-generic degeneracy point $p_0$ of $H$. If there is a one-parameter perturbation $H_t$ with only Weyl points, then these Weyl points converge to $p_0$, as $t$ tends to 0, and they can be separeted from other Weyl points similarly as in the proof of part (a) of Corollary~\ref{co:weylgene}. In general it is harder to formalize the concept of `Weyl points born from $p_0$'.
          \item Every Weyl point $p$ of $H_t$ has a sign ($\pm 1$), defined as the sign of $\det ((\mathrm{d}h_t)_p)$. The sum of the signs for the Weyl points born from $p_0$ does not depend on the perturbation, it is an invariant of the degeneracy point $p_0$ of $H$. Moreover, it is equal to the local degree $\deg_{p_0} (h)$ of $h$ at $p_0$, and also with the Chern number of the lowest eigenstate (up to sign). This number is the topological charge of the non-generic degeneracy point $p_0$ of $H$.

          \item One can consider the global version, i.e., the transversality  to the whole degeneracy set $\Sigma$. However, $\Sigma$ is not a manifold. A map $H: M^3 \to \mbox{Herm}(n)$ defined on a 3-manifold $M^3$ is transverse to $\Sigma$ if it is transverse to every stratum of $\Sigma$. In other words, such map has only isolated transverse two-fold degeneracy points, but not necessarily ground-state. A map which is non-transverse to $\Sigma$ can have various degeneracy patterns, for example, multifold degeneracy points or non-isolated degeneracies. Every map can be perturbed to a transverse one by an arbitrary small perturbation, hence, any arbitrarily complicated degeneracy splits into transverse two-fold degeneracy points (Weyl points).

          \end{enumerate}
      \end{rem}

\section{Conclusions}

Energy degeneracies in parameter-dependent quantum systems are often accompanied by interesting phenomenology. 
In this work, we connected such energy degeneracies, and the effects of perturbations breaking those degeneracies, to the geometry of degeneracy submanifolds of the Euclidean space of Hermitian matrices. 
One link we have found is that the Schrieffer-Wolff transformation, a standard perturbative method to treat the eigenvalue problem of quasidegenerate matrices, is, in fact, a local chart of the space of Hamiltonians, which is aligned with the degeneracy submanifold.
We also established a distance theorem, relating the Frobenius distance of a matrix from a degeneracy submanifold to the energy splitting of its energy eigenvalues.
As a consequence, we have found that the order of energy splitting, caused by a one-parameter perturbation of a degenerate energy level, is the same as the order of distancing of the perturbed Hamiltonian from the degeneracy submanifold.
Finally, as applications of our results, we have proven the protected nature of Weyl points using the transversality theorem, and shown that geometrical information on the degeneracy submanifold can be obtained using results known in the domains of topological order and quantum error correction.
We anticipate that the connections established in our work  will further enhance the already-existing and fruitful cross-fertilisation between physics --- band-structure theory, quantum information, quantum materials
--- and differential geometry.

\section*{Author Contributions}

G.~P. and A.~P. initiated the project.
A.~P. managed the project.
G.~P. wrote the initial draft of the manuscript.
G.~P. and Gy.~F. derived the rigorous proofs with assistance from A.~P. and D.~V..
Gy.~F. prepared the figures.
All authors took part in writing the final version of the manuscript.

\acknowledgements

We thank J. Asb\'oth, S.~Diaz, P. Fromholz, A.~Hamma, I.~Lovas, P.~L\'evay, T.~Rakovszky, A.~S\'{a}ndor, P.~Vrana for useful discussions.
This research is supported by 
the Ministry of Culture and Innovation and the National Research, Development and Innovation Office within the Quantum Information National Laboratory of Hungary (Grant No. 2022-2.1.1-NL-2022-00004),
by the European Union within the Horizon Europe research and innovation programme via the ONCHIPS project under grant agreement No 101080022,
and by the HUN-REN Hungarian Research Network through the Supported Research Groups Programme, HUN-REN-BME-BCE Quantum Technology Research Group (TKCS-2024/34).
D.~V. was supported by the National Research, Development and Innovation Office of Hungary under OTKA grant no. FK 146499.
and the Deutsche Forschungsgemeinschaft (DFG, German Research Foundation) under Germany’s Excellence Strategy through the W\"{u}rzburg-Dresden Cluster of Excellence on Complexity and Topology in Quantum Matter – \textit{ct.qmat} (EXC 2147, project-id 392019).

\newpage
\appendix

\section*{Appendix: Further notes on SW decomposition}
\label{ss:Bravyi}

In Appendix~\ref{ss:coordfree} we formulate the coordinate free version of the SW decomposition Theorem~\ref{th:sw}, and we provide an important observation about the dependence on the base point $H_0$ in Appendix~\ref{ss:depbase}. In Appendix~\ref{app:directrot} we summarize the approach of \cite{BravyiSW}, which describes $e^{iS}$ as a `direct rotation'. We present another description as a parallel transport in Appendix~\ref{app:para}. Finally, we specify the validity of decomposition \eqref{eq:koo} in Appendix~\ref{app:range}.

\section{Coordinate free description}\label{ss:coordfree}
In \cite{BravyiSW}, $H_0$ is not assumed to be degenerate, neither diagonal. Dropping the  degeneracy of the base point does not cause an essential change in the SW decomposition \eqref{eq:koo}. Indeed, the effective Hamiltonian of $H$ with respect to a non-degenerate Hermitian matrix $G $ is equal to $G_{\Sigma}+H_{\textup{eff}}-G$, where $H_{\textup{eff}}$ is the effective Hamiltonian of $H$ with respect to $H_0:=G_{\Sigma}$. However, from our point of view the degeneracy of the base point is important, since we study the geometry around $\Sigma_k$. Hence we everywhere assume that $H_0$ has $k$-fold ground state degeneracy.

If $H_0 \in \Sigma_k$ is not diagonal, then Theorem~\ref{th:sw} can be formulated by replacing the block diagonality conditions with their coordinate free generalizations. Let $\mathcal{P}$ and $\mathcal{P}_0$ be the sum of the eigenspaces corresponding to the lowest $k$ eigenvalues of $H$ and $H_0$ respectively.
Their respective Hermitian orthogonal complements $\mathcal{P}^{\perp}$ and $\mathcal{P}_0^{\perp}$ are the sum of the eigenspaces corresponding to the $n-k$ highest eigenvalues.
Let $P$ and $P_0$ be the orthogonal projectors onto $\mathcal{P}$ and $\mathcal{P}_0$ respectively. Then $I-P$ and $I-P_0$ are the respective orthogonal projectors onto $\mathcal{P}^{\perp}$ and $\mathcal{P}_0^{\perp}$, where $I$ is the $n \times n$ unitary matrix. 

The general version of Theorem~\ref{th:sw} induces exactly the same decomposition as \eqref{eq:koo} except for replacing the conditions (1)--(4) with the following ones:
\begin{enumerate}
    \item[(1)*] $P_0^{\perp} BP_0^{\perp}=B$.
    \item[(2)*] $T=c P_0$ with a $c \in \R$.
    \item[(3)*] $P_0 H_{\textup{eff}} P_0=H_{\textup{eff}}$ and $\tr (H_{\textup{eff}})=0$.
    \item[(4)*]  $P_0 S P_0=0,\;P_0^{\perp} S P_0^{\perp}=0$. 
\end{enumerate}

Although it can be deduced from \cite[Section 3.1.]{BravyiSW}, another way to see it is the change to the eigenbasis of $H_0$ in $\C^n$. Let $H_0'=U_0^{-1} H_0 U_0$ be a diagonalization of $H_0$ with increasing order of the eigenvalues, and consider the SW decomposition \eqref{eq:koo} of $H':=U_0^{-1} H U_0$ with respect to $H_0'$, that is, 
\begin{equation}
    H'=e^{iS'} \cdot (H'_0+B'+T'+H'_{\textup{eff}}) \cdot e^{-iS'}.
\end{equation}
Then, $B:=U_0B'U_0^{-1}$, $T:=U_0 T U_0^{-1}$, $H_{\textup{eff}}:=U_0 H_{\textup{eff}} U_0^{-1} $ and $S:=U_0 S' U_0^{-1}$ satisfy (1)*--(4)*, and 
\begin{equation}
    H=e^{iS} \cdot (H_0+B+T+H_{\textup{eff}}) \cdot e^{-iS},
\end{equation}
providing a SW decomposition of $H$ with respect to $H_0$ in the general sense.

A SW chart around a non-diagonal $H_0 \in \Sigma_k$ is formed by the coordinates of the matrices $B'$, $T'$, $S'$ (these are the local ccordinates of $\Sigma_k$) and $H_{\textup{eff}}'$ (transverse coordinates), generalizing Corollary~\ref{co:locchart}. This chart obviously depends on the choice of $U_0$.

\section{Dependence on the base point $H_0$}\label{ss:depbase} 
First of all, if two base points $H_0, H_0' \in \Sigma_k$ can be diagonalized in the same basis with increasing order of the eigenvalues, then the SW decomposition with respect to $H_0$ and $H_0'$ are essentially the same: the difference of the eigenvalues appears as a shift in $B$, $T$, while $H_{\textup{eff}}$ and $S$ remain unchanged. Hence, in the interesting case the base points have different eigenbases.

In the same way as we did it in case of diagonal $H_0$, one can define in general (for non-diagonal $H_0$) 
\begin{equation}
    H_{\textup{proj}}:=e^{iS} \cdot (H_0+B+T) \cdot e^{-iS}
\end{equation}
and prove that $H_{\textup{proj}}$ is equal to $H_{\Sigma}$, that is, the closest point of $\Sigma_k$ to $H$ constructed independently of the SW decomposition, see \eqref{eq:hsigma} and Section~\ref{ss:dist}. Hence, $H_{\textup{proj}}$ does not depend on the choice of the base point.

In contrast, the effective Hamiltonian deeply depends on the base point $H_0$.  If two base points $H_0$ and $H'_0$ cannot be simultaneously diagonalized, then the corresponding effective Hamiltonians $H_{\textup{eff}}$ and $H'_{\textup{eff}}$ of $H$ are obviously not equal in general, since (3)* in Section~\ref{ss:coordfree} imposes a different condition on them. 
Moreover, two matrices $H$ and $G$ having the same effective Hamiltonian $H_{\textup{eff}}=G_{\textup{eff}}$ with respect to $H_0$ can have different effective Hamiltonian $H'_{\textup{eff}} \neq G'_{\textup{eff}}$ with respect to $H_0'$, in general. However, their spectra and norms are the same, $ \| H'_{\textup{eff}} \| = \| G'_{\textup{eff}} \|$. To see it, observe that
\begin{equation}
    H-H_{\Sigma}=e^{iS(H)} \cdot H_{\textup{eff}} \cdot e^{-iS(H)}=
    e^{iS'(H)} \cdot H'_{\textup{eff}} \cdot e^{-iS'(H)},
\end{equation}
\begin{equation}
    G-G_{\Sigma}=e^{iS(G)} \cdot G_{\textup{eff}} \cdot e^{-iS(G)}=
    e^{iS'(G)} \cdot G'_{\textup{eff}} \cdot e^{-iS'(G)},
\end{equation}
where $S(H)$ and $S(G)$ ($S'(H)$ and $S'(G)$, respectively) denote the exponents coming from the SW decomposition of $H$ and $G$ with respect to $H_0$ ($H'_0$), see Section~\ref{ss:coordfree}. Then,  $H_{\textup{eff}}=G_{\textup{eff}}$ implies that
\begin{equation}\label{eq:Geff}
G'_{\textup{eff}}=e^{-iS'(G)} \cdot e^{iS(G)} \cdot e^{-iS(H)} \cdot e^{iS'(H)} \cdot H'_{\textup{eff}} 
\cdot e^{-iS'(H)} \cdot e^{iS(H)} \cdot e^{-iS(G)} \cdot e^{i S'(G)},
\end{equation}
which does not imply $H'_{\textup{eff}}= G'_{\textup{eff}}$ in general, however, it implies that the norms and spectra of $H'_{\textup{eff}}$ and $ G'_{\textup{eff}}$ are equal. Furthermore, if  $H_0$ and $H_0'$ can be diagonalized in the same basis (with increasing order of the eigenvalues), then $H_{\textup{eff}}= G_{\textup{eff}}$ implies $H'_{\textup{eff}}= G'_{\textup{eff}}$, since in this case $e^{iS(H)}=e^{iS'(H)}$ and $e^{iS(G)}=e^{iS'(G)}$ in Equation \eqref{eq:Geff}. Indeed, in this case $H_{\textup{eff}}= G_{\textup{eff}}=H'_{\textup{eff}}= G'_{\textup{eff}}$ also holds.

\section{Direct rotation}\label{app:directrot}

Here we recall the characterization of $e^{iS}$ as a direct rotation.
We follow Ref.~\cite{BravyiSW}. 
About the role of $S$, recall Remark~\ref{re:Sgenerates}: in first order, $S$ generates the off-block elements of $H$.

To define direct rotation, recall that $\mathcal{P}$ and $\mathcal{P}_0$ are the sum of the eigenspaces corresponding to the lowest $k$ eigenvalues of $H$ and $H_0$ respectively, and $e^{iS}$ is a unitary rotation between $\mathcal{P}_0$ and $\mathcal{P}$, specified as follows. 
$I-2 P_0$, and $I-2 P$ are the reflections with respect to $\mathcal{P}_0^{\perp}$ and $\mathcal{P}^{\perp}$, respectively. Assuming that $\sqrt{(I-2 P)(I-2P_0)}$ is defined (we assume that the matrix under the square root has no negative real eigenvalues, and we use the principal branch of the square root with a branch cut at the negative real axis), it is a unitary matrix called the direct rotation between the subspaces $\mathcal{P}_{0}$ and $\mathcal{P}$, see \cite{davis1969some, BravyiSW}.  By \cite[Section 2.2.]{BravyiSW}, $\sqrt{(I-2 P)(I-2P_0)}=e^{iS}$ holds. 
It follows that $PSP=P^{\perp} S P^{\perp}=0$ also holds for $S$, that is, $S$ is off-block with respect to $H_0$, and also with respect to $H$. Moreover, if $H_0$ is diagonal, then the direct rotation $e^{iS}$ has a special form,
\begin{equation}
e^{iS}=\begin{pmatrix}
    U_{1,1} & U_{1,2} \\
    -U_{1,2}^{\dagger} & U_{2,2} \\
\end{pmatrix}   , 
\end{equation}
where
\begin{eqnarray}
U_{1,1}&=&\sqrt{I_{k \times k}-U_{1,2}U_{1,2}^\dagger},\\
U_{2,2}&=&\sqrt{I_{(n-k) \times (n-k)}-U_{1,2}^\dagger U_{1,2}},
\end{eqnarray}
are Hermitian positive definite matrices, see \cite[Section 2.3.]{BravyiSW}.

\section{Parallel transport}\label{app:para}
Alternatively,  $e^{iS}$ can be described as a parallel transport in two conceptually different ways. We thank the anonymous reviewer for suggesting this approach.

 For simplicity, assume that $H_0$ is diagonal, hence, $S$ is off-block. $S$ defines a path (one-parameter subgroup) $U(t)=e^{itS}$ ($t \in [0;1] $) in the unitary group $U(n)$. Let $u_1(t), \dots, u_n(t)$ denote the columns of $U(t)$, and let $\mathcal{P}(t) \subset \C^n$ denote the subspace spanned by the vectors $u_1(t), \dots, u_k (t)$. Obviously, $\mathcal{P}(0)= \mathcal{P}_0$, $\mathcal{P}(1)=\mathcal{P}$, and $\mathcal{P}(t)^{\perp}$ is the subspace of $\C^n$ spanned by $u_{k+1}(t), \dots, u_n(t)$. Observe that 
\begin{equation}\label{eq:parallel}
    U(t)^{\dagger} \cdot  
    \dot{U}(t) =iS
\end{equation}
holds for all $t$. Since $S$ is off-block, this means that  for  $1 \leq i, j \leq k$  (or $ k+1 \leq i, j \leq n $)
\begin{equation}
    u_i(t)^{\dagger} \cdot \dot{u}_j(t)=0  
\end{equation} 
holds, that is, $\dot{u}_j(t)$ is orthogonal to $\mathcal{P}(t)$ for $j=1, \dots, k$.

The first way to interpret this as a parallel transport is to consider $\mathcal{P}(t)$ as a  complex vector bundle of rank $k$ over $[0;1]$. The orthogonality of $\dot{u}_j(t)$  to $\mathcal{P}(t)$ exactly defines the parallel transport with respect to the Berry connection, see \cite{Asboth,BohmBook}. In other words, the direct rotation $e^{iS}$ parallel transports a unitary basis of $\mathcal{P}_0$ to a unitary basis of $\mathcal{P}$ over the path $\mathcal{P}(t)$ of subspaces with respect to the Berry connection. 

For the second formulation of the same observation, let $\mathrm{Gr}_k(\C^n)$ denote the Grassmannian manifold consisting of the $k$ dimensional complex subspaces of $\C^n$. Consider the map $ \pi: U(n) \to \mathrm{Gr}_k(\C^n)$ which associates to each unitary matrix $U \in U(n)$ the subspace $\pi(U)=\mathcal{T} \in \mathrm{Gr}_k(\C^n)$ spanned by the first $k$ columns of $U$. This map $\pi$ defines a locally trivial fiber bundle over the base space $\mathrm{Gr}_k(\C^n)$ with total space $U(n)$. The fiber $\pi^{-1}(\mathcal{T})$ over $\pi(U)=\mathcal{T} $ consists of the matrices whose first $k$ columns span the subspace $\mathcal{T}$. 
These are exactly the matrices in form $U \cdot V$, where $V \in U(k) \times U(n-k)$ (i.e., $V$ is block diagonal), hence, $\pi^{-1}(\mathcal{T})$  is diffeomorphic to $U(k) \times U(n-k)$. Moreover, the fiber-wise right action of the group $U(k) \times U(n-k)$ makes $\pi: U(n) \to \mathrm{Gr}_k(\C^n)$ a $U(k) \times U(n-k)$-principal bundle, see e.g.  \cite[chap. 3]{dupont2003fibre} or \cite{husemoller1994fibre, BohmBook} for details.

Recall that the Lie algebra of $U(n)$ -- defined as the tangent space at the identity -- is the space of anti-hermitian matrices, that is, $T_I U(n)=\mathfrak{u}(n) =i \cdot \text{Herm}(n)$. Similarly to the Frobenius inner product in $\text{Herm}(n)$, $\mathfrak{u}(n)$ admits an inner product, which determines a left-right invariant Riemannian metric on $U(n)$.
This Riemannian metric induces an Ehresmann connection (that is, a subbundle of `horizontal subspaces' of the tangent bundle $TU(n)$) as follows. The Lie algebra $\mathfrak{u}(n)$ splits into the orthogonal subspaces $\mathfrak{u}(k) \oplus \mathfrak{u}(n-k) \oplus \mathfrak{h}$, where $\mathfrak{u}(k) \oplus \mathfrak{u}(n-k)$ is the Lie algebra of the subgroup $U(k) \times U(n-k) \subset U(n)$ and it consists of the block diagonal anti-hermitian matrices, and $\mathfrak{h} \subset \mathfrak{u}(n)$ is the subspace consisting of off-block anti-hermitian matrices. This decomposition can be copied to each tangent space $T_U U(n)$ via the tangent map $T_I L_U: T_I U(n) \to T_U U(n)$ of the left multiplication $L_U: U(n) \to U(n)$, defined by $L_U(V)=U \cdot V$ (for $U, V \in U(n)$). Since $L_U$ is the restriction of the left multiplication by $U$ on $\C^{n \times n}$, its tangent map is also the left multiplication by $U$, that is, $T_IL_U(A)=U \cdot A$ for all $A \in T_I U(n)=\mathfrak{u}(n)$. 

We obtain an orthogonal decomposition for each tangent space of $U(n)$ in form
\begin{equation}
    T_U U(n)=U \cdot \mathfrak{u}(k) \oplus U \cdot \mathfrak{u}(n-k) \oplus U \cdot \mathfrak{h} .
\end{equation}
Here $U \cdot \mathfrak{u}(k) \oplus U \cdot \mathfrak{u}(n-k)$ is the tangent space of the fiber $\pi^{-1}(\pi(U))$ at $U$, and $\mathfrak{h}_U:=U \cdot \mathfrak{h}$ is its orthogonal complement. The subspaces $ \mathfrak{h}_U \subset T_U U(n)$ form an Ehresmann connection \cite[chap. 6]{dupont2003fibre}, \cite{BohmBook}. Intuitively, $ \mathfrak{h}_U$ is  a `horizontal subspace', if the fibers are imagined as `vertical'. 

The connection $ \mathfrak{h}_U$ is  invariant under the fiber-wise right action of the group  $U(k) \times U(n-k)$ on $U(n)$, that is, $\mathfrak{h}_{U \cdot V}=\mathfrak{h}_U \cdot V$ holds for $U \in U(n)$, $V\in U(k) \times U(n-k)$. Indeed, \begin{equation}
    \mathfrak{h}_{U} \cdot V=  U \cdot \mathfrak{h}  \cdot V= U \cdot V \cdot \mathfrak{h} \cdot V^{-1}  \cdot V=U \cdot V \cdot \mathfrak{h}=\mathfrak{h}_{U \cdot V},
\end{equation}
using that the off-block antihermitian matrices are closed under the conjugate action of $U(k) \times U(n-k)$: $V \cdot \mathfrak{h} \cdot V^{-1}=\mathfrak{h}$ for $V\in U(k) \times U(n-k)$. 

By Equation~\eqref{eq:parallel}, $\dot{U}(t)=U(t) \cdot iS $, hence, $\dot{U}(t) \in U(t) \cdot \mathfrak{h}=\mathfrak{h}_{U(t)}$ for every $t$. That is, $U(t)$ is a parallel transport (horizontal lift) over the path $\mathcal{P}(t)$ in $\mathrm{Gr}_k(\C^n)$. 

Importantly, the parallel transport itself does not characterize the rotation $e^{iS}$. Instead, if the path $\mathcal{P}(t)$ of subspaces is given, then the parallel transport over this path provides $u(t)=e^{itS}$, in particular, $u(1)=e^{iS}$. Choosing another path in $\mathrm{Gr}_k(\C^n)$ joining $\mathcal{P}_0$ and $\mathcal{P}$, the parallel transport over this path possibly provides a different rotation matrix between the two subspaces. 

\section{The range of validity of the SW decomposition}\label{app:range} 
The range of validity can be discussed in different senses, e.g. (1) Given an $H_0$, which matrices $H$ have a SW decomposition with respect to $H_0$?
(2) How can we choose a neighborhood $\mathcal{V}_0$ of $H_0$ in $\mbox{Herm}(n)$ and a neighborhood $\mathcal{U}_0$ of $(0,0,0,0)$ in the $(S, B, T, H_{\textup{eff}})$ space such that every $H \in \mathcal{V}_0$ has a unique SW decomposition with $(S, B, T, H_{\textup{eff}}) \in \mathcal{U}_0$? Which are the largest possible choices? (3) What is the radius of convergence of the Taylor series expressing $H_{\textup{eff}}$ or $S$ in terms of $H$? These problems lead to the  analysis of the analytic map 
\begin{equation}\label{eq:f}
    f: (S, B, T, H_{\textup{eff}}) \mapsto H=e^{iS} \cdot (H_0+B+T+S+H_{\textup{eff}}) \cdot e^{-iS},
\end{equation}
determining its image and injective restrictions. The configurations $(S, B, T, H_{\textup{eff}})$ where the Jacobian of the map has maximal rank have a neighborhood on which the map is invertible, i.e., there is a unique decomposition. An additional question is the `lowest $k$ state property': for given neighborhoods $\mathcal{U}_0$ and $\mathcal{V}_0$ such that $f: \mathcal{U}_0 \to \mathcal{V}_0$ is a bijection, is it true that the first $k$ columns of $e^{iS}$ span the sum of the eigenspaces of $H$ corresponding to the lowest $k$ eigenvalues? Cf. the proof Theorem~\ref{th:sw} in Section~\ref{ss.prsw}.

This article is not intended to provide complete answers to these questions, but we list some basic observations about them, and summarize the corresponding results of \cite{BravyiSW}. Throughout the observations below, $f$ denotes the map defined by Equation~\eqref{eq:f} with a fixed $H_0 \in \Sigma_k$ specified in each case.
\begin{itemize}
\item Question (1) is trivial: for every $H_0$ and $H$ there is a SW decomposition of $H$ with respect to $H_0$. It can be constructed by taking a direct rotation between the eigenspace of $H_0$ corresponding to the lowest $k$ eigenvalues to any $k$-dimensional eigenspace of $H$. This shows that $f$ is a surjective map, and also shows that without any additional requirement (formulated by the other questions) the SW decomposition itself is irrelevant.
\item Without any restriction, the exponent $S$ is not unique in the SW decomposition \eqref{eq:koo},   since the exponential map $S \mapsto e^{iS}$ is not injective. 
    \item If we fix $S$, the other terms $B$, $T$ and $H_{\textup{eff}}$ can be modified as long as the resulting matrix $H$ does not reach $\Sigma_{[k, k+1]}$, that is, until $\lambda_k \neq \lambda_{k+1}$. Furthermore, different configurations of $B$, $T$ and $H_{\textup{eff}}$ provide different matrices $H$, hence, the uniqueness of the SW decomposition cannot break down in this way.
    \item For a fixed $S$ and varying $B, T, H_{\textup{eff}}$, possibly $H$ crosses $\Sigma_{[k, k+1]} $. Then, two eigenvalues of the two blocks of $\widetilde{B}$ are exchanged, hence, property (3) of Theorem~\ref{th:sw} breaks down. A more careful analysis shows that the unique decomposition property also fails when $H$ reaches $\Sigma_{[k, k+1]}$, since a continuous family of exponents $S$ provides the same $H$. Indeed, the eigenspace $\mathcal{P}$ corresponding to the lowest $k$ eigenvalues has a freedom if $\lambda_{k}=\lambda_{k+1}$. Moreover, for the same reason, the Jacobian of the map $f$ has determinant 0 if $H \in \Sigma_{[k, k+1]}$. Cf. the proof of Theorem~\ref{th:sw} in Section~\ref{ss.prsw}, where the maximal rank of the Jacobian is crucial for the local uniqueness of the SW decomposition.

  \item Given $G \notin \Sigma_{[k, k+1]}$, choose $H_0:=G_{\Sigma}$. Then, $G$ obviously has a SW decomposition with respect to $H_0$ with $T=B=S=0$ and $G_{\textup{eff}}=G-H_0$. Moreover, there is a neighborhood of $G $ in $\mbox{Herm}(n)$ and $(0,0,0, G_{\textup{eff}})$ in the $(T, B, S, H_{\textup{eff}})$ space, such that $f$ is an analytic diffeomorphism between these neighborhoods. It can be seen as follows.
  Without losing generality we can assume that $G$ and $H_0=G_{\Sigma}$ are diagonal. The proof of Theorem~\ref{th:sw} can be modified to show that the Jacobian of $f$ at $(0,0,0, G_{\textup{eff}})$ has maximal rank, using that the eigenvalues of $G$ satisfies $\lambda_k < \lambda_{k+1}$.

    \item Choose $G \notin \Sigma_{[k, k+1]}$ as above, and choose $H_0=G_{\Sigma}$. The elements of the line segment joining $G$ and $H_0$ also have a SW decomposition with $T=B=S=0$, only $H_{\textup{eff}}$ varies. We can show that there is a neighborhood of the line segment joining $G$ and $H_0$ in $\mbox{Herm}(n)$ whose elements admit a unique SW decomposition, with a suitable restriction for its terms. This neighborhood can be imagined as a `capsule' around the line segment, cf. Figure~\ref{fig:capsule}. The existence of such capsule can be deduced as follows.
    
    Consider the line segment $L$ joining $(0,0,0, H-H_0)$ and $(0,0,0,0)$ in the $(S,B,T, H_{\textup{eff}})$ space. $L\subset \R^{n^2}$ is a compact submanifold and $f|_L$ is injective on $L$, furthermore, the Jacobian of $f$ has maximal rank at the points of $L$. By a generalization of the inverse function theorem, see \cite[Pg. 19., Exercise 10; Pg. 56., Exercise 14]{guillemin2010differential}, $f$ is a diffeomorphism between a neighborhood $\mathcal{L}_0$ of $L$ and its image. The image of this neighborhood  $f(\mathcal{L}_0)=\mathcal{V}_0 \subset \mbox{Herm}(n)$ admits a unique SW decomposition with $(S,B,T, H_{\textup{eff}}) \in \mathcal{L}_0$.

    \begin{figure}
	\begin{center}		\includegraphics[width=0.4\columnwidth]{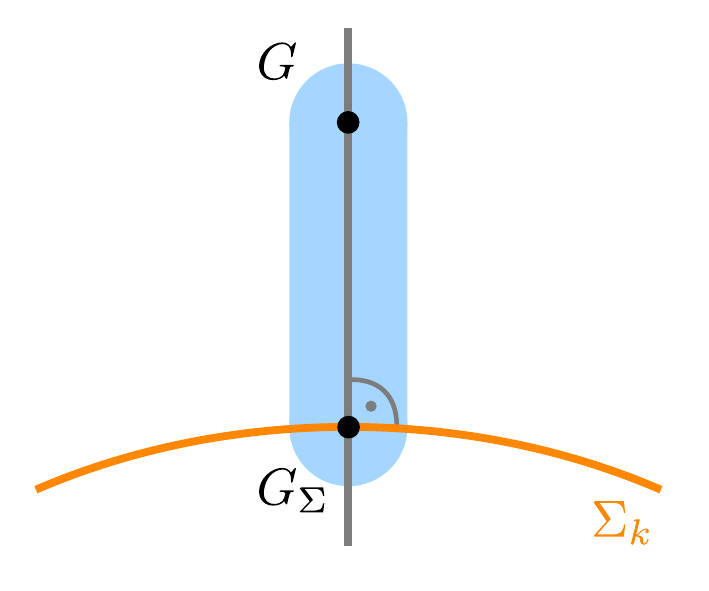}
	\end{center}
	\caption{`Capsule': a neighborhood of the line segment joining $G \notin \Sigma_{[k, k+1]}$ and $H_0=G_{\Sigma}$ in $\mbox{Herm}(n)$ whose elements admit a unique SW decomposition. \label{fig:capsule}}
    \end{figure}
    
    Furthermore, we can define a `discus' (of dimension $n^2$) extending the capsule, with unique SW decomposition. We take a disc $\mathcal{B}$ of dimension $k^2-1$ consisting of matrices $H$ with  $H_{\Sigma}=H_0$ and $d(H,H_0)\leq d(G, H_0)$. The disc $\mathcal{B}$ is orthogonal to $\Sigma_k$ at $H_0$.  Then, there is a neighborhood $\mathcal{V}_0$ of $\mathcal{B}$ with a unique SW decompostition, resulting a discus shaped volume in $\mathrm{Herm}(n)$.
\end{itemize}

We can give other conditions for the range of validity of the SW decomposition \eqref{eq:koo} based on \cite{BravyiSW}. Let $r_0$ denote half of the spectral gap of $H_0$, i.e.
\begin{equation}
    r_0=\frac{\lambda_{k+1}-\lambda_k}{2},
\end{equation}
where $\lambda_i$ denote the eigenvalues of $H_0$. Recall that $\| \cdot \|_2$ denotes the operator 2-norm, cf. Remark~\ref{re:op2norm}. Then, based on \cite{BravyiSW}, we can deduce the following: The matrices $H$ satisfying 
$ \| H- H_0 \|_2 < r_0 $
have a unique SW decomposition with $\| S \|_2 < \pi /2$. Since $\|H-H_0\|_2 \leq \| H-H_0 \|$, the more strict condition $ \| H- H_0 \| < r_0 $
also implies a unique SW decomposition  with $\| S \|_2 < \pi /2$. Indeed, by \cite[Lemma 3.1.]{BravyiSW}, $\| H- H_0 \|_2 < r_0$ implies that $\| P-P_0 \|_2 < 1$, and in this case there is a unique $S$ with $\| S \|_2 < \pi/2$ such that $e^{iS}$ is the direct rotation  between $\mathcal{P}$ and $\mathcal{P}_0$, see \cite[Corollary 2.2., Lemma 2.3.]{BravyiSW}. 

Next we show that geometrically $r_0$ is the radius of the  largest   open ball around $H_0$ in operator 2-norm which does not intersect $\Sigma_{[k, k+1]}$. First, take a closest element $G_0$ of $\Sigma_{[k, k+1]}$ to $H_0$ in the Frobenius norm. Then $r_0$ is equal to the distance of $H_0$ and $G_0$ in the operator 2-norm. Indeed, in a similar way we constructed $H_{\Sigma}$ and we proved in Section~\ref{ss:split} that  it is the closest element of $\Sigma_k$ to $H$, one can construct the closest element of $\Sigma_{[k, k+1]}$ to $H_0$ by collapsing $\lambda_k$ and $\lambda_{k+1}$ to their mean value $(\lambda_k+\lambda_{k+1})/2$ (and leave the other eigenvalues unchanged). Because the $k$-th eigenvalue of $H_0$ is degenerate, the closest element $G_0$ of $\Sigma_{[k, k+1]}$ is not unique: the construction depends on the choice of the eigenspace corresponding to $\lambda_k$, which can be any one dimensional subspace of $\mathcal{P}_0$. The distance of $H_0$ and $G_0$ in the operator 2-norm is obviously $r_0$. 

Moreover, $G_0$ is a closest point of $\Sigma_{[k,k+1]}$ to $H_0$ in the operator 2-norm as well. To show this, assume that $H_0+K \in \Sigma_{[k,k+1]}$ for a Hermitian  matrix $K$. Let $\mu$ denote the $k$-th eigenvalue of $H_0+K$, which is equal to the $k+1$-th eigenvalue of $H_0+K$.  One can deduce from Weyl's inequality~\cite{Weyl1912} that $|\mu-\lambda_k|\leq \|K\|_2$ and also $|\mu-\lambda_{k+1}|\leq \|K\|_2$. But at least one of the left sides is at least $r_0$, i.e., $r_0 \leq |\mu-\lambda_k|$ or $r_0 \leq |\mu-\lambda_{k+1}|$ holds, implying that $r_0 \leq \|K\|_2$.

Summarizing the above arguments,  
\begin{equation}
    \mathcal{V}_0=\{H \in \mbox{Herm}(n) \ | \ \|H-H_0 \|_2 < r_0 \}
\end{equation}
is the largest open ball around $H_0$ in operator 2-norm which does not intersect $\Sigma_{[k, k+1]}$. There is a unique SW decomposition on $\mathcal{V}_0$ if we restrict $(S, B, T, H_{\textup{eff}})$ such that $\|S\|_2 <\pi/2$. The decomposition automatically satisfies the lowest $k$ state property, since Ref.~\cite{BravyiSW} uses the direct rotation between the eigenspaces of $H_0$ and $H$ corresponding to the lowest $k$ eigenvalues.

We do not know, whether the Jacobian of $f$ has maximal rank in the domains indicated by \cite{BravyiSW}. There is a weaker result in \cite{BravyiSW} about analytic one parameter families $H(t)=H_0 + tH_1$: The radius of convergence of the power series expressing $H_{\textup{eff}}(t)$ and $S(t)$ in terms of $H(t)$ around $H_0$ is -- surprisingly, only -- at least $r_0/4$.  

\printbibliography

\end{document}